\documentclass[10pt,twocolumn]{article}
\usepackage[margin=0.72in,top=0.70in,bottom=0.78in]{geometry}
\usepackage{amsmath,amssymb,amsthm}
\usepackage{booktabs}
\usepackage{array}
\usepackage{graphicx}
\usepackage{url}
\usepackage{microtype}
\usepackage{xcolor}
\usepackage[numbers,square,sort&compress]{natbib}
\usepackage[hidelinks]{hyperref}
\setcitestyle{numbers,square,semicolon}
\graphicspath{{./}}
\newcommand{\TP}{\mathrm{TP}}
\newcommand{\Rhat}{\widehat R}
\newcommand{\shat}{\widehat\sigma}
\newcommand{\calT}{\mathcal T}
\newcommand{\calC}{\mathcal C}
\newcommand{\calE}{\mathcal E}
\newcommand{\calP}{\mathcal P}
\newcommand{\calU}{\mathcal U}
\theoremstyle{plain}
\newtheorem{proposition}{Proposition}
\newtheorem{lemma}{Lemma}
\newtheorem{corollary}{Corollary}
\theoremstyle{remark}
\newtheorem{remark}{Remark}
\theoremstyle{definition}
\newtheorem{definition}{Definition}
\makeatletter
\newcommand{\fullwidthfrontmatter}[1]{%
  \twocolumn[{\begin{@twocolumnfalse}#1\end{@twocolumnfalse}}]%
}
\makeatother
\begin{document}
\fullwidthfrontmatter{
\begin{center}
{\LARGE\bfseries GoCo: structure-guided conformal calibration for error-controlled Gene Ontology annotation release\par}
\vspace{0.8em}
{\large Chong Wang$^{1,2}$, Yongzhao Shao$^3$, and Peng Liu$^{1,4,*}$\par}
\vspace{0.6em}
{\small
$^1$Department of Statistics, Iowa State University, Ames, Iowa, USA\\
$^2$Department of Veterinary Diagnostic and Production Animal Medicine, College of Veterinary Medicine, Iowa State University, Ames, Iowa, USA\\
$^3$Division of Biostatistics, Department of Population Health, NYU Grossman School of Medicine, New York, New York, USA\\
$^4$Laurence H. Baker Center for Bioinformatics and Biological Statistics, Iowa State University, Ames, Iowa, USA\\
$^*$Corresponding author: Peng Liu, Department of Statistics, Iowa State University, 2117 Snedecor Hall, 2438 Osborn Drive, Ames, IA 50011, USA. Tel: +1 515 294 7806. E-mail: \texttt{pliu@iastate.edu}\par}
\end{center}
\vspace{0.4em}
\begin{abstract}
Computational gene-function prediction generates annotations at a scale that experimental characterization cannot match, yet evaluations using subsequently acquired experimental evidence show that substantial errors remain even among leading methods. Explicit error-controlled calibration for Gene Ontology (GO) prediction remains largely undeveloped. GO prediction is both multilabel and hierarchical, and many predictors incorporate gene--gene relationships such as co-essentiality and network proximity to improve prediction. We introduce GoCo, a structure-guided framework that calibrates how scores from an existing predictor are converted into released annotations. GoCo certifies a global threshold for a GO-hierarchy-aware gene-level false discovery proportion (FDP), then constructs a nested release path under the same asymptotic risk-control guarantee. At the primary benchmark operating point ($\alpha=0.10$, $\delta=0.50$), GoCo-M, developed on Wainberg co-essentiality modules and applied without retuning to three independent studies with the same shared-evidence structure, kept mean held-out gene-level FDP at or below target while releasing 6.3--17.6\% more reference-supported GO terms than Multilabel. For gene-specific neighbourhoods rather than shared module evidence, we developed GoCo-N, where scores are effectively untied and ordering acts on individual calls. Developed on FunMap and validated without retuning on the STRING v12.0 experiments-only network, GoCo-N released 33.8\% and 36.4\% more reference-supported GO terms than Multilabel on the same refined score grid in FunMap and STRING, respectively, while keeping mean held-out FDP below target. The same pattern held under the stricter $\delta=0.10$ setting. Thus, predictor-native biological structure can be reused at calibration---not only prediction---to release more supported annotations at a chosen gene-level error target without retraining the predictor.
\end{abstract}
\vspace{1.0em}
}

\section{Introduction}
Computational gene- and protein-function prediction is a rapidly developing area of bioinformatics. Community assessments such as CAFA have tracked hundreds of methods over more than a decade \cite{radivojac2013,zhou2019,jiang2016cafa2,lin2024review}, and computational annotations now outnumber experimentally supported ones by orders of magnitude \cite{go2023,skunca2012quality}. Most methodological development has focused on improving or benchmarking prediction scores rather than controlling error in released annotations. Consequently, predicted annotations enter databases, enrichment analyses and mechanistic interpretation without an accompanying error statement. Benchmark-level $F_{\max}$ (the maximum F-measure across thresholds) and area under the precision--recall curve (AUPR) summarize predictive performance rather than the error of a released annotation set \cite{clark2013information,piovesan2024evaluator,kulmanov2020deepgoplus,yao2021netgo2}, and annotation sets built at $F_{\max}$-optimal thresholds can violate GO's consistency constraints, which rule out mutually exclusive or taxonomically contradictory functions \cite{tawfiq2026completeness}. For downstream users, the practical question is therefore not only how well functions are ranked, but how to release useful annotations while limiting unsupported ones.

Our approach is motivated by two forms of structure in GO prediction that are directly relevant to error-controlled release. The first is the GO hierarchy: GO is a directed acyclic graph (DAG) with multiple inheritance, and under the TruePath rule an annotation to a term implies annotations to its retained ancestors \cite{go2023}. The second is predictor-native gene--gene structure: predictors use co-regulation, co-essentiality, physical interaction, sequence similarity and integrated networks to improve function prediction \cite{zhang2024search,zhu2025mkfgo,hu2025dahgl,kan2026protdml,wang2026ssas,mostafavi2008,kustatscher2019,wainberg2021}. A genome-wide score cutoff does not by itself exploit either structure: it treats the ontology as a flat label list and cannot use predictor-native evidence to retain more supported annotations at the same error target.

Several distribution-free calibration methods can add an error statement to released annotations without retraining the predictor. Split conformal prediction guarantees marginal coverage from an exchangeable calibration sample \cite{angelopoulos2023gentle}. Risk-controlling prediction sets and conformal risk control extend this idea to bounded losses that are monotone in a release threshold \cite{bates2021,angelopoulos2024crc}. Learn-Then-Test (LTT) further allows nonmonotone candidate sequences, returning with probability at least $1-\delta$ only candidate policies whose risk (expected loss) is at most $\alpha$ \cite{angelopoulos2025ltt}. These approaches have been applied in molecular biology research \cite{boger2025,ding2024cpec,lopezdecastro2025conformal,corbetta2026cellontology}. Hierarchical conformal methods also construct prediction sets on trees or DAGs when each unit has one true output \cite{mortier2026hierarchical,zhang2025structured}. Boger et al.\ \cite{boger2025} proposed a hierarchy-aware formulation for tree-structured labels, but the multiple-parent structure of GO is not a tree and therefore requires different treatment. For GO release, the formulation must accommodate two features: a gene can have many true terms, and GO is a DAG with multiple inheritance. We define the release error as the fraction of terms released for a gene that are unsupported by the reference. Our first layer, GoDag, tests whether the average gene-level FDP for a global threshold is within a chosen limit, with a stated bound on the probability of false certification.

Existing release rules also leave two practical gaps. First, accuracy-optimized or fixed-size rules such as an $F_{\max}$ threshold or top-$k$ terms per gene state no gene-level error rate. Second, calibrated methods restricted to a single global threshold can leave supported predictions unreleased when the score path is coarse. Conservative concentration-bound certificates can compound this loss of yield (Results). The balance between avoiding unsupported calls and retaining supported functions depends on the application: some users may place greater weight on avoiding unsupported annotations, whereas exploratory applications may place greater value on retaining potentially informative functions. Accordingly, we evaluate release rules along two complementary dimensions: gene-level FDP for unsupported releases and reference-supported yield, a complementary measure of how many supported functions are retained rather than missed, not a second error rate. The GoCo framework addresses error control and yield in sequence: GoDag supplies the gene-level error certificate, and the second layer enriches the candidate path while retaining that certification, with the appropriate refinement determined by the predictor's evidence structure. For shared evidence, the limitation is especially clear: genes sharing module evidence can receive the same score for a term, so thousands of gene--term pairs tie exactly and no threshold can split such a block---it admits all of it, possibly above target, or none, leaving the budget unspent (Results). But the predictor has already exposed which genes in the block share which evidence, and GoCo turns that into a nested family of release policies certified by one test.

Our contributions are fourfold. First, we define a hierarchy-aware gene-level FDP estimand and its global-threshold certificate, GoDag. Second, we construct candidate release paths for predictors that can identify the evidence behind each call; Proposition~\ref{prop:validity} and Supplementary Corollary~\ref{supp-cor:order} establish asymptotic finite-population fixed-sequence control for any ordering learned from the ranking fold. Third, Supplementary Lemma~\ref{supp-lem:power} decomposes the gain into granularity and ordering terms, clarifying how the candidate path should match the predictor's evidence structure. Fourth, we develop two structure-specific implementations: GoCo-M for evidence shared between genes, developed on Wainberg and evaluated without retuning on Sanger Project Score, Project DRIVE and the HAP1 genetic-interaction map; and GoCo-N for per-gene neighbourhood evidence, developed on FunMap and validated without retuning on STRING. We also provide a gene-level case study of the released annotations. Code, principal analysis inputs and split-level results for the primary analyses are public, together with aggregate summaries for the secondary comparator panel.

\section{Materials and methods}
\subsection{Step 1: define the annotation release problem}
\label{sec:estimand}
GoCo leaves the upstream predictor unchanged. It requires (i) a score $s_{ig}$ for each gene $i$ and GO term $g$ and (ii) the evidence behind each candidate call, such as annotated module members or network neighbours. A call is a proposed gene--GO annotation. A release policy $\pi$ specifies which calls to retain, giving each gene a set $S_i(\pi)$. A global threshold uses $S_i(\lambda)=\{g:s_{ig}\ge\lambda\}$; GoCo also considers more selective policies (Figure~\ref{fig:workflow}).

\begin{figure*}[!t]
\centering
\includegraphics[width=0.99\textwidth]{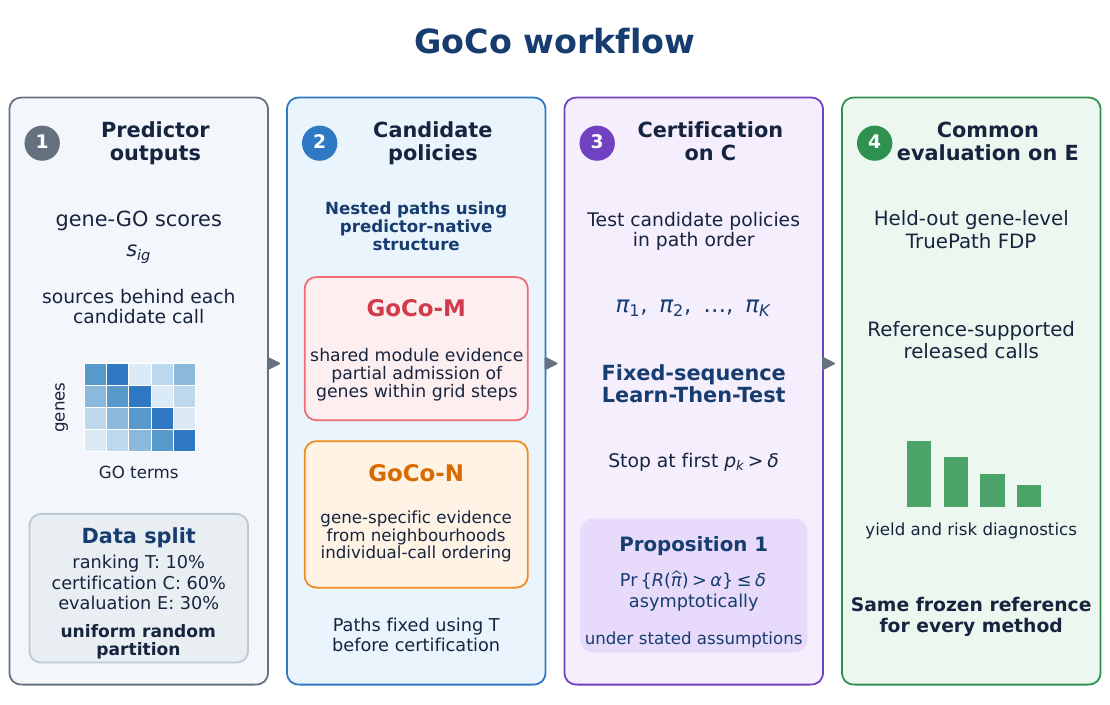}
\caption{Overview of GoCo and its two structural instantiations. Module analyses use 70\% calibration for global methods and 60\% certification for GoCo-M; neighbourhood analyses certify all methods on the same 60\% fold.}
\par\small\textit{Alt text:} Workflow of GoCo. Predictor scores and evidence define two nested candidate-policy families: GoCo-M for shared module evidence with tied scores and GoCo-N for gene-specific neighbourhood evidence. Candidate policies are certified in fixed sequence on the certification fold, and held-out TruePath FDP and reference-supported calls are evaluated on the evaluation fold.
\label{fig:workflow}
\end{figure*}

\label{sec:lossfamily}
The Direct reference contains annotations recorded before adding their GO ancestors. TruePath closure adds the retained ancestors to give the acceptable set $A_i^{\TP}$. The gene-level false discovery proportion is
\begin{equation}
L_i(S_i)=
\begin{cases}
1-\dfrac{|S_i\cap A_i^{\TP}|}{|S_i|}, & |S_i|>0,\\[4pt]
0, & |S_i|=0.
\end{cases}
\label{eq:fdp}
\end{equation}
Closure applies only to the reference; prediction scores are not propagated to ancestors. ``Unsupported'' means absent from this frozen reference, not biologically false. Each gene receives equal weight: one unsupported call among one released term and 20 among 20 both give $L_i=1$; one among 20 gives $L_i=0.05$. Releasing no terms (abstention) gives zero loss \cite{bates2021,angelopoulos2024crc}. We report supported yield separately because FDP does not penalize withheld functions.

\subsection{Step 2: separate ranking, certification and evaluation}
Each of 100 seeded random permutations divides a fixed gene panel into a ranking fold $\calT$ (10\%), a certification fold $\calC$ (60\%) and an evaluation fold $\calE$ (30\%). Ranking labels build the candidate policies; certification labels select a policy; evaluation labels assess its performance. For GoCo, the certified risk is the average FDP over the non-ranking pool $\calP=\calC\cup\calE$, of size $N$:
\begin{equation}
R(\pi)=\frac1N\sum_{i\in\calP}L_i(\pi),\qquad L_i(\pi)=L_i(S_i(\pi)).
\label{eq:risk}
\end{equation}
The target risk is $\alpha=0.10$, with sensitivity analyses at 0.05 and 0.20. Once a candidate policy is fixed, its pool losses are constants. Conditional on the ranking fold, certification genes are sampled uniformly without replacement from that pool. This sampling design supports inference without assuming independent gene losses, an important distinction under biological dependence \cite{shao2007,shang2012fdp}. The guarantee concerns the fixed panel and GO snapshot, not new genes outside the panel.

Module-family global methods have no ranking step and calibrate on $\calT\cup\calC$ (70\%); their certified pool is the whole panel. All neighbourhood-family methods certify on the same 60\% fold. These differences are retained when interpreting comparisons. Seeds, rounding and sample sizes are in Supplementary Note~\ref{supp-note:splits}.

\subsection{Step 3: construct policies from the predictor's evidence}
\label{sec:path}
GoDag, the first layer of GoCo, selects one global threshold using the TruePath loss. GoCo adds policies to exploit evidence already exposed by the predictor. A candidate path is a sequence of annotation sets that only grow as policies are visited. The entire path must be fixed using ranking labels and frozen inputs before certification.

\paragraph{Shared module evidence: GoCo-M.}
Many genes can receive identical scores because they share evidence. A threshold must admit or exclude an exact tie as a whole. GoCo-M keeps the predictor's threshold grid but inserts policies that admit only some affected genes within a grid step. A logistic model trained on ranking-fold calls estimates each call's probability of being unsupported. Features describe the original score, term frequency and the reliability of its supporting sources. For each source, reliability is learned from how often the calls it supports for ranking-fold genes are reference-unsupported. Source rates are shrunk toward the ranking-fold average to stabilize estimates from small counts. Only ranking-fold labels are used as training outcomes; frozen predictor scores remain inputs.

Affected genes are ordered by predicted FDP increase per predicted supported term. The ranking fold determines successive admission cutoffs; deterministic identifier hashes resolve ties. This is a practical ordering heuristic, not an optimality claim. The complete source attribution, shrinkage formula, model, insertion criteria and cutoffs are retained in Supplementary Notes~\ref{supp-note:ranking-model} and~\ref{supp-note:construction}, with constants in Supplementary Table~\ref{supp-tab:constants}.

\paragraph{Gene-specific neighbourhood evidence: GoCo-N.}
FunMap and STRING assemble evidence separately for each gene. Their scores have few substantial ties in the operating range, so a refined global grid provides a stronger threshold baseline. GoCo-N instead orders individual calls by their predicted unsupported probability, admitting the most promising first. This creates a separate nested path, not merely GoCo-M at a finer resolution.

We chose these implementations on Wainberg and FunMap, respectively, then applied them without retuning to external studies. The choice uses the predictor's evidence structure and score multiplicities, not evaluation outcomes. It is not universally optimal; an alternative that chooses between paths using ranking labels alone is assessed in Supplementary Note~\ref{supp-note:adaptive}.

\subsection{Step 4: certify the selected release policy}
\label{sec:cert}
For candidate $\pi_k$, let $\Rhat_k$ be its certification-fold mean FDP and $\shat_k$ its standard deviation, calculated with divisor $m=|\calC|$. We retain the one-sided normal test used in the released Multilabel implementation \cite{boger2025}:
\begin{equation}
p_k=\Phi\!\left(\frac{\Rhat_k-\alpha}{\max(\shat_k,10^{-6})/\sqrt m}\right),
\label{eq:pvalue}
\end{equation}
where $\Phi$ is the standard normal distribution function. Fixed-sequence Learn-Then-Test \cite{angelopoulos2025ltt} tests candidates in their prespecified order, stops at the first $p_k>\delta$ and returns the preceding policy $\widehat\pi$. If the first candidate fails, it abstains; if all pass, it returns the last. For module-family global methods, the mean, standard deviation and sample size use the 70\% calibration fold instead.

Certification means statistical control of the chance of releasing a policy whose pool risk exceeds $\alpha$, not just estimating its FDP. Under the stated finite-population variance and sampling conditions, that chance is asymptotically at most $\delta$ for $0<\delta\le0.50$. Any ordering learned solely from the ranking fold and frozen inputs is covered. The finite-population central limit argument \cite{hajek1960,erdos1959}, assumptions and formal statements are in Proposition~\ref{prop:validity}, Supplementary Corollary~\ref{supp-cor:order} and Supplementary Note~\ref{supp-note:proofs-proposition}. Replacing the normal test with Hoeffding's bound gives finite-sample control on the same path, usually with lower yield \cite{hoeffding1963}.

\begin{proposition}[Conditional fixed-sequence risk control]\label{prop:validity}
Let the candidate path satisfy Supplementary Definition~\ref{supp-def:path}, let $\delta\in(0,\tfrac12]$, and let $k^\ast$ be the first index with $R(\pi_{k^\ast})>\alpha$ (if none exists no error is possible). Consider a sequence of such problems in which $m\to\infty$, $N-m\to\infty$ and $m/N\to f\in[0,1)$, and suppose the pool variance $S^2(\pi_{k^\ast})=(N-1)^{-1}\sum_{i\in\calP}\{L_i(\pi_{k^\ast})-R(\pi_{k^\ast})\}^2$ stays bounded away from zero, and write $z_\delta=\Phi^{-1}(\delta)$. Then, conditional on $\calT$,
\[
\limsup\ \Pr\{R(\widehat\pi)>\alpha\mid\calT\}\ \le\ \Phi\!\bigl(z_\delta/\sqrt{1-f}\,\bigr)\ \le\ \delta,
\]
with equality in the second inequality only when $\delta=\tfrac12$ or $f=0$. The bound holds for every $\calT$-measurable ordering of the affected genes within each block, and unconditionally when the variance condition holds uniformly in $\calT$.
\end{proposition}

\noindent\textit{Proof idea.} If the returned policy is invalid, the first invalid candidate on the prespecified path must have passed its test. False certification is therefore bounded by that single test, without assuming independence among candidate tests or monotonicity of FDP along the path. See Supplementary Note~\ref{supp-note:proofs-proposition} for the full proof.

We report $\delta=0.50$ to match the released Multilabel default and $\delta=0.10$ as a stricter analysis. At 0.50, Equation~\eqref{eq:pvalue} reduces to $\Rhat_k\le\alpha$ and its formal false-certification bound is permissive (0.50). Held-out FDP can exceed $\alpha$ even when pool risk does not; these are different events. The exact relationship and stronger asymptotic bound under sampling without replacement are in Supplementary Remark~\ref{supp-rem:delta}.

\subsection{Data, comparators and reporting}
\label{sec:data}
Wainberg's co-essentiality predictor uses leave-one-gene-out GO enrichment of overlapping modules \cite{wainberg2021}. Its reconstructed scores match the archived file to numerical precision. Sanger Project Score \cite{behan2019} and Project DRIVE \cite{mcdonald2017} supply CRISPR and RNAi evidence to select among these modules without GO-based selection; HAP1 uses published genetic-interaction clusters \cite{billmann2026}. Panels contain 15,794, 14,885, 6,283 and 15,020 genes, respectively. All four use the GO release of 1 November 2019. For these module datasets, the reference includes all evidence codes, removes NOT-qualified annotations and root terms, and uses \emph{is\_a} and \emph{part\_of} for TruePath closure. FunMap \cite{shi2025funmap} and STRING v12.0 experiments-only \cite{szklarczyk2023string} use target-excluded top-50 random-walk neighbourhoods and GO release 2024-09-08. Their full construction and leakage checks are in Supplementary Note~\ref{supp-note:neighbourhood}; module reconstruction and reference details are in Supplementary Note~\ref{supp-note:data-details}.

\label{sec:comparators}
All methods use the same frozen scores, split streams and held-out TruePath reference within each dataset. Multilabel uses the Direct loss, whereas GoDag uses TruePath and GoCo also uses the predictor's evidence. Generic Direct-loss comparators use Hoeffding \cite{hoeffding1963}, Hoeffding--Bentkus \cite{bentkus2004}, an independent-observation normal approximation, empirical Bernstein \cite{maurer2009bernstein}, McDiarmid \cite{mcdiarmid1989}, Janson \cite{janson2004}, network heteroskedasticity and autocorrelation consistent (HAC) variance \cite{newey1987}, or Platt \cite{platt1999} and isotonic \cite{zadrozny2002} calibration. Graph/HAC assumptions are not established by our sampling design, so those are working-model comparisons. Platt and isotonic calibration provide no guarantee for the target risk. The method inventory is retained in Supplementary Table~\ref{supp-tab:methods}.

\label{sec:reporting}
We report mean held-out gene-level FDP, supported calls and genes, total calls, and pool-level failure and held-out exceedance frequencies over 100 splits. A supported call is a released gene--GO pair confirmed by the frozen reference. Split variation reflects random partitioning of this panel, not sampling new biological populations; frequency estimates have Monte Carlo standard errors up to 0.05. Supplementary Note~\ref{supp-note:conventions} defines all quantities and the handling of abstention.

\section{Results}
\subsection{Tied scores leave supported annotations unreleased}
\label{sec:ties}
Wainberg contains 2,470 gene--GO calls with the identical score 789.650. Its threshold step from 800 to 775 raises full-panel TruePath risk from 0.0924 to 0.1398 (Figure~\ref{fig:ties}). No threshold can divide the dominant tie, so GoDag stops at 800 in all 100 splits despite a target of 0.10. Partial admission provides intermediate policies. Random admission reaches the target at about 16\% in a descriptive full-panel interpolation; GoCo's informed ordering allowed certification at admission settings of 35--40\% in 89 splits. The interpolation is not itself a certificate. Supplementary Lemma~\ref{supp-lem:power} separates the contributions of finer admission steps and informative ordering without claiming guaranteed yield improvement. Across the certified module-family blocks, the ordering statistic's Spearman correlation with realised loss increments was $0.22$--$0.73$ on non-ranking genes (Supplementary Table~\ref{supp-tab:blockdiag}).

\begin{figure*}[!t]
\centering
\includegraphics[width=0.99\textwidth]{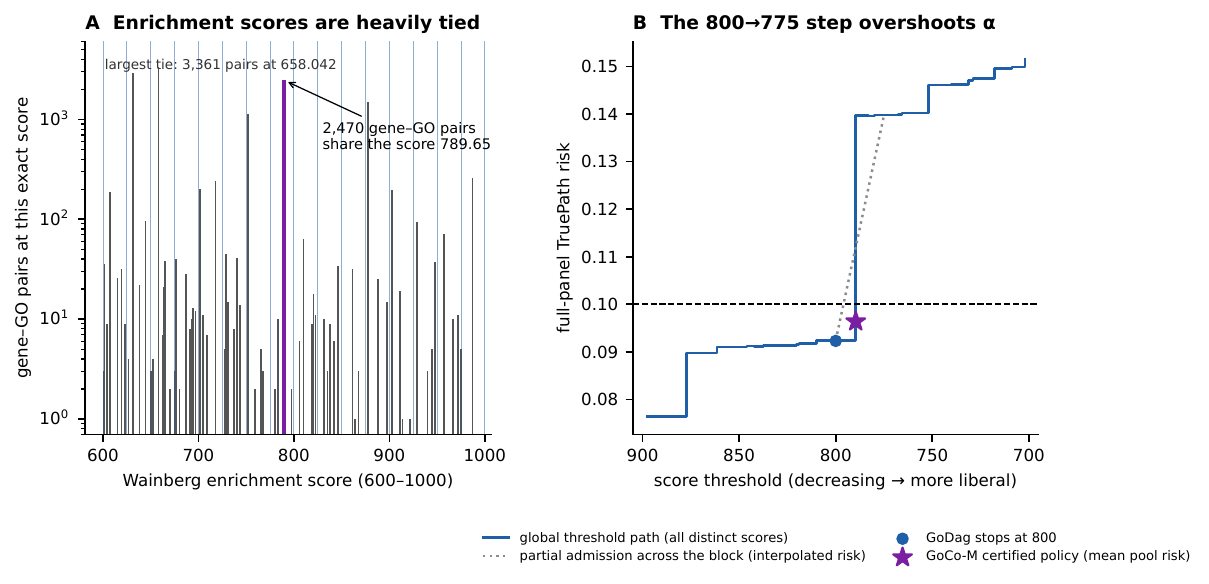}
\caption{Score ties and the global-threshold ceiling in Wainberg. (A) Exact-score multiplicities between 600 and 1000. (B) Full-panel TruePath risk across distinct thresholds, highlighting the $800\to775$ step that crosses the $\alpha=0.10$ target and the partial-admission path within that step.}
\par\small\textit{Alt text:} Two-panel illustration of score ties in Wainberg. Panel A shows large multiplicities of identical enrichment scores, including 2,470 gene--GO pairs at score 789.65. Panel B shows the full-panel TruePath risk jumping across the 0.10 target between thresholds 800 and 775, whereas partial admission interpolates within the tied block.
\label{fig:ties}
\end{figure*}

\subsection{Shared module evidence: GoCo-M}
\label{sec:module}
At $\alpha=0.10$, $\delta=0.50$, GoCo-M's mean held-out FDP was 0.0964, 0.0980, 0.0971 and 0.0943 in Wainberg, Sanger, DRIVE and HAP1 (Table~\ref{tab:panelmodule50}; Figure~\ref{fig:primary}). Supported yield increased by 8.7\%, 14.3\%, 17.6\% and 6.3\% relative to Multilabel. GoDag improved on Multilabel when TruePath changed the selected threshold; GoCo-M added supported calls in all four datasets. Comparisons with GoDag assess the full procedure, not the path alone, because their certification folds differ (60\% versus 70\%). Gains persisted at $\delta=0.10$: 7.8\%, 7.7\%, 15.6\% and 6.3\%, respectively (Supplementary Table~\ref{supp-tab:detail}; Supplementary Figure~\ref{supp-fig:supp-primary}).

No generic Direct-loss comparator exceeded GoDag in supported yield: the concentration-bound members annotated fewer genes at a realised FDP below target, and at $\delta=0.50$ the normal-type rules coincide with Multilabel exactly (Supplementary Remark~\ref{supp-rem:delta}). GoCo-M's distinguishing feature is a richer set of candidate policies, not a sharper bound on the same threshold path.

\begin{table*}[htbp]
\centering
\caption{Module-structure comparison of mean held-out FDP and supported yield at $\alpha=0.10$, $\delta=0.50$, over 100 splits. Wainberg is the development data set; Sanger, DRIVE and HAP1 are external studies with the same shared-evidence structure.}
\label{tab:panelmodule50}
\scriptsize
\setlength{\tabcolsep}{3pt}
\resizebox{\linewidth}{!}{\begin{tabular}{l rrr rrr rrr rrr}
\toprule
Method & \multicolumn{3}{c}{Wainberg ($n_{\mathcal E}$=4,738)} & \multicolumn{3}{c}{Sanger ($n_{\mathcal E}$=4,465)} & \multicolumn{3}{c}{DRIVE ($n_{\mathcal E}$=1,885)} & \multicolumn{3}{c}{HAP1 ($n_{\mathcal E}$=4,506)} \\
\cmidrule(lr){2-4}\cmidrule(lr){5-7}\cmidrule(lr){8-10}\cmidrule(lr){11-13}
 & Gene FDP & Genes with a & Correct & Gene FDP & Genes with a & Correct & Gene FDP & Genes with a & Correct & Gene FDP & Genes with a & Correct \\
 & & correct call & calls & & correct call & calls & & correct call & calls & & correct call & calls \\
\midrule
\multicolumn{13}{l}{\emph{Direct loss, one global threshold (ontology-blind)}} \\
Multilabel & 0.0923 & 185.7 & 339.6 & 0.0953 & 155.4 & 333.0 & 0.0945 & 96.1 & 248.1 & 0.0719 & 77.0 & 129.4 \\
Hoeffding/LTT & 0.0914 & 182.6 & 333.3 & 0.0876 & 148.1 & 308.0 & 0.0864 & 84.2 & 206.1 & 0.0719 & 77.0 & 129.4 \\
Hoeffding--Bentkus & 0.0923 & 185.7 & 339.6 & 0.0916 & 150.9 & 316.4 & 0.0923 & 91.8 & 233.9 & 0.0719 & 77.0 & 129.4 \\
IID-normal & 0.0923 & 185.7 & 339.6 & 0.0953 & 155.4 & 333.0 & 0.0945 & 96.1 & 248.1 & 0.0719 & 77.0 & 129.4 \\
Empirical Bernstein & 0.0919 & 183.6 & 335.7 & 0.0894 & 149.3 & 311.5 & 0.0882 & 86.0 & 213.1 & 0.0719 & 77.0 & 129.4 \\
McDiarmid (binary incidence) & 0.0804 & 166.1 & 291.8 & 0.0777 & 142.2 & 288.6 & 0.0782 & 78.4 & 183.3 & 0.0719 & 77.0 & 129.4 \\
Janson dependency graph & 0.0775 & 162.1 & 281.0 & 0.0769 & 141.7 & 287.0 & 0.0746 & 76.0 & 175.7 & 0.0719 & 77.0 & 129.4 \\
Network-HAC ($b=1$) & 0.0923 & 185.7 & 339.6 & 0.0953 & 155.4 & 333.0 & 0.0945 & 96.1 & 248.1 & 0.0719 & 77.0 & 129.4 \\
Platt calibration & 0.0908 & 182.3 & 332.1 & 0.0949 & 155.6 & 334.4 & 0.0914 & 91.7 & 233.6 & 0.0719 & 77.0 & 129.4 \\
Isotonic calibration & 0.0919 & 184.8 & 337.6 & 0.0956 & 156.5 & 337.8 & 0.0917 & 92.0 & 235.3 & 0.0719 & 77.0 & 129.4 \\
\addlinespace[2pt]\multicolumn{13}{l}{\emph{TruePath loss, one global threshold (GO hierarchy)}} \\
GoDag & 0.0923 & 185.7 & 339.6 & 0.0980 & 160.9 & 353.8 & 0.0965 & 98.7 & 258.6 & 0.0719 & 77.0 & 129.4 \\
\addlinespace[2pt]\multicolumn{13}{l}{\emph{TruePath loss, structure-guided partial admission (GO hierarchy and module-shared evidence)}} \\
\textbf{GoCo-M} & \textbf{0.0964} & \textbf{197.1} & \textbf{369.1} & \textbf{0.0980} & \textbf{166.9} & \textbf{380.7} & \textbf{0.0971} & \textbf{105.7} & \textbf{291.7} & \textbf{0.0943} & \textbf{80.8} & \textbf{137.7} \\
\bottomrule
\end{tabular}}
\par\vspace{3pt}\begin{minipage}{\linewidth}\footnotesize Gene FDP: mean held-out TruePath false discovery proportion. Correct calls: reference-supported gene--GO pairs. Genes with a correct call: genes with at least one supported release. All global methods use the same 25-unit grid and 70\% calibration; GoCo-M uses 10\% ranking and 60\% certification. Platt/isotonic rules are $\delta$-free and have no risk guarantee. Normal-type rules coincide at $\delta=0.50$ on the same path, loss and calibration sample (Supplementary Remark~\ref{supp-rem:delta}).\end{minipage}
\end{table*}
\begin{figure*}[!t]
\centering
\includegraphics[width=0.99\textwidth]{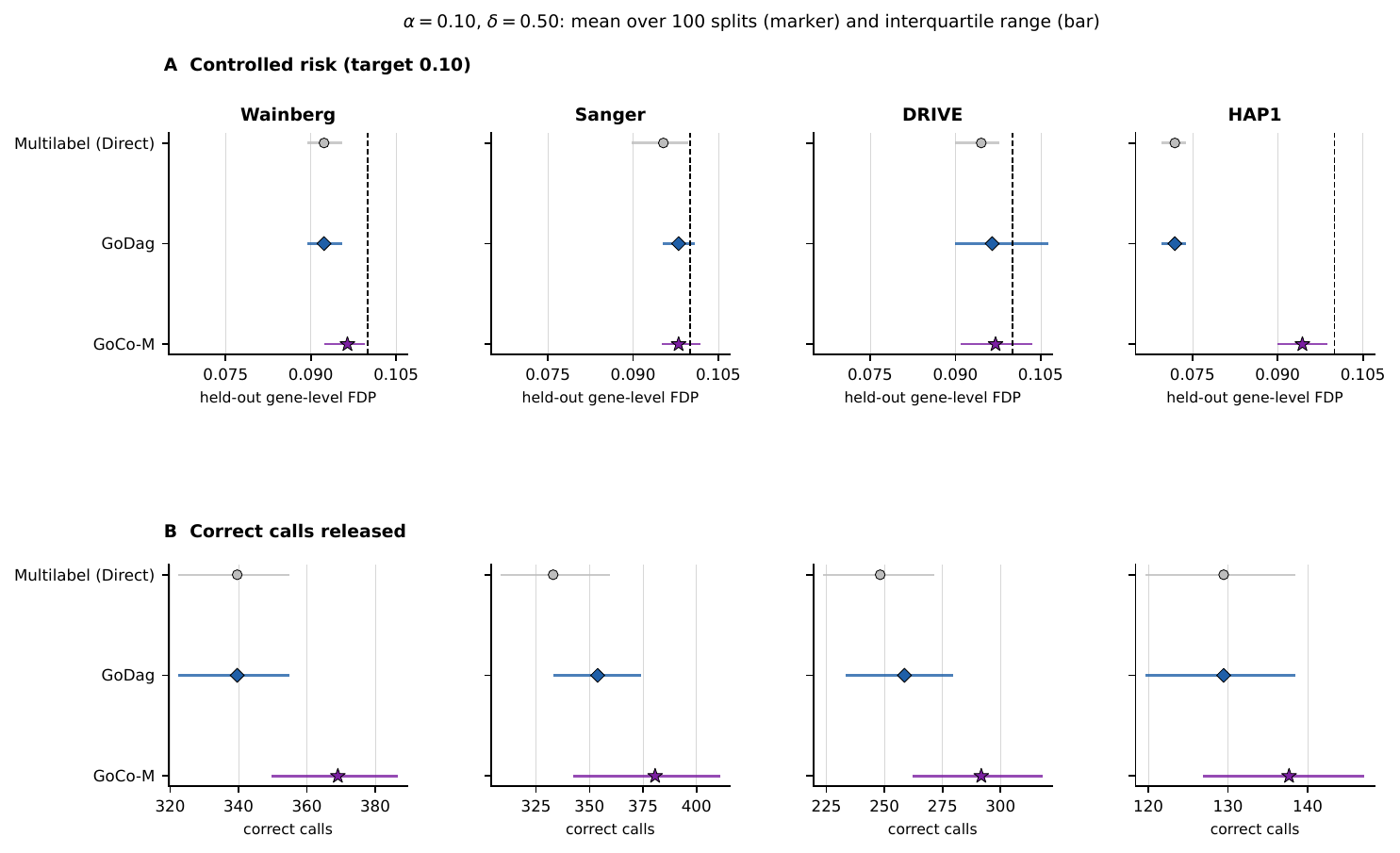}
\caption{Primary comparison of Multilabel, GoDag and GoCo-M at $\alpha=0.10$, $\delta=0.50$: held-out gene-level TruePath FDP and reference-supported calls over 100 splits.}
\par\small\textit{Alt text:} Primary comparison across Wainberg, Sanger, DRIVE and HAP1. The upper row shows mean held-out gene-level TruePath FDP and interquartile ranges for Multilabel, GoDag and GoCo-M relative to the 0.10 target. The lower row shows reference-supported released gene--GO calls, with GoCo-M releasing the most in each data set.
\label{fig:primary}
\end{figure*}

\subsection{Confidence and error-target sensitivity}
GoCo-M exceeded the pool-risk target in 5--26\% of splits at $\delta=0.50$, below that setting's permissive formal bound. At $\delta=0.10$, pool failure occurred in zero or one split per module dataset. These frequencies are consistent with, but too coarse to assess precisely, asymptotic upper bounds near 0.013 for GoCo and 0.010 for global methods. Held-out exceedance remained more frequent (Supplementary Tables~\ref{supp-tab:detail} and~\ref{supp-tab:riskdecomp}).

At $\alpha=0.05$ and 0.20, mean held-out FDP remained below target for the principal methods, and GoCo-M matched or exceeded Multilabel's supported yield (Figure~\ref{fig:sweep}; Supplementary Tables~\ref{supp-tab:sweep50} and~\ref{supp-tab:sweep10}). Gains were largest where a grid step crossed the target sharply. The zero gain in DRIVE at $\alpha=0.20$, $\delta=0.50$ arose because both methods reached the grid's reporting floor.

\begin{figure*}[!t]
\centering
\includegraphics[width=0.99\textwidth]{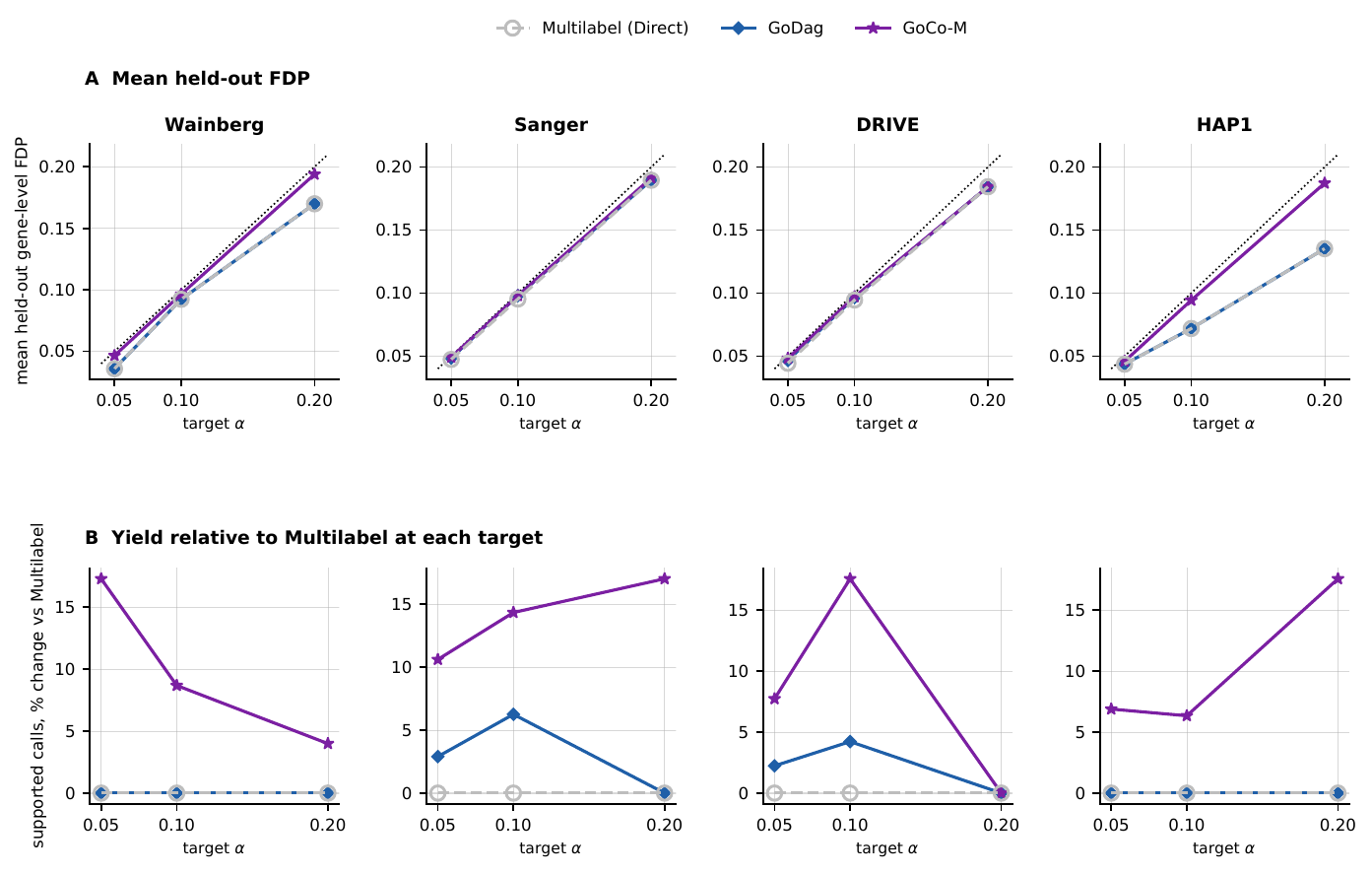}
\caption{Error-target sensitivity at $\delta=0.50$. Top: mean held-out gene-level FDP; the dotted line is the target. Bottom: supported-call gain relative to Multilabel. Method colours match Figures~\ref{fig:ties} and~\ref{fig:primary} and Supplementary Figure~\ref{supp-fig:supp-primary}. Open grey circles keep Multilabel visible where it overlaps GoDag. Full values are in Supplementary Table~\ref{supp-tab:sweep50}.}
\par\small\textit{Alt text:} Four-column sensitivity plots for Wainberg, Sanger, DRIVE and HAP1. Top panels show mean held-out gene-level TruePath FDP for Multilabel, GoDag and GoCo-M at $\alpha=0.05$, 0.10 and 0.20 relative to the target. Bottom panels show supported-call gain relative to Multilabel across the same targets.
\label{fig:sweep}
\end{figure*}

\subsection{Yield and biological interpretation}
Only about 8--12\% of held-out genes received calls in the module analyses. The gene-averaged target includes abstentions and controls neither the FDP among annotated genes alone nor the pooled fraction of unsupported calls. Additional yield therefore comes with additional unsupported releases. At $\delta=0.10$, GoCo-M released 1,896 versus GoDag's 1,614 calls in Wainberg and 1,252 versus 847 in HAP1; the supplement reports their supported fractions (Supplementary Tables~\ref{supp-tab:riskdecomp} and~\ref{supp-tab:estimands}).

In Wainberg split~0 at $\delta=0.50$, GoCo-M released 351 held-out calls that GoDag did not, over 158 genes; 37 (10.5\%) were reference-supported. All four interleukin-27/35 signalling calls added for IL6ST and JAK1 in an interleukin-6 signal-transduction module were supported. CDKN1A belonged to a module mixing a regulator with its targets; none of the seven transferred TP53 terms was supported. These examples illustrate both useful evidence sharing and its limits. N4BP1's added Toll-like receptor signalling annotation was absent from the frozen reference despite published biological support \cite{shi2021n4bp1}. The full case study is retained in Supplementary Note~\ref{supp-note:case-study} and Table~\ref{supp-tab:casestudy}.

\subsection{Gene-specific neighbourhood evidence: GoCo-N}
\label{sec:secondfamily}
Developed on FunMap and validated without retuning on STRING's experiments-only network, GoCo-N increased supported calls by 33.8\% and 36.4\% relative to Multilabel on the same refined score grid. Mean held-out FDP was 0.0921 and 0.0944 (Table~\ref{tab:panelnbr50}). Compared with refined-grid GoDag, it added 342 and 246 supported calls. The stricter $\delta=0.10$ gains over Multilabel were 34.0\% and 35.4\%. Supplementary Table~\ref{supp-tab:frontier} compares yield with the best global threshold at matched realised risk.

\begin{table*}[htbp]
\centering
\caption{Neighbourhood-structure comparison at $\alpha=0.10$, $\delta=0.50$, over 100 splits. FunMap is the development data set; the STRING v12.0 experiments-only network provides the external validation.}
\label{tab:panelnbr50}
\scriptsize
\setlength{\tabcolsep}{5pt}
\begin{tabular}{l rrr rrr}
\toprule
Method & \multicolumn{3}{c}{FunMap ($n_{\mathcal E}$=3,095)} & \multicolumn{3}{c}{STRING ($n_{\mathcal E}$=2,417)} \\
\cmidrule(lr){2-4}\cmidrule(lr){5-7}
 & Gene FDP & Correct & Genes with & Gene FDP & Correct & Genes with \\
 & & calls & a call & & calls & a call \\
\midrule
\multicolumn{7}{l}{\emph{Direct loss, one global threshold on the refined grid (ontology-blind)}} \\
Multilabel & 0.0888 & 1245 & 734 & 0.0899 & 885 & 432 \\
Hoeffding/LTT & 0.0820 & 1182 & 699 & 0.0820 & 826 & 405 \\
Hoeffding--Bentkus & 0.0855 & 1213 & 717 & 0.0855 & 857 & 419 \\
IID-normal & 0.0888 & 1245 & 734 & 0.0899 & 885 & 432 \\
Empirical Bernstein & 0.0835 & 1193 & 706 & 0.0831 & 837 & 410 \\
Platt calibration & 0.0647 & 1061 & 615 & 0.1003 & 989 & 487 \\
Isotonic calibration & 0.0824 & 1184 & 701 & 0.0874 & 869 & 424 \\
\addlinespace[2pt]\multicolumn{7}{l}{\emph{TruePath loss, one global threshold (GO hierarchy)}} \\
GoDag & 0.0901 & 1261 & 742 & 0.0828 & 828 & 406 \\
GoDag, refined grid & 0.0986 & 1324 & 783 & 0.0978 & 962 & 472 \\
\addlinespace[2pt]\multicolumn{7}{l}{\emph{TruePath loss, structure-guided ordering of calls (GO hierarchy and neighbourhood evidence)}} \\
\textbf{GoCo-N} & \textbf{0.0921} & \textbf{1666} & \textbf{1110} & \textbf{0.0944} & \textbf{1208} & \textbf{837} \\
\bottomrule
\end{tabular}
\par\vspace{3pt}\begin{minipage}{\linewidth}\footnotesize Gene FDP: mean held-out TruePath false discovery proportion. Correct calls: reference-supported gene--GO pairs. Genes with a call includes unsupported releases. Direct-loss rules use the refined grid; GoDag is shown on both coarse and refined grids. All methods certify on the same 60\% fold. Platt/isotonic rules have no risk guarantee. McDiarmid, Janson and network-HAC variants were not evaluated here because their graph/variance specification was not prespecified.\end{minipage}
\end{table*}

The structural choice is useful but not uniformly best. At $\alpha=0.10$, $\delta=0.10$, GoCo-N lost 142 supported terms relative to GoCo-M on Wainberg, while it performed better on Sanger and HAP1. A ranking-fold-only selection rule agreed with the held-out winner in 69--100\% of splits across predictors (Supplementary Note~\ref{supp-note:adaptive}). This remains a refinement rather than the primary prespecified analysis.

\subsection{Comparison with accuracy-based release rules}
A validation $F_{\max}$ threshold and one top-scoring term per gene produced gene-level FDPs of 0.18--0.99 and 0.17--0.98, respectively. Both exceeded 0.10 in every split of every predictor (Supplementary Table~\ref{supp-tab:current}). They can recover more supported functions by releasing more calls, but provide no way to certify the accompanying gene-level error rate.

\section{Discussion}
\label{sec:discussion}
GoCo addresses a practical problem: deciding which annotations to release from an existing predictor at a stated error target. GoDag supplies a hierarchy-aware global-threshold baseline; GoCo uses evidence behind calls to improve supported yield. The protocol separates ordering from certification, allowing users to compare release policies without retraining their predictor.

The target must match the user's purpose. GoCo's gene-averaged FDP follows risk-control formulations used in conformal calibration and enzyme annotation \cite{angelopoulos2024crc,ding2024cpec}. It is not pooled call-level FDR \cite{jin2023selection,bates2023outliers} or a guarantee for each gene or subgroup \cite{vovk2013conditional,gibbs2025conditional}. The pooled-call comparator in Supplementary Table~\ref{supp-tab:estimands} releases almost nothing at the same numerical target, but addresses a different estimand; pair-level exchangeability is not established by our gene-level split. Using side information to improve yield also connects GoCo to covariate-adaptive testing \cite{ignatiadis2016ihw,lei2018adapt} and smoothed conformal prediction \cite{vovk2005book}, with different targets and constructions.

GO is incomplete and open-world: missing annotations are not evidence of absent function \cite{warwick2020negative,jiang2014incomplete,zhapa2024pugo}. Functional relevance also varies by tissue, cell state, development, species and experimental condition. Supported yield and reference-based FDP should therefore be interpreted together, not as complete measures of biological correctness. Evidence sources must exclude the target gene's own annotations to avoid circular support. Supplementary Note~\ref{supp-note:conventions} records the remaining label information already present in frozen enrichment scores.

\subsection*{Limitations and use beyond these predictors}
The guarantee is asymptotic, requires nonvanishing variance at the first invalid candidate and applies to this fixed panel and reference. The primary $\delta=0.50$ setting is a benchmark convention, not a strong deployment confidence level. Sanger, DRIVE and HAP1 contribute independent studies but share Wainberg's gene universe and reference; they are not independent predictor families. The better path cannot yet be predicted reliably from structure alone.

GoDag needs only scores. GoCo additionally needs interpretable evidence behind individual calls. Protein complexes and co-regulation clusters are possible shared sources; homology hits \cite{kulmanov2020deepgoplus,zhang2024search}, embedding neighbours \cite{littmann2021embedding,boger2025,ding2024cpec} and integrated networks \cite{mostafavi2008} are possible gene-specific sources. These extensions require empirical validation, not simply adoption of the structural default.

\section{Conclusions}
GoCo combines an unchanged predictor, GO-aware error assessment and structure-guided release policies. Across six applications, the structure-specific GoCo implementations increased supported yield while keeping mean held-out gene-level FDP at or below target; the four external applications were evaluated without retuning. The practical contribution is a reproducible workflow for comparing annotation release rules, interpreting their error statements and retaining more reference-supported functions.

\section*{Key Points}
\begin{itemize}
\item Accuracy-optimized and fixed-size release rules provide no gene-level error certificate; FDP ranged from 0.17 to 0.99 across these applications.
\item GoCo combines GO-aware calibration with evidence already exposed by an unchanged predictor.
\item An asymptotic certificate controls pool-average gene-level FDP for paths built without certification or evaluation labels, not pooled-call FDR or each gene's error.
\item GoCo-M increased supported yield by 6.3--17.6\% over Multilabel across four module datasets at the primary operating point.
\item GoCo-N increased supported yield by 33.8\% and 36.4\% in FunMap and STRING against Multilabel on the same refined score grid.
\end{itemize}

\section*{Supplementary data}
Supplementary material accompanies this article online.

\section*{Acknowledgements}
We thank the authors of the co-essential module atlas, Sanger Project Score, Project DRIVE, the HAP1 genetic-interaction map, FunMap and STRING for making their data publicly available.

\textit{Declaration of the use of artificial intelligence.} The authors used Claude Opus 5 (Anthropic) and ChatGPT GPT-5.6 Sol (OpenAI) as aids for language editing of author-written text, consistency checking, and assistance with code writing and documentation. All scientific decisions, interpretations and conclusions were made by the authors. All model-assisted suggestions and code were reviewed and verified by the authors, who take full responsibility for the manuscript.

\section*{Funding}
This work was supported by the Plant Sciences Institute, Iowa State University; and the Laurence H. Baker Center for Bioinformatics and Biological Statistics, Iowa State University.

\section*{Author contributions}
Chong Wang (Conceptualization, Data curation, Formal analysis, Investigation, Methodology, Software, Validation, Visualization, Writing---original draft, Writing---review \& editing), Yongzhao Shao (Methodology, Writing---review \& editing), Peng Liu (Methodology, Funding acquisition, Writing---review \& editing). All authors read and approved the final manuscript.

\section*{Conflict of interest}
None declared.

\section*{Data availability}
No new experimental data were generated in this study. The reference implementation of both families (\texttt{goco\_rerun.py} and \texttt{goco\_learned.py} for the module family; \texttt{goco\_second\_family.py} and \texttt{run\_goco\_second\_\allowbreak family.py} for the neighbourhood family), the measurability and verification tests, the table and figure scripts, the explicit fold assignments of all 100 splits, the tie-multiplicity counts, the principal split-level analysis outputs and the aggregate summaries used for the secondary comparator panel are openly available in the GitHub repository \url{https://github.com/ChongWangStat/ConformalGoCo}. The released software and accompanying reproducibility materials are archived on Zenodo under the stable concept DOI 10.5281/zenodo.22677009 \cite{wang2026gocozenodo}.

\section*{Author descriptions}
Chong Wang is Professor in the Department of Statistics and in the Department of Veterinary Diagnostic and Production Animal Medicine at Iowa State University. His research covers biostatistics, epidemiology, diagnostic test evaluation, network meta-analysis and bioinformatics.

Yongzhao Shao is Professor in the Division of Biostatistics, Department of Population Health, NYU Grossman School of Medicine. His research spans biostatistics for precision medicine, cancer and environmental epidemiology, survival analysis and gene--environment interaction.

Peng Liu is Professor of Statistics at Iowa State University and Director of the Laurence H. Baker Center for Bioinformatics and Biological Statistics. Her research covers statistical genomics, bioinformatics, Bayesian methods and high-dimensional inference.

\clearpage
\setcounter{equation}{0}
\setcounter{table}{0}
\setcounter{figure}{0}
\setcounter{proposition}{0}
\setcounter{lemma}{0}
\setcounter{corollary}{0}
\setcounter{remark}{0}
\setcounter{definition}{0}
\renewcommand{\theequation}{S\arabic{equation}}
\renewcommand{\theproposition}{S\arabic{proposition}}
\renewcommand{\thelemma}{S\arabic{lemma}}
\renewcommand{\thecorollary}{S\arabic{corollary}}
\renewcommand{\theremark}{S\arabic{remark}}
\renewcommand{\thedefinition}{S\arabic{definition}}
\renewcommand{\thetable}{S\arabic{table}}
\renewcommand{\thefigure}{S\arabic{figure}}
\newcounter{suppnote}
\newcounter{suppsubnote}[suppnote]
\renewcommand{\thesuppnote}{S\arabic{suppnote}}
\renewcommand{\thesuppsubnote}{\thesuppnote.\arabic{suppsubnote}}
\newcommand{\suppnote}[2]{%
  \refstepcounter{suppnote}%
  \section*{Supplementary Note \thesuppnote: #1}%
  \label{supp-#2}%
}
\newcommand{\suppsubnote}[2]{%
  \refstepcounter{suppsubnote}%
  \subsection*{\thesuppsubnote\ #1}%
  \label{supp-#2}%
}
\fullwidthfrontmatter{
\begin{center}
{\LARGE\bfseries Supplementary material for: GoCo: structure-guided conformal calibration for error-controlled Gene Ontology annotation release\par}
\vspace{0.8em}
{\large Chong Wang$^{1,2}$, Yongzhao Shao$^3$, and Peng Liu$^{1,4,*}$\par}
\vspace{0.6em}
{\small
$^1$Department of Statistics, Iowa State University, Ames, Iowa, USA\\
$^2$Department of Veterinary Diagnostic and Production Animal Medicine, College of Veterinary Medicine, Iowa State University, Ames, Iowa, USA\\
$^3$Division of Biostatistics, Department of Population Health, NYU Grossman School of Medicine, New York, New York, USA\\
$^4$Laurence H. Baker Center for Bioinformatics and Biological Statistics, Iowa State University, Ames, Iowa, USA\\
$^*$Corresponding author: Peng Liu, Department of Statistics, Iowa State University, 2117 Snedecor Hall, 2438 Osborn Drive, Ames, IA 50011, USA. Tel: +1 515 294 7806. E-mail: \texttt{pliu@iastate.edu}\par}
\end{center}
\vspace{0.4em}
\begin{abstract}
This supplement contains the proofs supporting the main-text risk guarantee and a definition of the GoDag baseline (Supplementary Note~\ref{supp-note:proofs}), the construction and implementation details of the method---split mechanics, interpretation of the frozen GO reference, the per-call ranking model, the arithmetic behind Remark~\ref{supp-rem:delta}, the uniform grid path and a worked micro-example of one block (Supplementary Note~\ref{supp-note:methods})---and the second structural instantiation with its validation without retuning and ranking-fold selection rule (Supplementary Note~\ref{supp-note:neighbourhood}). Tables report split-level operating characteristics (Table~\ref{supp-tab:detail}); the error-target sweep at $\delta=0.10$ (Table~\ref{supp-tab:sweep10}); certification, pool and held-out risks with yield diagnostics (Table~\ref{supp-tab:riskdecomp}); block diagnostics (Table~\ref{supp-tab:blockdiag}); illustrative released genes (Table~\ref{supp-tab:casestudy}); the gene-level versus pooled-call estimand comparison (Table~\ref{supp-tab:estimands}); source-object structures (Table~\ref{supp-tab:structures}); gains over the matched-risk frontier in the neighbourhood-structure family (Table~\ref{supp-tab:frontier}); the ranking-fold selection rule (Table~\ref{supp-tab:selection}); commonly used release rules on all six predictors (Table~\ref{supp-tab:current}); prespecified constants (Table~\ref{supp-tab:constants}); the method inventory (Table~\ref{supp-tab:methods}); and the $\delta=0.50$ error-target sweep (Table~\ref{supp-tab:sweep50}). Figure~\ref{supp-fig:supp-primary} gives the primary comparison at $\delta=0.10$. The full construction and additional formal statements are given in Supplementary Note~\ref{supp-note:extended}, with AI-use disclosure in Supplementary Note~\ref{supp-note:ai}. Main-text cross-references are explicitly identified; equations and formal results labelled S refer to this supplement.
\end{abstract}
\noindent\textbf{Keywords:} Gene Ontology; conformal risk control; Learn-Then-Test
\vspace{1.0em}
}

\setcounter{lemma}{0}\renewcommand{\thelemma}{S\arabic{lemma}}
\setcounter{remark}{0}\renewcommand{\theremark}{S\arabic{remark}}

\suppnote{Proofs}{note:proofs}

\suppsubnote{Setting and notation}{note:proofs-setting}
The gene panel $\calU=\{1,\ldots,n\}$, the score matrix $(s_{ig})$, the GO snapshot and the reference annotations are fixed. For every release policy $\pi$ and gene $i$ the loss $L_i(\pi)=L_i(S_i(\pi))\in[0,1]$ of Equation~\eqref{eq:fdp} is therefore a fixed number. The only random object is a uniformly distributed permutation $\varsigma$ of $\calU$. With $n_{\calT}=\operatorname{round}(0.10n)$ and $n_{\calT}+m=\operatorname{round}(0.70n)$, the ranking fold is $\calT=\varsigma(\{1,\ldots,n_{\calT}\})$, the certification fold is $\calC=\varsigma(\{n_{\calT}+1,\ldots,n_{\calT}+m\})$ and the evaluation fold is $\calE=\varsigma(\{n_{\calT}+m+1,\ldots,n\})$. The pool is $\calP=\calU\setminus\calT$ with $N=|\calP|=m+n_{\calE}$ and $f=m/N$. For a policy $\pi$ write $R(\pi)=N^{-1}\sum_{i\in\calP}L_i(\pi)$, $\Rhat_{\calC}(\pi)=m^{-1}\sum_{i\in\calC}L_i(\pi)$, $\Rhat_{\calE}(\pi)=n_{\calE}^{-1}\sum_{i\in\calE}L_i(\pi)$, $\sigma^2(\pi)=N^{-1}\sum_{i\in\calP}\{L_i(\pi)-R(\pi)\}^2$ and $S^2(\pi)=N\sigma^2(\pi)/(N-1)$. For a candidate $\pi_k$ abbreviate $\Rhat_k=\Rhat_{\calC}(\pi_k)$ and $\shat_k^2=m^{-1}\sum_{i\in\calC}\{L_i(\pi_k)-\Rhat_k\}^2$ (divisor $m$, matching \texttt{ddof=0} in the implementation), and let $\tilde\sigma_k=\max(\shat_k,10^{-6})$ and $z_\delta=\Phi^{-1}(\delta)$.

\suppsubnote{The split as a sampling design}{note:proofs-design}

\begin{lemma}[Conditional sampling design]\label{supp-lem:design}
Conditional on $\calT$, the certification fold $\calC$ is uniformly distributed over the $m$-element subsets of $\calP$, i.e.\ a simple random sample without replacement (SRSWOR) of size $m$ from $\calP$, and $\calE=\calP\setminus\calC$ is a SRSWOR of size $n_{\calE}$ from $\calP$. Unconditionally, $\calT\cup\calC$ is a SRSWOR of size $n_{\calT}+m$ from $\calU$.
\end{lemma}
\begin{proof}
Fix $T_0$ with $|T_0|=n_{\calT}$ and an $m$-subset $C_0\subseteq\calU\setminus T_0$. The number of permutations with $\varsigma(\{1,\ldots,n_{\calT}\})=T_0$ and $\varsigma(\{n_{\calT}+1,\ldots,n_{\calT}+m\})=C_0$ is $n_{\calT}!\,m!\,n_{\calE}!$, independent of $C_0$; hence all $\binom{N}{m}$ subsets are equally likely given $\calT=T_0$. The complement of a uniform $m$-subset of $\calP$ is a uniform $n_{\calE}$-subset. The last claim follows in the same way from the count $(n_{\calT}+m)!\,n_{\calE}!$.
\end{proof}

\begin{lemma}[Moments under SRSWOR]\label{supp-lem:moments}
Let $b_i\in[0,1]$, $i\in\calP$, be fixed numbers with mean $\bar b_{\calP}$ and $S_b^2=(N-1)^{-1}\sum_{i\in\calP}(b_i-\bar b_{\calP})^2$. If $\calC$ is a SRSWOR of size $m$ from $\calP$ and $\bar b_{\calC}=m^{-1}\sum_{i\in\calC}b_i$, then $\mathbb E\,\bar b_{\calC}=\bar b_{\calP}$, $\operatorname{Var}\bar b_{\calC}=(1-f)S_b^2/m\le1/(2m)$, and with $\bar b_{\calE}$ the mean over $\calE$, $N\bar b_{\calP}=m\bar b_{\calC}+n_{\calE}\bar b_{\calE}$.
\end{lemma}
\begin{proof}
Let $I_i=\mathbf 1\{i\in\calC\}$; then $\Pr(I_i=1)=f$ and $\Pr(I_i=I_j=1)=m(m-1)/\{N(N-1)\}$ for $i\ne j$, so $\operatorname{Var}I_i=f(1-f)$ and $\operatorname{Cov}(I_i,I_j)=-f(1-f)/(N-1)$. Unbiasedness follows from $\mathbb E\sum_iI_ib_i=f\sum_ib_i$. For the variance, $\operatorname{Var}(\sum_iI_ib_i)=f(1-f)\{\sum_ib_i^2-(N-1)^{-1}\sum_{i\ne j}b_ib_j\}=f(1-f)N S_b^2$, and dividing by $m^2=f^2N^2$ gives $(1-f)S_b^2/m$. Since $b_i\in[0,1]$, $\sigma_b^2\le1/4$ and $S_b^2\le N/\{4(N-1)\}\le1/2$ for $N\ge2$; the variance formula is Theorem 2.2 of Cochran (1977). The identity is $\sum_{\calP}b_i=\sum_{\calC}b_i+\sum_{\calE}b_i$.
\end{proof}

\suppsubnote{Finite-population central limit theorem}{note:proofs-clt}

\begin{lemma}[H\'ajek 1960; Erd\H{o}s and R\'enyi 1959]\label{supp-lem:fpclt}
For each $t$ let $b_{t,i}\in[0,1]$, $i\in\calP_t$, $|\calP_t|=N_t$, be fixed numbers with mean $\bar b_t$, variance $\sigma_t^2=N_t^{-1}\sum_i(b_{t,i}-\bar b_t)^2$ and $S_t^2=N_t\sigma_t^2/(N_t-1)$, and let $\calC_t$ be a SRSWOR of size $m_t$ from $\calP_t$ with $f_t=m_t/N_t$. If $m_t(1-f_t)\sigma_t^2\to\infty$, then
\[
W_t=\frac{\bar b_{\calC_t}-\bar b_t}{\sqrt{(1-f_t)S_t^2/m_t}}\ \xrightarrow{\ d\ }\ \mathcal N(0,1).
\]
\end{lemma}
\begin{proof}
Let $s_t^2=\operatorname{Var}\sum_{i\in\calC_t}b_{t,i}=m_t(1-f_t)S_t^2$ (Lemma~\ref{supp-lem:moments}). H\'ajek (1960) proved that $W_t$ is asymptotically standard normal if and only if, for every $\tau>0$, $\{\sum_i(b_{t,i}-\bar b_t)^2\}^{-1}\sum_{i:|b_{t,i}-\bar b_t|>\tau s_t}(b_{t,i}-\bar b_t)^2\to0$; Erd\H{o}s and R\'enyi (1959) gave a sufficient condition of the same type. Since $|b_{t,i}-\bar b_t|\le1$ and $s_t^2\ge m_t(1-f_t)\sigma_t^2\to\infty$, the index set $\{i:|b_{t,i}-\bar b_t|>\tau s_t\}$ is empty once $\tau s_t>1$, so the condition holds. Note that $m_t(1-f_t)\sigma_t^2\to\infty$ forces $m_t\to\infty$ and $N_t-m_t\to\infty$, because $\sigma_t^2\le1/4$ and $m_t(1-f_t)\le\min(m_t,N_t-m_t)$.
\end{proof}

\suppsubnote{Structure of the candidate path}{note:proofs-path}

\begin{lemma}[Nestedness and $\calT$-measurability]\label{supp-lem:nested}
Let $B_j$, $\calT_j$, $u_i$, $h_i$, $k(q)$, $Q_{j,q}$ and $\pi_{j,q}$ be as in Supplementary Note~\ref{supp-note:construction}. (i) For every block $j$, every $q\le q'$ in $\mathcal Q\cup\{1\}$ and every gene $i$, $S_i(\lambda_j)\subseteq S_i(\pi_{j,q})\subseteq S_i(\pi_{j,q'})\subseteq S_i(\lambda_{j+1})$; consequently the GoCo path satisfies $S_i(\pi_k)\subseteq S_i(\pi_{k+1})$ for all $i,k$. (ii) The sequence $\pi_1,\ldots,\pi_K$, including $K$, is a function of $\calT$ and frozen objects only, and $K\le J+(J-1)|\mathcal Q|$.
\end{lemma}
\begin{proof}
(i) Because $\lambda_j>\lambda_{j+1}$, $S_i^{H}=S_i(\lambda_j)\subseteq S_i(\lambda_{j+1})=S_i^{L}$. For $i\notin B_j$ all four sets coincide. For $i\in B_j$, $S_i(\pi_{j,q})\in\{S_i^{H},S_i^{L}\}$ with $S_i(\pi_{j,q})=S_i^{L}$ exactly when $i\in Q_{j,q}$, so it suffices to show $Q_{j,q}\subseteq Q_{j,q'}$. The lexicographic order on $(u_i,h_i,\operatorname{id}(i))$ is a deterministic total order on $B_j$. For $q\le q'<1$, $k(q)=\max\{1,\lfloor q|\calT_j|\rfloor\}\le k(q')$, so the $k(q)$-th smallest ranking-fold triple precedes or equals the $k(q')$-th, whence $Q_{j,q}\subseteq Q_{j,q'}$. If $q'=1$, then $Q_{j,1}=B_j$ by definition, so the inclusion also holds. Consecutive elements of the path are $(\lambda_j,\pi_{j,q_{\min}})$, $(\pi_{j,q},\pi_{j,q'})$ with $q<q'<1$, $(\pi_{j,q_{\max}},\pi_{j,1})$ or $(\lambda_j,\lambda_{j+1})$; because $\pi_{j,1}=\lambda_{j+1}$, the chain of inclusions gives the claim.

(ii) The grid and the sets $S_i^{H},S_i^{L}$, and hence $c_i$, $D_i$, $B_j$ and the source-attribution weights $a_{ir}$, are functions of frozen predictor objects. The true quantities $\Delta_i$ and $\phi_i$ are label-dependent; in constructing the path they are evaluated only for ranking-fold genes $i\in\mathcal T_j$. The block means and the source estimates of Equation~\eqref{supp-eq:sourcerisk} therefore use only ranking-fold labels. The per-call model behind Equation~\eqref{supp-eq:utility} is fitted on ranking-fold calls alone, and its features are frozen predictor output together with rate estimates built from ranking-fold labels, so the fitted $\widehat p_{ig}$ are $\calT$-measurable; hence so are $\widehat\Delta_i$ and $\widehat\tau_i$ for a pool gene, which sum those probabilities over that gene's frozen added calls $D_i$, and so is $u_i=\widehat\Delta_i/\max(\widehat\tau_i,10^{-3})$; $h_i$ is a fixed function of the gene identifier. For $q<1$, the cutoff $(u^\ast,h^\ast,\operatorname{id}^\ast)$ is an order statistic among ranking-fold genes and $Q_{j,q}$ compares each affected gene's fixed-given-$\calT$ triple with it; the endpoint $Q_{j,1}=B_j$ is fixed by definition. The insertion triggers compare $|\calT|^{-1}\sum_{i\in\calT}L_i(\lambda_{j+1})$ with $0.8\alpha$ and $|\calT_j|$ with $5$. Thus every decision shaping the sequence is determined by $\calT$; each of the $J$ thresholds is visited once and each of the $J-1$ blocks contributes at most $|\mathcal Q|$ policies. If $\calT_j$ is empty the trigger fails and no adaptive policy is inserted, so Equation~\eqref{supp-eq:sourcerisk} is never evaluated on an empty fold.
\end{proof}
Nestedness of the released sets does not imply monotonicity of the losses: $L_i(\pi_{k+1})<L_i(\pi_k)$ is possible when added terms are supported. Only the total number of released calls is monotone along the path; the validity argument below uses neither property.

\suppsubnote{Consistency of the variance estimator}{note:proofs-variance}

\begin{lemma}[Consistency of $\shat_k$]\label{supp-lem:sigma}
Condition on $\calT$ and fix a policy $\pi$ with pool values $b_i=L_i(\pi)$, pool variance $\sigma^2(\pi)$ and sample variance $\shat^2=m^{-1}\sum_{i\in\calC}b_i^2-\bar b_{\calC}^2$. If $m\to\infty$ then $\shat^2-\sigma^2(\pi)\to0$ in probability. If in addition $\liminf\sigma^2(\pi)>0$, then $\shat/S(\pi)\to1$ in probability.
\end{lemma}
\begin{proof}
Apply Lemma~\ref{supp-lem:moments} to $b_i$ and to $b_i^2$, both in $[0,1]$: each sample mean has variance at most $1/(2m)\to0$, so by Chebyshev's inequality $\bar b_{\calC}\to\bar b_{\calP}$ and $m^{-1}\sum_{\calC}b_i^2\to N^{-1}\sum_{\calP}b_i^2$ in probability; hence $\shat^2\to\sigma^2(\pi)$. If $\liminf\sigma^2(\pi)>0$, the continuous-mapping theorem gives $\shat/\sigma(\pi)\to1$, and since $S(\pi)/\sigma(\pi)=\sqrt{N/(N-1)}\to1$, it follows that $\shat/S(\pi)\to1$ in probability.
\end{proof}

\suppsubnote{Proof of main-text Proposition~\ref{prop:validity}}{note:proofs-proposition}
\begin{proof}
Condition on $\calT$ throughout. By Lemma~\ref{supp-lem:nested}(ii) the sequence $\pi_1,\ldots,\pi_K$ is fixed and by Lemma~\ref{supp-lem:design} $\calC$ is a SRSWOR of size $m$ from $\calP$. The rule returns $\widehat\pi=\pi_{\hat k}$ with $\hat k=\max\{k:p_1\le\delta,\ldots,p_k\le\delta\}$, or abstains (risk $0\le\alpha$) if $p_1>\delta$.

\emph{Step 1: reduction to the first invalid candidate.} Let $\mathcal K_0=\{k:R(\pi_k)>\alpha\}$. If $\mathcal K_0=\emptyset$ there is nothing to prove. Otherwise let $k^\ast=\min\mathcal K_0$, a function of $\calT$. On $\{R(\widehat\pi)>\alpha\}$ a policy was certified and $\hat k\in\mathcal K_0$, so $\hat k\ge k^\ast$ and, by definition of $\hat k$, $p_k\le\delta$ for all $k\le\hat k$, in particular $p_{k^\ast}\le\delta$. Hence $\Pr\{R(\widehat\pi)>\alpha\mid\calT\}\le\Pr\{p_{k^\ast}\le\delta\mid\calT\}$. This step uses only the stopping rule; the joint distribution of $(p_1,\ldots,p_K)$, all computed from the same fold, plays no role, and neither does nestedness.

\emph{Step 2: asymptotic level at $k^\ast$.} Write $k=k^\ast$, $b_i=L_i(\pi_k)$ and $S_k^2=S^2(\pi_k)$. Since $\Phi$ is increasing, $p_k\le\delta$ if and only if
\[
\sqrt m(\Rhat_k-\alpha)\le z_\delta\tilde\sigma_k.
\]
Because $R(\pi_k)>\alpha$, on the event $\{p_k\le\delta\}$,
\[
\sqrt m\{\Rhat_k-R(\pi_k)\}
<
\sqrt m(\Rhat_k-\alpha)
\le z_\delta\tilde\sigma_k.
\]
Since $\delta\le\tfrac12$, $z_\delta\le0$, and since $\tilde\sigma_k=\max(\shat_k,10^{-6})\ge\shat_k$,
\[
z_\delta\tilde\sigma_k\le z_\delta\shat_k.
\]
Thus, on $\{\shat_k>0\}$,
\[
\{p_k\le\delta\}
\subseteq
\left\{
\frac{\sqrt m\{\Rhat_k-R(\pi_k)\}}{\shat_k}\le z_\delta
\right\}.
\]
The hypotheses $m\to\infty$, $N-m\to\infty$ and $\liminf S_k^2>0$ imply $m(1-f_t)\sigma_k^2\to\infty$, where $f_t=m/N$, so Lemma~\ref{supp-lem:fpclt} gives
\[
W_k=
\frac{\Rhat_k-R(\pi_k)}
{\sqrt{(1-f_t)S_k^2/m}}
\xrightarrow{\ d\ }\mathcal N(0,1).
\]
By Lemma~\ref{supp-lem:sigma}, $\shat_k/S_k\to1$ in probability, and therefore, by Slutsky's theorem,
\[
\frac{1}{\sqrt{1-f_t}}
\frac{\sqrt m\{\Rhat_k-R(\pi_k)\}}{\shat_k}
=
W_k\frac{S_k}{\shat_k}
\xrightarrow{\ d\ }\mathcal N(0,1).
\]
Moreover, Lemma~\ref{supp-lem:sigma} and the variance assumption imply $\Pr(\shat_k=0\mid\calT)\to0$. Hence, with $f_t\to f$,
\[
\limsup\Pr(p_k\le\delta\mid\calT)
\le
\Phi\!\Big(\frac{z_\delta}{\sqrt{1-f}}\Big).
\]
For $\delta\le\tfrac12$, $z_\delta\le0$ and $\sqrt{1-f}\le1$, so $\Phi(z_\delta/\sqrt{1-f})\le\Phi(z_\delta)=\delta$, with equality iff $z_\delta=0$ ($\delta=\tfrac12$) or $f=0$. Combining with Step 1 proves the conditional bound. The $o(1)$ is in $m$: it collects the normal-approximation error for $W_{k^\ast}$ and the Slutsky remainder for a single, $\calT$-determined index; the path length $K$ enters nowhere.

\emph{Numerical value.} When the variance floor is asymptotically inactive, $\tilde\sigma_k/S_k\to1$, so the ratio of the scale used by Equation~\eqref{supp-eq:pvalue} to the without-replacement standard error approaches $(1-f)^{-1/2}\approx1.73$ at $f=2/3$, making the normal test conservative. At $\delta=0.10$, main-text Proposition~\ref{prop:validity} therefore gives the asymptotic upper bound $\Phi(-1.2816/\sqrt{1/3})=\Phi(-2.22)\approx0.0132$; at $\delta=0.50$ the formal upper bound is $0.50$. These are upper bounds, not claims about the actual failure probability of the selected procedure.

\emph{Step 3: ordering-free validity (Corollary~\ref{supp-cor:order}).} Steps 1--2 used three properties of the candidate sequence: it is fixed given $\calT$; $K$ is finite; and the pool variance is bounded below at the first invalid candidate. The first holds for any rule that orders affected genes by a statistic computed from $\calT$, frozen objects and randomisation independent of the split (hash only; mean added-call score; $\widehat\Delta_i/c_i^{\,e}$; Equation~\eqref{supp-eq:utility}), with or without the insertion triggers, which are themselves $\calT$-measurable. Nothing in the proof refers to the association between the ordering and $\Delta_i$, which therefore affects only which index $\hat k$ is reached.

\emph{Step 4: unconditional statement.} Steps 1--2 give $\Pr\{R(\widehat\pi)>\alpha\mid\calT\}\le\Phi(z_\delta/\sqrt{1-f})+\varepsilon_m(\calT)$ with $\varepsilon_m(\calT)\to0$. If the pool variance at the first invalid candidate is bounded below uniformly in $\calT$, a Berry--Esseen bound for sampling without replacement (H\"oglund 1978), whose remainder depends on the population only through moment ratios that are bounded for values in $[0,1]$ with variance bounded below, makes $\varepsilon_m$ uniform in $\calT$, and $\Pr\{R(\widehat\pi)>\alpha\}=\mathbb E[\Pr\{R(\widehat\pi)>\alpha\mid\calT\}]\le\Phi(z_\delta/\sqrt{1-f})+o(1)$.

\emph{Global methods.} In the module-family analyses, for GoDag and the global baselines that use Equation~\eqref{supp-eq:pvalue}, the candidate sequence is the deterministic grid, the calibration sample $\calT\cup\calC$ is a SRSWOR of size $n_{\calT}+m$ from $\calU$ (Lemma~\ref{supp-lem:design}) and the estimand is the panel mean $n^{-1}\sum_{i\in\calU}L_i(\pi)$; Steps 1--2 apply verbatim with $\calP$ replaced by $\calU$ and $f$ by $0.7$. In the neighbourhood-family analyses, the deterministic global paths are certified on $\calC$, so Steps 1--2 apply directly with pool $\calP$ and $f=m/N$. Network-HAC variance substitution is not covered by this argument for $\delta<\tfrac12$; at $\delta=\tfrac12$ a normal-type rule with positive scale reduces to the same empirical-mean check. Empirical Bernstein and Hoeffding-type bounds use different tests, so this cancellation does not apply to them.

\emph{Finite-sample alternative.} Hoeffding's inequality holds for the mean of a SRSWOR of bounded values (Hoeffding 1963, Theorem~4), so $p_k^{\mathrm H}=\exp\{-2m(\alpha-\Rhat_k)_+^2\}$ is exactly super-uniform under $H_k$ conditional on $\calT$, and Step 1 gives a finite-sample bound $\delta$ for the same path at the cost of a wider margin ($\sqrt{\log(1/\delta)/(2m)}\approx0.011$ at $\delta=0.10$, $m=9{,}477$, versus $0.0038$ for Equation~\eqref{supp-eq:pvalue}).
\end{proof}

\suppsubnote{Proof of Remark~\ref{supp-rem:delta}}{note:proofs-remark}
\begin{proof}
\emph{$\delta$ identity.} Since $\Phi$ is strictly increasing and $\tilde\sigma_k>0$, $p_k\le\delta\iff\sqrt m(\Rhat_k-\alpha)/\tilde\sigma_k\le z_\delta\iff\Rhat_k\le\alpha+z_\delta\tilde\sigma_k/\sqrt m$. At $\delta=0.50$, $z_\delta=0$. At $\delta=0.10$, $z_\delta=-1.2816$; with $\tilde\sigma_k=0.29$ and $m=9{,}477$, $|z_\delta|\tilde\sigma_k/\sqrt m=0.0038$. Since $\Rhat_k$ depends only on $\calC$ and the candidate, any two methods that share a path, a loss and a certification fold and differ only in a normal-type variance estimator certify the same policies at $\delta=0.50$.

\emph{Split identity.} By Lemma~\ref{supp-lem:moments} with $b_i=L_i(\pi)$, $NR(\pi)=m\Rhat_{\calC}(\pi)+n_{\calE}\Rhat_{\calE}(\pi)$ for every policy, including $\pi=\widehat\pi$ (the identity is algebraic and holds whether or not the policy was selected using $\calC$). Subtracting $N\alpha=m\alpha+n_{\calE}\alpha$ and dividing by $n_{\calE}$ gives $\Rhat_{\calE}(\pi)-\alpha=(N/n_{\calE})\{R(\pi)-\alpha\}-(m/n_{\calE})\{\Rhat_{\calC}(\pi)-\alpha\}$. In the $10/60/30$ design $N/n_{\calE}=3.000$ and $m/n_{\calE}=2.000$ up to rounding (Wainberg: $14{,}215/4{,}738$ and $9{,}477/4{,}738$); these are also the factors for the neighbourhood-family global methods, which certify on $\calC$. For module-family global methods, with $n$ in place of $N$ and $n_{\calT}+m$ in place of $m$, the factors are $10/3$ and $7/3$.

\emph{Exceedance at $\delta=0.50$.} If a policy is certified at $\delta=0.50$ then $\Rhat_{\calC}(\widehat\pi)\le\alpha$, so the second term is nonnegative and $\Rhat_{\calE}(\widehat\pi)-\alpha\ge(N/n_{\calE})\{R(\widehat\pi)-\alpha\}$; hence $\{R(\widehat\pi)>\alpha\}\subseteq\{\Rhat_{\calE}(\widehat\pi)>\alpha\}$ and the held-out exceedance frequency bounds the pool-level failure frequency from above. The converse fails: $\Rhat_{\calE}(\widehat\pi)>\alpha$ also occurs when $R(\widehat\pi)\le\alpha$ but $\Rhat_{\calC}(\widehat\pi)$ is sufficiently far below $\alpha$. Under SRSWOR the standard deviation of $\Rhat_{\calC}(\pi)$ for a fixed policy with $\sigma(\pi)=0.29$ is $\sqrt{(1-f)/m}\,\sigma\approx0.0017$ in Wainberg; because the fixed-sequence rule stops at the first candidate whose certification mean exceeds $\alpha$, the selected candidate typically has $\Rhat_{\calC}(\widehat\pi)$ a few thousandths below $\alpha$ (a selection effect), $R(\widehat\pi)$ within a few thousandths of $\alpha$, and the evaluation-fold mean above $\alpha$ in a substantial fraction of splits.

\emph{Exceedance at $\delta=0.10$.} More generally, certification gives $\Rhat_{\calC}(\widehat\pi)-\alpha\le z_\delta\tilde\sigma/\sqrt m$, and the exact split identity therefore gives \[ \Rhat_{\calE}(\widehat\pi)-\alpha \ge \frac{N}{n_{\calE}}\{R(\widehat\pi)-\alpha\} -\frac{m}{n_{\calE}}\frac{z_\delta\tilde\sigma}{\sqrt m}. \] For illustration, with $\tilde\sigma=0.29$ and $m=9{,}477$, the certification margin is approximately $0.0038$, giving the corresponding approximations $0.0076$ and $0.0025$ under the Wainberg fold ratios. These numerical margins are illustrative and candidate-specific. A held-out exceedance does not, in general, imply $R(\widehat\pi)>\alpha$; conversely, for a certified policy with $\delta\le0.50$, $R(\widehat\pi)>\alpha$ implies $\Rhat_{\calE}(\widehat\pi)>\alpha$.
\end{proof}

\suppsubnote{Proof of Lemma~\ref{supp-lem:power}}{note:proofs-lemma}
\begin{proof}
Condition on $\calT$, fix a block and write $B=B_j\cap\calP$, $M=|B|$.

\emph{(a) Granularity.} For a policy that relaxes exactly the set $Q\subseteq B$, $L_i(\pi)=L_i(S_i^{H})+\Delta_i$ for $i\in Q$ and $L_i(\pi)=L_i(S_i^{H})$ otherwise, so $R(\pi)=R(S^{H})+N^{-1}\sum_{i\in Q}\Delta_i$ by linearity of the pool mean; with $Q=\{i:\operatorname{rk}(i)\le k\}$ this is the first claim. If $Q^{\mathrm{rand}}$ is a uniformly random $k$-subset of $B$, each $i\in B$ belongs to it with probability $k/M$, so $\mathbb E\sum_{i\in Q^{\mathrm{rand}}}\Delta_i=k\bar\Delta_B=(k/M)\sum_{i\in B}\Delta_i$, and $\sum_{i\in B}\Delta_i=N\{R(S^{L})-R(S^{H})\}$ because genes outside $B$ have $\Delta_i=0$. For the global threshold, $S_i(\lambda)=S_i^{H}\cup\{g\in D_i:s_{ig}\ge\lambda\}$ for $\lambda\in[\lambda_{j+1},\lambda_j]$, so the released sets change only when $\lambda$ crosses one of the $r$ distinct scores in the block, giving at most $r-1$ policies strictly between the endpoints; if all added rows share one score $s^\ast$, then $S_i(\lambda)=S_i^{H}$ for every $i\in B$ when $\lambda>s^\ast$ and $S_i(\lambda)=S_i^{L}$ for every $i\in B$ when $\lambda\le s^\ast$, so no global threshold separates affected genes. In Wainberg 2,470 of the 2,484 rows of the block $800\to775$ share the score $789.65$.

\emph{(b) Ordering value.} With $Q(k)=\{i:\operatorname{rk}(i)\le k\}$, $\sum_{i\in Q(k)}\Delta_i-\mathbb E\sum_{i\in Q^{\mathrm{rand}}}\Delta_i=\sum_{i\in B}(\Delta_i-\bar\Delta_B)\mathbf 1\{\operatorname{rk}(i)\le k\}=M\operatorname{Cov}_B\{\mathbf 1(\operatorname{rk}(i)\le k),\Delta_i\}$, where $\operatorname{Cov}_B(x,y)=M^{-1}\sum_{i\in B}(x_i-\bar x_B)(y_i-\bar y_B)=M^{-1}\sum_{i\in B}(x_i-\bar x_B)y_i$; dividing by $N$ gives the first identity. Summing over $k=1,\ldots,M$ and using $\sum_{k=1}^M\mathbf 1\{\operatorname{rk}(i)\le k\}=M+1-\operatorname{rk}(i)$,
\begin{align*}
&\frac1M\sum_{k=1}^{M}\Big\{\sum_{i\in Q(k)}\Delta_i-\mathbb E\sum_{i\in Q^{\mathrm{rand}}_{(k)}}\Delta_i\Big\}\\
&\qquad=\frac1M\sum_{i\in B}(\Delta_i-\bar\Delta_B)\{M+1-\operatorname{rk}(i)\}\\
&\qquad=-\operatorname{Cov}_B\{\operatorname{rk}(i),\Delta_i\},
\end{align*}
so the average risk saving of the ordering is $N^{-1}\operatorname{Cov}_B\{\operatorname{rk}(i),\Delta_i\}$, positive exactly when late-admitted genes tend to have larger increments. For a uniformly random ordering independent of the increments, $\mathbb E\operatorname{rk}(i)=(M+1)/2$ for every $i$, so the expected covariance is zero. If $\Delta_i$ is nondecreasing in $\operatorname{rk}(i)$, then for every $k$ the first $k$ genes have the $k$ smallest increments, so $\sum_{i\in Q(k)}\Delta_i$ is minimal among all $k$-subsets and in particular at most its random expectation (Chebyshev's sum inequality). Ties at the cutoff are resolved by the fixed gene hash to provide a deterministic reproducible total order. The random-order expectation above is a theoretical benchmark and does not require or imply that this fixed hash is probabilistically independent of $\Delta$.
\end{proof}

\begin{remark}[Interpretation and empirical counterparts]\label{supp-rem:interp}
Lemma~\ref{supp-lem:power} is a statement about fixed numbers given $\calT$. It explains how partial admission and an informative ordering can reduce fixed-pool risk at a given admission level, but it does not guarantee a larger certified prefix or an optimal yield because certification also depends on sampling variation and first-failure stopping. Part (a) describes granularity that a global threshold can forgo, especially within an exact tie, and part (b) expresses the ordering contribution through a rank--increment covariance on unseen genes. Whether that covariance is favourable for the ranking-fold-trained $u_i$ on pool genes is an empirical question; the out-of-sample Spearman correlation between $u_i$ and $\Delta_i$ over affected pool genes is reported per certified block in Table~\ref{supp-tab:blockdiag}. GoCo-M and GoCo-N are two prespecified candidate-path constructions whose relative performance is assessed empirically.
\end{remark}

\suppsubnote{The GoDag baseline}{note:godag}
GoDag is the hierarchy-aware global-threshold comparator used throughout: its candidate path is a deterministic score grid and its loss is the gene-level FDP of Equation~\eqref{eq:fdp} with the TruePath acceptable set. In the module-family analyses it calibrates on the first 70\% of each permutation, whereas in the neighbourhood-family analyses it certifies on $\calC$ (60\%) to match the other methods. The GoDag rows of the main text use the normal $p$-value of Equation~\eqref{supp-eq:pvalue}; replacing it by the corresponding Hoeffding bound with the appropriate calibration-sample size gives an exact finite-sample certificate at a cost in yield.

\suppsubnote{References for this note}{note:proofs-refs}
Cochran WG. \textit{Sampling Techniques}, 3rd ed. Wiley; 1977.

Erd\H{o}s P, R\'enyi A. On the central limit theorem for samples from a finite population. \textit{Publ Math Inst Hungar Acad Sci}. 1959;4:49--61.

H\'ajek J. Limiting distributions in simple random sampling from a finite population. \textit{Publ Math Inst Hungar Acad Sci}. 1960;5:361--374.

H\"oglund T. Sampling from a finite population: a remainder term estimate. \textit{Scand J Stat}. 1978;5:69--71.

Hoeffding W. Probability inequalities for sums of bounded random variables. \textit{J Am Stat Assoc}. 1963;58:13--30.
 
\clearpage

\suppnote{Construction and implementation details}{note:methods}

This note collects the construction and implementation details of the main text: the split mechanics (Section~\ref{supp-note:splits}), the interpretation of the frozen TruePath reference (Section~\ref{supp-note:reference}), the source-level estimate and the per-call ranking model (Section~\ref{supp-note:ranking-model}), the derivation arithmetic behind Remark~\ref{supp-rem:delta} (Section~\ref{supp-note:split-identity}), the uniform grid path and the choice of ranking statistic (Section~\ref{supp-note:grid-path}), the worked micro-example of a block (Section~\ref{supp-note:block-example}), further conventions and definitions (Section~\ref{supp-note:conventions}), source-object generalization (Section~\ref{supp-note:source-objects}), and the gene-level case study of the released annotations (Section~\ref{supp-note:case-study}). Equations and results are cross-referenced to their current locations.

\suppsubnote{Splits: seeds, rounding and realised fold sizes}{note:splits}
Each of the 100 analysis streams draws a uniformly random permutation of the panel (seed $20260905+s$ for stream $s$ in the module family; the neighbourhood family of Supplementary Note~\ref{supp-note:neighbourhood} uses seed $20260910+s$ with the same cut) and cuts it into a \emph{ranking fold} $\calT$ (the first $n_{\calT}=\operatorname{round}(0.10\,n)$ genes; its labels build the ordering and candidate path), a \emph{certification fold} $\calC$ (the next $m=\operatorname{round}(0.70\,n)-n_{\calT}$ genes, about 60\%) and an \emph{evaluation fold} $\calE$ (the remaining $n_{\calE}$ genes, about 30\%; rounding is the implementation's round-half-to-even). The realised sizes $(n_{\calT},m,n_{\calE})$ are $(1{,}579,\,9{,}477,\,4{,}738)$ in Wainberg, $(1{,}488,\,8{,}932,\,4{,}465)$ in Sanger, $(628,\,3{,}770,\,1{,}885)$ in DRIVE and $(1{,}502,\,9{,}012,\,4{,}506)$ in HAP1. The pool of non-ranking genes is $\calP=\calC\cup\calE$ with $N=n-n_{\calT}$ and $f=m/N\approx2/3$. GoCo uses $\calT$ only to build its candidate path and $\calC$ only to certify it. In the module family, global methods including GoDag have no ranking step and calibrate on $\calT\cup\calC$ (70\%); in the neighbourhood family, all compared methods certify on the same $\calC$ (60\%) fold. Evaluation-fold labels are used by no method before evaluation; the upstream predictor is frozen and treated as a constant.

\suppsubnote{Interpretation of the frozen TruePath reference}{note:reference}
The main analysis uses the frozen TruePath reference throughout. A released term is counted as \emph{unsupported} when it is absent from that reference for the gene after TruePath closure. This is an operational evaluation label, not a claim that the biological function is absent. Because GO is open-world and incomplete, the biological interpretation of reference-unsupported calls is treated separately in the Discussion; all methods here are evaluated against the same frozen reference.

\suppsubnote{The per-call ranking model}{note:ranking-model}

\paragraph{Where the estimate comes from.} The structure that creates the ties also supplies the evidence for the estimate. A call $(i,g)$ exists because the annotated members $\mathcal R_{ig}$ of gene $i$'s winning module carry term $g$; genes sharing module evidence can be offered overlapping terms by overlapping sets of members. If those members generated unsupported terms for ranking-fold genes, they are likely to do so again for $i$. Let $\widehat\rho_r$ be the fraction of ranking-fold calls supported by member $r$ whose term is unsupported, shrunk toward the overall ranking-fold rate as in Equation~\eqref{supp-eq:sourcerisk} with the weight of two pseudo-calls, and let $\widehat\rho_{ig}$ be the mean of $\widehat\rho_r$ over $r\in\mathcal R_{ig}$. This is the co-essentiality signal GoCo uses; Equation~\eqref{supp-eq:sourcerisk} could equally be applied to gene-level increments through the attribution matrix $a_{ir}$; any such variant is $\calT$-measurable and so covered by Corollary~\ref{supp-cor:order}.

\paragraph{Features.} GoCo fits a logistic regression of the indicator ``call unsupported'' on seven per-call features, over all ranking-fold calls at or above the lowest grid threshold. Three are the co-essentiality signal and its companions, learned from ranking-fold labels only: $\widehat\rho_{ig}$; the shrunk unsupported rate of the term $g$ itself across ranking-fold calls; and the log frequency of $g$ among the TruePath annotations of ranking-fold genes, a proxy for how specific the term is. Four are frozen predictor output: the log score, the number of supporting members $|\mathcal R_{ig}|$, the number of calls the gene has above the grid floor, and the module size excluding the target (Wainberg) or the target--module correlation (Sanger, DRIVE) or the cluster $z$-score (HAP1). We deliberately exclude the predictor's background count: in Wainberg it is the term's annotation count \emph{excluding the target}, so it would encode the target's own label.

\paragraph{Cross-fitting and hyperparameters.} For ranking-fold calls, the classifier fits and the raw term- and source-rate counts are cross-fitted between two hash-defined halves. The module implementation uses a common overall unsupported-rate shrinkage target computed from the whole ranking fold for both halves, and the term-frequency feature is computed from whole-ranking-fold TruePath counts with the scored gene removed. Thus the ranking-fold predictions are cross-fitted at the classifier and raw-rate-count levels but are not fully held-half-exclusive for every label-derived summary. Pool calls are scored by the model fitted on the whole ranking fold. Writing $\widehat p_{ig}$ for the fitted probabilities, define $\widehat L_i(S)=|S|^{-1}\sum_{g\in S}\widehat p_{ig}$ for $|S|>0$ and $\widehat L_i(\varnothing)=0$. The estimates entering Equation~\eqref{supp-eq:utility} are $\widehat\Delta_i=\widehat L_i(S_i^{L})-\widehat L_i(S_i^{H})$ and $\widehat\tau_i=\sum_{g\in D_i}(1-\widehat p_{ig})$. The feature set, the model class and its hyperparameters ($C=0.5$ on standardised features) were fixed on ten Wainberg development splits (streams $0$--$9$) before the external analyses; the neighbourhood-family instantiation was fixed on FunMap streams $0$--$9$. The model-class check replaces the logistic model by a histogram gradient-boosting classifier (150 trees, depth 3, learning rate 0.06). All label-derived quantities remain functions of $\calT$ alone: no certification- or evaluation-fold labels enter the ordering, so Corollary~\ref{supp-cor:order} applies.

\suppsubnote{Remark~\ref{supp-rem:delta}: the split identity and held-out exceedance}{note:split-identity}
The split identity gives, for every policy,
\[
\Rhat_{\calE}(\pi)-\alpha=\frac{N}{n_{\calE}}\{R(\pi)-\alpha\}-\frac{m}{n_{\calE}}\{\Rhat_{\calC}(\pi)-\alpha\},
\]
with $(N/n_{\calE},m/n_{\calE})=(3,2)$ for GoCo and for all methods in the neighbourhood-family matched-certification analysis, and $(10/3,7/3)$ for module-family global methods, up to rounding. At $\delta=0.50$ the certified policy is the last one before the first candidate whose certification-fold mean exceeds $\alpha$; on a fine path $\Rhat_{\calC}(\widehat\pi)$ sits just below $\alpha$, $R(\widehat\pi)$ lies within about one standard error ($\shat\sqrt{(1-f)/m}\approx0.0017$) of $\alpha$, and the evaluation-fold mean overshoots $R(\widehat\pi)$ by twice the certification-fold undershoot. Held-out exceedance of $\alpha$ in a substantial fraction of GoCo splits at $\delta=0.50$ is therefore expected---19--42\% across the four data sets (Table~\ref{supp-tab:detail})---and it is an upper bound on the pool-level failure frequency: if $\Rhat_{\calC}(\widehat\pi)\le\alpha$ then $R(\widehat\pi)>\alpha$ implies $\Rhat_{\calE}(\widehat\pi)>\alpha$. Because every label in this benchmark is known, $R(\widehat\pi)$ is computable in each split, and the fraction of splits with $R(\widehat\pi)>\alpha$ is the direct empirical check of main-text Proposition~\ref{prop:validity}. GoDag on the Wainberg grid stops at 800 with a calibration mean $\Rhat_{\calT\cup\calC}\approx0.092$, far below $\alpha$, which explains both its 1\% exceedance and its unused budget; on the finer external grids GoDag's held-out exceedance (36\% and 28\% in Sanger and DRIVE) is comparable to GoCo's. Section~\ref{supp-note:proofs-remark} proves these statements.

\suppsubnote{The uniform grid path and the choice of ranking statistic}{note:grid-path}
Two design choices define the construction certified here. First, the candidate path is defined uniformly: the same grid as GoDag with partial admission inserted in every step that meets a ranking-fold trigger, rather than a hand-picked block. This makes later thresholds available after inserted policies and helps explain the observed gains, but does not guarantee a larger certified prefix: a construction restricted to a single block could not move beyond it when a larger target permits further release. Second, the ordering statistic is the predicted loss increment per predicted supported term, used as a cost-per-gain heuristic for prioritising candidates expected to add less gene-level FDP per supported term; its form and the model behind it were fixed on the development data and applied unchanged externally.

\suppsubnote{A worked micro-example of one block}{note:block-example}
The Wainberg grid step from $\lambda_j=800$ to $\lambda_{j+1}=775$ is the block that GoCo splits at $\alpha=0.10$. It contains 2,484 gene--GO rows, of which 2,470 share the score $789.650$ exactly and the remaining 14 are spread over the scores $797.626$, $783.124$ and $779.901$; the affected set $B_j$ contains 1,267 genes, of which about 127 fall in the ranking fold. The dominant 2,470-call exact tie at score $789.650$ cannot be divided by any global threshold; the remaining 14 rows provide only three additional score locations within the step. The full-panel TruePath risk moves from $0.0924$ at threshold 800 to $0.1398$ at 775. GoCo instead evaluates, for each admission fraction $q$ in the grid $\mathcal Q$, the policy that gives $S_i^{L}$ to the $q$-quantile of affected genes with the smallest $u_i$ of Equation~\eqref{supp-eq:utility} and $S_i^{H}$ to the rest. In the full-panel linear interpolation of Lemma~\ref{supp-lem:power}(a), the risk reaches $\alpha=0.10$ at about $q=0.16$; this is not itself a certification calculation. The GoCo ordering bends the path below that interpolation and was certified at $q=0.35$ or $q=0.40$ in 89 of the 100 splits ($q=0.40$ in 64 of them). In split~0 the certified policy admitted 40\% of the affected genes.

\paragraph{Where the ties come from, in detail.} The enrichment score of a Wainberg call is a function of three small integers: the number $k_{\mathrm{sup}}$ of annotated module members supporting the term, the module size $n_{\mathrm{mod}}$ excluding the target, and the background count $K_{\mathrm{bg}}$. Because 85\% of scored calls come from modules with $n_{\mathrm{mod}}\le5$, many (module, term) pairs share a count triple and therefore a score. The 2,470 calls at $789.65$ belong to 1,256 genes, 482 terms and 347 different modules, and 2,018 of them have $(k_{\mathrm{sup}},n_{\mathrm{mod}},K_{\mathrm{bg}})=(1,4,5)$. The external resources tie for the same reason: Sanger has 1,880 pairs at one value among the 29,120 pairs at or above score 200, spread over 847 genes and 393 modules; DRIVE 723 of 9,258; HAP1 14,245 of 122,262. The Sanger and DRIVE grid steps used at $\alpha=0.10$ contain 20 and 36 distinct values, so a finer global grid would recover part of the lost granularity there, whereas the HAP1 step $475\to450$ is again almost a single tie: 2,089 of its 2,115 calls share one value. Full multiplicities are in \texttt{tie\_counts.csv} in the reproducibility package.

\suppsubnote{Further conventions and definitions}{note:conventions}
\emph{What the estimand does and does not control.} The certified risk $R(\pi)$ averages the per-gene FDP of Equation~\eqref{eq:fdp} over all pool genes, with abstaining genes contributing zero. The loss penalizes reference-unsupported terms among those released but does not penalize reference-supported terms that are withheld; reference-supported yield therefore measures retained supported functions separately from the error rate; it is not a false-negative rate. In the module-structure analyses, only about 8--12\% of genes receive any term at $\alpha=0.10$, so the average is dominated by abstentions and constrains no individual annotated gene: because abstention carries no loss, a policy that releases entirely unsupported sets for 9\% of pool genes and abstains on the rest has $R=0.09$ and is admissible at $\alpha=0.10$. Nor does $R(\pi)$ bound the pooled fraction of unsupported gene--GO calls, which weights genes by set size. GoCo does not control this pooled quantity. Table~\ref{supp-tab:estimands} includes a pair-based conformal-selection comparator targeting pooled call-level FDR as a descriptive comparison; because our random split samples genes rather than individual gene--GO pairs, the usual pair-level exchangeability conditions are not established here and we make no formal pooled-call FDR guarantee for these benchmark runs. We certify $R(\pi)$ and report, for every method, the number of released calls, the number of annotated genes, the reference-supported term count and the supported terms per supported gene as descriptive yield diagnostics (Table~\ref{supp-tab:riskdecomp}).

\emph{The three qualifications of main-text Proposition~\ref{prop:validity}.} The result is asymptotic in $m$ ($m=3{,}770$--$9{,}477$ here); the path length is irrelevant because a single test matters. It controls the pool risk $R(\widehat\pi)$, not the evaluation-fold mean, which is itself a noisy estimate of $R(\widehat\pi)$ (Remark~\ref{supp-rem:delta}). At $\delta=0.50$ the formal bound is $\tfrac12$ and the normal certificate reduces to the empirical calibration check; we use this setting as the primary benchmark, while $\delta=0.10$ provides the stricter certification analysis.

\emph{The target-excluded enrichment background.} One property of the Wainberg predictor deserves mention: its enrichment background excludes the target gene, so a gene that is itself annotated with a term receives a slightly larger background-corrected score for that term (by a factor $K/(K-1)$ for a term with $K$ annotated panel genes). This is a feature of the published predictor, enters the frozen scores of every method identically, and is not used by any ordering beyond the score itself; we excluded the raw background count from the learned features for this reason. For GoCo-M, the grid, insertion triggers, admission grid, $\kappa$, ordering family and learned model specification were fixed on Wainberg; the neighbourhood-family instantiation and its model specification were fixed on FunMap streams $0$--$9$. These choices are listed in Table~\ref{supp-tab:constants} and none affects validity.

\emph{Score comparison and reporting.} Scores are compared after rounding to six decimals. Besides the means and exceedance frequencies of the main text we report, for each paired comparison, the fraction of splits in which one method exceeds another.

\emph{Source sets.} The sources entering Equation~\eqref{supp-eq:sourcefrac} are restricted to the gene panel, and a gene whose added calls had no retained source would fall back to the block mean of Equation~\eqref{supp-eq:sourcerisk}. In all four data sets every added call has at least one retained source, so $\sum_r a_{ir}=1$ for every affected gene in every certified block and this fallback is never used (Table~\ref{supp-tab:blockdiag}).

\emph{Sign of the increment.} A gene moving out of abstention has $S_i^{H}=\emptyset$ and $\Delta_i$ equal to the false discovery proportion of its added calls, which is nonnegative; $\Delta_i<0$ occurs only when $S_i^{H}$ is nonempty and the added calls are better supported than the existing set.

\emph{Compatibility of the gene-averaged target.} Because abstention costs nothing, a gene-averaged target of $0.10$ is compatible with a low supported fraction of the released call list, which is why the released-call and supported-term counts are reported beside the certified risk.

\emph{What the ordering has to estimate.} Admitting gene $i$ changes pool risk by $\Delta_i/N$ and gains $\tau_i=|D_i\cap A_i^{\TP}|$ reference-supported terms. We use $\Delta_i/\tau_i$ as a cost-per-gain ordering heuristic, prioritising genes expected to add less gene-level FDP per supported term; this ordering is not claimed to solve the discrete allocation problem exactly. Neither quantity can be evaluated for a gene the certificate will later use because both depend on that gene's labels. The ordering therefore estimates this ratio from the ranking fold and frozen predictor output, and Corollary~\ref{supp-cor:order} ensures that the quality of that estimate affects yield rather than validity.

\emph{Evaluation and reporting.} For every method and split we record, on the evaluation fold, the mean gene-level TruePath FDP (a held-out assessment of the selected policy), the number of genes with a nonempty set, the number of released calls, the supported fraction of calls, the supported-gene yield (evaluation-fold genes with at least one TruePath-supported released term) and the GO-term yield (supported evaluation-fold gene--GO calls); on the certification fold, the certification mean; and the pool risk $R(\widehat\pi)$ of the certified policy, which every label in the benchmark makes computable and which main-text Proposition~\ref{prop:validity} controls asymptotically. The pool risk is the mean loss over the calibration sample together with its corresponding evaluation fold: the whole panel for module-family global methods, and $\calC\cup\calE$ for GoCo and for all methods in the neighbourhood-family matched-certification comparison. For descriptive comparison the tables report this quantity using the TruePath loss for every method. Multilabel certifies the Direct loss instead, so for it the event formally controlled by its certificate is the Direct-loss pool risk exceeding $\alpha$; those frequencies are reported separately in the footnote of Table~\ref{supp-tab:detail}. Because the Direct acceptable set is contained in its TruePath closure, the Direct loss is at least as large as the TruePath loss for the same release. Reported are means over the 100 splits, the fraction of splits with $\Rhat_{\calE}>\alpha$ (held-out exceedance) and with $R(\widehat\pi)>\alpha$ (pool-level failure, the event bounded by main-text Proposition~\ref{prop:validity}), and paired split-wise differences in GO-term yield with the fraction of splits in which one method exceeds another. Splits are overlapping random partitions of one fixed panel, and every reported quantity concerns comparisons of release rules on that panel rather than a biological superpopulation. Variation across the 100 split streams reflects the randomisation design and Monte Carlo error; an exceedance or failure frequency $f$ estimated from 100 splits has Monte Carlo standard error $\sqrt{f(1-f)/100}\le0.05$, so frequencies of $0.00$--$0.01$ are consistent with, but cannot resolve, the asymptotic boundary values of main-text Proposition~\ref{prop:validity}. The most liberal candidate on every path is the predictor's own reporting floor---score $600$ in Wainberg, $100$ in the three external resources and $-\log_{10}q=1$ in FunMap and STRING---and it is worth recording when that floor binds, since a comparison in which the global baseline runs out of grid is not a comparison of methods. Over the 100 splits it binds only in DRIVE at $\alpha=0.20$, where at $\delta=0.50$ every method reaches the end of the grid in every split (which is why the difference there is exactly zero; at $\delta=0.10$ GoCo-M stops one step short in one split), and in one Sanger split at $\alpha=0.20$, $\delta=0.50$. Elsewhere the certified cutoff sits well above it: $625$--$650$ in Wainberg at $\alpha=0.20$, and $16$--$62$ against a floor of $1$ in FunMap and STRING, so extending either grid further down---as a coarser published grid would---changes no arm and no comparison.

\emph{Comparators.} Every method uses the same 100 split streams, frozen scores and held-out TruePath truth within each data set. \emph{Multilabel} is our re-implementation of the released multilabel statistic of Boger et al.\ (one-sided normal $p$-value with empirical standard deviation, $\delta=0.50$ default) on the Direct loss; it calibrates on $\calT\cup\calC$ in the module family and on $\calC$ in the neighbourhood family, and ignores the ontology. \emph{GoDag} applies the same certificate to the TruePath loss and selects a single global threshold, using the same family-specific calibration sample. The generic Direct-loss calibrators are reported in the corresponding main-text comparison tables. \emph{GoCo-M} adds predictor-native shared-evidence structure through the candidate path of Supplementary Note~\ref{supp-note:construction}, whereas \emph{GoCo-N} adds predictor-native neighbourhood evidence through the call-resolution path of the same section (Supplementary Note~\ref{supp-note:neighbourhood}).

\emph{Run time.} The reference implementation is vectorised over the certification fold. On a laptop CPU, all methods for 100 splits at three error targets and two values of $\delta$ took 3 minutes for Wainberg, 13 minutes for Sanger and 5 minutes each for DRIVE and HAP1, dominated by the sparse source-attribution products.

\emph{Out-of-sample informativeness of the ordering.} Across the certified blocks, GoCo's statistic $u_i$ is rank-correlated with the realised increment $\Delta_i$ on affected pool genes at Spearman $0.22$--$0.73$, against $-0.08$--$0.43$ for the raw score of the added calls; the per-block values are in Table~\ref{supp-tab:blockdiag}.

\suppsubnote{Source objects beyond co-essentiality}{note:source-objects}
GoCo is defined for any predictor that exposes a pair $(s_{ig},\mathcal R_{ig})$: a frozen score and, for each candidate call, the evidence objects that produced it. The distinction relevant to the candidate path is whether that evidence is shared across genes or assembled separately for each gene. Shared evidence can create exact tied blocks addressed by GoCo-M, whereas gene-specific evidence permits call-level ordering as in GoCo-N. Table~\ref{supp-tab:structures} summarizes the source-object structures considered here, and Supplementary Note~\ref{supp-note:neighbourhood} develops the neighbourhood family on FunMap with validation on STRING. Every source object must be target-excluded: allowing the focal gene's own annotations to contribute to the evidence for the call later evaluated would inflate retrospective support. Corollary~\ref{supp-cor:order} ensures that the quality of the resulting ordering affects yield rather than validity.

\suppsubnote{Gene-level case study of the released annotations}{note:case-study}
In Wainberg split~0 GoCo-M admitted 40\% of the affected genes and released 351 evaluation-fold gene--GO calls that GoDag did not, over 158 genes, of which 37 (10.5\%) are supported by the frozen reference. Table~\ref{supp-tab:casestudy} lists the eight co-essential modules that contributed the most of them; the 28 that contributed at least four are all in the shipped split-level output. The interleukin-6 signal-transduction module \{IL6ST, JAK1, LIFR, SLC39A1, STAT3\} transfers correctly: IL6ST and JAK1 each acquire interleukin-27- and interleukin-35-mediated signalling (GO:0070106, GO:0070757) from the other members of the receptor complex, and the reference supports all four calls. In the TNF-receptor-associated module \{BIRC2, CYLD, N4BP1, TRAF2, TRAF3\}, BIRC2 acquires the CD40 receptor complex (GO:0035631) from TRAF2 and TRAF3 and regulation of the necroptotic process (GO:0060544) from CYLD, both supported. In the anaphase-promoting complex module \{ANAPC5, ANAPC13, ANAPC15, UBE2C, UBE2S\}, which contributes five admitted calls and so falls below the table's cut, three are supported, including GO:0005680 for ANAPC15 and UBE2C. Where a module mixes a regulator with its targets the transfer fails: CDKN1A, in \{CDKN1A, TP53, TP53BP1, UBE2K, USP28\}, acquires seven terms that belong to TP53 as a transcription factor---core promoter sequence-specific DNA binding, transcription cofactor binding, regulation of thymocyte apoptosis---and the reference supports none of them; DDX21 and NCL likewise each receive four nucleolar terms from NOLC1. Distinguishing these two regimes is exactly what the ordering is estimating, and it does so imperfectly but informatively (Table~\ref{supp-tab:blockdiag}). The examples also show that a reference-unsupported call is not necessarily biologically false: N4BP1 acquires regulation of Toll-like receptor signalling (GO:0034121) from BIRC2, which the frozen reference does not carry, although N4BP1 is an established negative regulator of TLR-induced NF-$\kappa$B signalling (Shi et al.\ 2021). Core promoter binding for CDKN1A, by contrast, illustrates an unsupported transfer of a regulator's function to its target; these retrospective data do not establish biological absence. Reference-supported yield is therefore a conservative measure of recovered biological support, while the measured gene-level FDP can overstate biological error relative to a more complete or context-specific annotation reference.

\suppsubnote{References for this note}{note:methods-refs}

Shi H, Sun L, Wang Y, et al. N4BP1 negatively regulates NF-$\kappa$B by binding and inhibiting NEMO oligomerization. \textit{Nat Commun}. 2021;12:1379.

Wainberg M, Kamber RA, Balsubramani A, et al. A genome-wide atlas of co-essential modules assigns function to uncharacterized genes. \textit{Nat Genet}. 2021;53:638--649.
 
\clearpage

\suppnote{The neighbourhood-structure instantiation, its validation, and the choice of candidate path}{note:neighbourhood}

This note reports the framework applied to a predictor of a different kind---a random-walk
neighbourhood on a functional gene network rather than a co-essential module---together with validation without retuning on a second, unrelated network. It also records the converse: the candidate
path that is best for this structure, GoCo-N, is \emph{worse} than GoCo-M on Wainberg, and the
decomposition of Lemma~\ref{supp-lem:power} says why the two structures call for different paths.

\suppsubnote{FunMap: the predictor and the frozen objects}{note:funmap}
FunMap is a machine-learned pan-cancer functional gene network derived from CPTAC proteogenomics.
We take the authors' released network and their own top-50 random-walk-with-restart (RWR)
neighbourhoods and retrain nothing. FunMap's edges were learned from CPTAC proteomic and transcriptomic
co-expression against a co-functionality gold standard drawn from Reactome, ProteomeHD, BioPlex, HuRI,
HI-Union and BioGRID; Gene Ontology was not among its training labels or its edge scores and was used
by its authors only to interpret modules after the fact. Because Reactome co-membership is correlated
with GO Biological Process co-annotation, we treat FunMap as the development predictor of this family
and use STRING for validation without retuning (Section~\ref{supp-note:string}). For a focal gene the wrapper removes that gene from its own
neighbourhood and removes its annotations from the enrichment background, then runs hypergeometric
over-representation over GO terms of size $[10,2000]$ with Benjamini--Hochberg adjustment within
aspect; the score of a call is $-\log_{10}$ of the adjusted $p$-value. The GO release is the
2024-09-08 archive, which predates the FunMap paper's acceptance; the same frozen tables are used for STRING. The panel is the 10,316
GO-annotated genes of the 10,525-gene network, giving 3,303,702 scored gene--GO calls.

The predictor therefore exposes exactly the two objects the construction needs
(Section~\ref{sec:estimand} of the main text): a frozen score $s_{ig}$, and for each call a source set
$\mathcal R_{ig}$, here the annotated neighbours carrying the term. One refinement is introduced by the structure. Wainberg's module members are unordered, so Equation~\eqref{supp-eq:sourcefrac} spreads each call's unit of evidence uniformly over them; a random-walk neighbourhood is \emph{ranked} by proximity, so the same attribution is weighted by $w(\text{rank})=1/\log_2(1+\text{rank})$. Equation~\eqref{supp-eq:sourcerisk} is used with this weighted attribution to construct the source-derived rate features; GoCo-N then orders individual calls by the fitted unsupported probabilities $\widehat p_{ig}$.

\paragraph{The per-call model.} A logistic regression of ``call reference-unsupported'' is fitted on ranking-fold calls and applied to every call. In the neighbourhood implementation, the two label-derived rate features, their shrinkage targets and the classifier coefficients are fitted on one hash-defined half of the ranking fold and used to score calls of the other half; every certification- or evaluation-fold call is scored by the model fitted on the whole ranking fold. This differs slightly from the module implementation described in Section~\ref{supp-note:ranking-model}, but both constructions use ranking-fold labels only. Splits of this family use seed $20260910+s$ (Section~\ref{supp-note:splits}). Its features are the neighbourhood analogues of the module features: from frozen predictor output, the log score, the number of annotated neighbours carrying the term, the number of annotated
neighbours, the call's BH rank within its gene and aspect, the number of sources, their total
proximity weight, the proximity rank of the closest source and the gene's number of calls; from
ranking-fold labels only, the proximity-weighted unsupported rate of the call's sources and the
term's unsupported rate, each shrunk by two pseudo-calls. The enrichment's background count for the
term is deliberately \emph{not} a feature. Because the focal gene is removed from the background,
that count equals the term's annotation count minus one exactly when the focal gene is itself
Direct-annotated with the term, so it would encode the label being predicted; this is the same
construction excluded from the Wainberg model for the same reason (Section~\ref{supp-note:conventions}). A draft of this analysis
that included it gave GoCo-N yields about 2--4\% higher than those reported here; every number in
this paper is from the model without it.

\suppsubnote{STRING: validation without retuning, and why this network and not another}{note:string}
Validation uses STRING v12.0 (human), restricted to the \emph{experiments} evidence channel. Three properties support its use as an external validation here.

\emph{Independent of the development predictor.} STRING v12.0 was frozen in July 2023, eighteen
months before FunMap's publication, and is built from physical-assay evidence (DIP, BioGRID, HPRD,
IntAct, MINT, PDB), not from CPTAC proteogenomics. Neither network is an input to the other.

\emph{Free of the annotations being evaluated.} This is the binding constraint, and it eliminates
most candidate networks. GO enters STRING through two named columns, \texttt{database} (which lists
Gene Ontology among its curated sources) and \texttt{textmining}; we read neither, using only
column~10. The experiments channel is calibrated against KEGG and Complex Portal rather than GO.
Selection cannot be GO-informed either: the minimum \texttt{combined\_score} over the 13,715,404
shipped rows is exactly 150 and \texttt{combined\_score} $\ge$ \texttt{experiments}, so no pair with
an experiments score above that floor can have been censored on GO-derived grounds. Networks whose
edge weights are \emph{fitted} to GO co-annotation---HumanNet, whose log-likelihood scores are
trained against a GO Biological Process gold standard; HumanBase/GIANT, whose posteriors use
co-annotation to propagated GO BP terms as the positive class; and GeneMANIA, whose integration
weights are regressed on GO BP co-annotation---were excluded for this reason, since evaluating GO
predictions against a network built from GO would be circular.

\emph{Identical downstream pipeline.} At STRING's published medium-confidence cut ($\ge400$) the
experiments channel gives 92,885 undirected edges over 13,079 genes. Everything after the edge list
is unchanged from FunMap: top-50 RWR neighbourhood, focal gene removed from the neighbourhood and
from the background, hypergeometric over-representation on terms of size $[10,2000]$, BH within
aspect. Two facts must be declared rather than assumed. First, STRING publishes no neighbourhoods,
so ours are reconstructed (restart $0.5$, 50 power iterations); the validation therefore exercises
the same predictor \emph{recipe} rather than a second frozen author artefact, a weaker claim than
FunMap's. Second, we hold the GO truth fixed at the same frozen 2024-09-08 tables used for FunMap and
restrict the analysis universe to the 8,496 genes those tables cover (the enrichment background is
this universe minus the focal gene), of which 8,055 have a neighbourhood inside it and form the
calibrated panel; the network is then the only thing that differs between the two applications.
This costs coverage relative to STRING's full 13,079 genes and buys exact comparability of the
reference. Our BH implementation reproduces FunMap's
published adjusted $p$-values on 97.6\% of 3.3 million calls to within $10^{-9}$, which fixes the
enrichment convention to the published one rather than to our reading of it.

\suppsubnote{Why the structure differs, and what Lemma~\ref{supp-lem:power} implies}{note:structure-diff}
Lemma~\ref{supp-lem:power} splits the value of partial admission inside a block into a \emph{granularity} term---the
ability to stop part-way through, worth something only because a global threshold must take the
block whole---and an \emph{ordering} term, the covariance
$(M/N)\,\mathrm{Cov}_B\{\mathbf 1(\mathrm{rk}(i)\le k),\Delta_i\}$, which is negative when
small-increment genes are admitted first. The two applications differ in which term is available.

In Wainberg the score is an enrichment statistic of three small integers, so 2,470 gene--GO pairs
share the value 789.650 and the largest tie in the operating range holds 3,361. An exact tie has no
interior threshold, so \emph{no} global path, however fine, reaches the granularity term. Partial admission is necessary to reach it; GoCo-M uses gene-level admission within grid steps, whereas GoCo-N defines a separate path that orders individual calls by the model.

In FunMap and STRING the adjusted $p$-values carry no material ties in the operating range, so refining the global grid recovers much of the granularity available from threshold refinement at no cost in validity---the grid is a function of frozen scores alone. GoCo-N then uses the remaining predictor-native information to construct an individual-call ordering. This is a prespecified candidate-path heuristic motivated by the decomposition in Lemma~\ref{supp-lem:power}, not an optimality result of that lemma. The candidate policies become the top-$\ell$ calls under the ranking-fold statistic---exactly $\ell$ calls, ties in $\widehat p_{ig}$ being broken by the frozen gene hash and then by call index---for $\ell$ on a fixed geometric grid of fractions of the call pool, nested in $\ell$ and measurable with respect to $\calT$, so Definition~\ref{supp-def:path} is satisfied and main-text Proposition~\ref{prop:validity} and Corollary~\ref{supp-cor:order} apply without modification. This is the instantiation \textbf{GoCo-N} of Supplementary Note~\ref{supp-note:construction}, fixed on FunMap; \textbf{GoCo-M} is the instantiation of Equation~\eqref{supp-eq:policy}, fixed on Wainberg. Both are admissible for every predictor; Section~\ref{supp-note:neighbourhood-comparison} compares them on FunMap and STRING, and Section~\ref{supp-note:both-instantiations} on the four shared-evidence data sets.

The two instantiations use different ranking resolutions. Equation~\eqref{supp-eq:utility} ranks genes within a GoCo-M block by predicted loss increment per predicted supported term. GoCo-N instead defines a separate call-level heuristic that orders individual calls by $\widehat p_{ig}$, with calls predicted most likely to be supported entering first. This probability ordering is not an algebraic specialisation of Equation~\eqref{supp-eq:utility}. The ranking model and certification framework are shared, but the candidate family and the resolution of the ordering differ.

\suppsubnote{What the comparison shows}{note:neighbourhood-comparison}
The neighbourhood-structure panel of the main text gives the primary comparison, and Table~\ref{supp-tab:frontier} reports gains over the exhaustive matched-risk frontier of the global path---for each arm and split, the largest correct-call yield of any global threshold, over every distinct score value of the frozen score file and abstention, whose realised held-out risk does not exceed the risk the arm realised---isolating candidate-path quality from simply spending more of the available risk budget. Three points follow.

First, refining the global grid increased supported-call yield by approximately 5.0\% at the primary $\alpha=0.10$, $\delta=0.50$ operating point in FunMap; at $\delta=0.10$ the corresponding increase was approximately 1.5\%.

Second, GoCo-M---the instantiation fixed on Wainberg---transfers to this structure but captures only part of the available gain.
Measured against the exhaustive matched-risk frontier at the risk each arm actually realised ($\delta=0.10$), GoCo-M gains +101, +196 and +164 supported calls at $\alpha=0.05,0.10,0.20$ in FunMap (ahead of the frontier in 100\%, 100\% and 100\% of splits) and +46, +138 and +197 in STRING (100\%, 100\%, 100\%).

Third, GoCo-N gains +155, +379 and +519 on the same frontier in FunMap (ahead in 95\%, 100\% and 100\% of splits) and +172, +268 and +561 in STRING (100\%, 97\%, 99\%), at realised risk comparable to GoCo-M's (lower in 3 of the six combinations), and it widens coverage substantially: 1,084 against 779 annotated genes in FunMap at $\alpha=0.10$. The pattern replicates on STRING at every target.

\suppsubnote{Both instantiations on the four dependency data sets}{note:both-instantiations}
GoCo-N was also run on Wainberg, Sanger, DRIVE and HAP1 with the per-call model of
Supplementary Note~\ref{supp-note:construction} unchanged---the same features and the same cross-fitting---so that only the
candidate path differed. At $\alpha=0.10$, $\delta=0.10$, GoCo-N released 142 fewer
reference-supported terms than GoCo-M in Wainberg and 69 fewer in DRIVE, but 219 more
in Sanger and 88 more in HAP1. Thus the prespecified structural mapping is a useful
default rather than an oracle choice: even among predictors with shared evidence and
exact ties, the better candidate path can vary. This motivates the ranking-fold selection
rule of Section~\ref{supp-note:adaptive}, which chooses between the two paths before certification without using
evaluation-fold outcomes.

\suppsubnote{An adaptive alternative to the structural mapping}{note:adaptive}
The main text fixes the instantiation by evidence structure, on one development predictor per
structure. Because both paths satisfy Definition~\ref{supp-def:path}, any rule that chooses between them from $\calT$
and frozen inputs alone inherits main-text Proposition~\ref{prop:validity} unchanged (Corollary~\ref{supp-cor:order}), so the choice can also be
made per split at no cost in validity. The rule we report walks each instantiation on $\calT$ alone---the GoCo-M path with its ranking-fold cutoffs, the GoCo-N path over ranking-fold calls ordered by $\widehat p_{ig}$---with the per-call model cross-fitted over ten hash-defined folds of $\calT$ using the family-specific feature construction described above, stops each at the last candidate whose ranking-fold mean loss is at most $\alpha$, and carries forward the instantiation with the larger ranking-fold supported-call yield.

Table~\ref{supp-tab:selection} reports its verdict for every predictor and target beside the instantiation that in fact won on the evaluation fold. At $\alpha=0.10$ it chose GoCo-N in 0\%, 96\%, 18\% and 89\% of splits on Wainberg, Sanger, DRIVE and HAP1 and in 92\% and 76\% on FunMap and STRING---agreeing with the structural mapping on the two development predictors and departing from it on Sanger and HAP1---and it agreed with the held-out winner in 100\%, 96\%, 80\%, 89\%, 83\% and 69\% of splits. Where the two instantiations are far apart the rule is essentially never wrong; where they are closer it errs in a fraction of splits (31\% on STRING at $\alpha=0.10$, where GoCo-N's held-out margin over GoCo-M is 13\%), and the price is paid in power alone: what it selected released 6.0--73.4\% more supported terms than the cutoff across the six predictors at $\alpha=0.10$, $\delta=0.10$, against 6.3--23.3\% for GoCo-M alone and -34.4--75.2\% for GoCo-N alone.

We report the structural mapping in the main text because it is prespecified, needs no per-data-set
decision and is what a user of a new predictor can apply immediately; the adaptive rule is available
to a user who would rather spend ranking-fold labels on the choice. Predicting from a predictor's
structure alone which instantiation wins, or by how much, remains open.

\suppsubnote{Measurability}{note:measurability}
The operational measurability test (\texttt{test\_measurability.py}) was repeated for both applications: scrambling the TruePath labels of every gene outside the ranking fold leaves the per-call ordering bit-identical, over 241,154 pool calls in FunMap and 470,735 in STRING, with maximum absolute difference $0$ in both, and likewise leaves the GoCo-M statistic $u_i$ of every affected pool gene in every block unchanged. Source sets are drawn from neighbours' annotations, never the focal gene's: the FunMap
neighbourhood file contains no self-loops in 499,261 rows, and the STRING neighbourhoods are
constructed with the focal gene removed. The scrambling test detects any use of certification- or
evaluation-fold labels through the truth tables; it cannot detect a label that a frozen predictor
column carries by construction, which is why the background count is excluded by inspection rather
than by test.
 
\suppnote{Candidate construction, formal guarantees and data reconstruction}{note:extended}
This note specifies the candidate construction and formal guarantees underlying the operational steps in the main Methods. Numbered citations refer to the main reference list.
\suppsubnote{Candidate construction}{note:construction}

\paragraph{Blocks and increments.}
On the grid $\lambda_1>\cdots>\lambda_J$ shared by GoDag and GoCo (Table~\ref{supp-tab:constants}), block $j$ steps from conservative sets $S_i^{H}=S_i(\lambda_j)$ to liberal sets $S_i^{L}=S_i(\lambda_{j+1})\supseteq S_i^{H}$, and for gene $i$
\begin{equation}
\Delta_i=L_i(S_i^{L})-L_i(S_i^{H}),\qquad c_i=|S_i^{L}|-|S_i^{H}|,\qquad D_i=S_i^{L}\setminus S_i^{H},
\label{supp-eq:increment}
\end{equation}
with $B_j=\{i:c_i>0\}$ the \emph{affected} genes, $\calT_j=B_j\cap\calT$ the affected ranking genes and $\phi_i=1-|D_i\cap A_i^{\TP}|/c_i$ the unsupported fraction of the added calls.

\paragraph{Fractional source attribution.}
For each added call $g\in D_i$ the predictor names the sources $\mathcal R_{ig}$ behind it---in Wainberg the winning module's annotated members, gene $i$ excluded---and each call spreads a unit of evidence equally over them:
\begin{equation}
a_{ir}=\frac{1}{c_i}\sum_{g\in D_i}\frac{\mathbf 1\{r\in\mathcal R_{ig}\}}{|\mathcal R_{ig}|},\qquad i\in B_j .
\label{supp-eq:sourcefrac}
\end{equation}
Every added call has a retained source in all four data sets (Supplementary Table~\ref{supp-tab:blockdiag}).

\paragraph{Source-level shrinkage.}
For $y_i=\Delta_i$ or $y_i=\phi_i$, $i\in\calT_j$, let $\bar y$ denote the mean over the affected ranking genes. The source estimate is the attribution-weighted mean over the ranking genes $r$ supports, shrunk toward $\bar y$ by one pseudo-gene,
\begin{equation}
\widehat y_r=\frac{\sum_{i\in\calT_j}a_{ir}\,y_i+\kappa\,\bar y}{\sum_{i\in\calT_j}a_{ir}+\kappa},\qquad\kappa=1,
\label{supp-eq:sourcerisk}
\end{equation}
and the prediction $\widehat y_i=\sum_r a_{ir}\widehat y_r$ for an affected gene depends, outside $\calT$, on losses only through ranking-fold genes: its own labels never enter.

\paragraph{The GoCo statistic.}
Admitting gene $i$ changes pool risk by $\Delta_i/N$ and gains $\tau_i=|D_i\cap A_i^{\TP}|$ supported terms. We therefore use $\Delta_i/\tau_i$ as a cost-per-gain ordering heuristic, prioritising genes expected to add less gene-level FDP per supported term; this ordering is not claimed to solve the discrete allocation problem exactly, and optimality is not required for validity. Because both quantities are unknown for pool genes, GoCo estimates per-call probabilities $\widehat p_{ig}$ of being unsupported from ranking-fold information and uses them to construct $\widehat\Delta_i$ and $\widehat\tau_i$. Its leading feature is the shrunk rate at which a call's sources $\mathcal R_{ig}$ produced unsupported ranking-fold terms; all seven features come from ranking-fold labels or frozen output, so the fitted probabilities are $\calT$-measurable (Supplementary Note~\ref{supp-note:ranking-model}). With $\widehat L_i(S)=|S|^{-1}\sum_{g\in S}\widehat p_{ig}$ for $|S|>0$ and $\widehat L_i(\varnothing)=0$, $\widehat\Delta_i=\widehat L_i(S_i^{L})-\widehat L_i(S_i^{H})$ and $\widehat\tau_i=\sum_{g\in D_i}(1-\widehat p_{ig})$, GoCo orders affected genes by
\begin{equation}
u_i=\frac{\widehat\Delta_i}{\max(\widehat\tau_i,\,10^{-3})},
\label{supp-eq:utility}
\end{equation}
the predicted loss increase per predicted supported term, so genes backed by reliable members and gaining several such terms go first. Any other $\calT$-measurable ordering could be substituted without touching the certificate (Corollary~\ref{supp-cor:order}).

\paragraph{Partial admission.}
Order $B_j$ lexicographically by $(u_i,h_i,\operatorname{id}(i))$, where $h_i\in[0,1)$ is a SHA-256 hash of the gene identifier and $\operatorname{id}(i)$ is the gene identifier used as a final deterministic tie-break. For $q\in\mathcal Q$ with $q<1$ (Table~\ref{supp-tab:constants}), let $(u^\ast,h^\ast,\operatorname{id}^\ast)$ be the $k(q)$-th smallest triple among the affected ranking genes, $k(q)=\max\{1,\lfloor q|\calT_j|\rfloor\}$, and set
\begin{equation}
\begin{aligned}
Q_{j,q}&=\{i\in B_j:(u_i,h_i,\operatorname{id}(i))\le(u^\ast,h^\ast,\operatorname{id}^\ast)\},\\
S_i(\pi_{j,q})&=
\begin{cases}
S_i^{L}, & i\in Q_{j,q},\\
S_i^{H}, & i\notin Q_{j,q}.
\end{cases}
\end{aligned}
\label{supp-eq:policy}
\end{equation}
For the endpoints we define $Q_{j,0}=\varnothing$, so $\pi_{j,0}$ is the global-threshold policy at $\lambda_j$ (already on the path, hence never inserted), and $Q_{j,1}=B_j$, so $\pi_{j,1}$ is exactly the next global-threshold policy at $\lambda_{j+1}$. The cutoff is learned on $\calT$ and applied to every affected gene; candidate admission uses this deterministic total order, so the $Q_{j,q}$ are nested, with $Q_{j,1}=B_j$ corresponding to $\lambda_{j+1}$ (Supplementary Note~\ref{supp-note:proofs-path}).

\begin{definition}[Candidate path]\label{supp-def:path}
A candidate path is a finite sequence of policies $\pi_1,\ldots,\pi_K$ that is (i) \emph{nested}, $S_i(\pi_k)\subseteq S_i(\pi_{k+1})$ for all $i$ and $k$, and (ii) \emph{$\calT$-measurable}: the sequence, including its length $K$, is a function of the ranking fold (its members and their labels), the frozen inputs, and any auxiliary randomisation fixed before the split. All conditional statements below also condition on that fixed randomisation. Labels of $\calC$ and $\calE$ are never used.
\end{definition}

GoDag's path is the label-free grid; GoCo-M's visits the same grid and, after $\lambda_j$, inserts $\pi_{j,q}$ for $q\in\mathcal Q$ in increasing order before $\lambda_{j+1}$ whenever the ranking-fold mean loss at $\lambda_{j+1}$ exceeds $0.8\alpha$ and $|\calT_j|\ge5$. Both triggers use $\calT$ only, so every such path satisfies Definition~\ref{supp-def:path} (Supplementary Lemma~\ref{supp-lem:nested}); the construction and its constants (Table~\ref{supp-tab:constants}) are identical across data sets.

GoCo-N instead orders individual calls by $\widehat p_{ig}$; its neighbourhood-specific features and the ranking-fold-only rule for choosing between paths are specified in Supplementary Note~\ref{supp-note:neighbourhood}. Both constructions satisfy Definition~\ref{supp-def:path}.

\suppsubnote{Certification rule and formal guarantees}{note:formal-guarantees}

For candidate $\pi_k$ define
\begin{align}
\Rhat_k &= \frac1m\sum_{i\in\calC}L_i(\pi_k),\qquad
\shat_k^2 = \frac1m\sum_{i\in\calC}\{L_i(\pi_k)-\Rhat_k\}^2,
\label{supp-eq:riskse}\\
p_k &= \Phi\!\left(\frac{\Rhat_k-\alpha}{\max(\shat_k,10^{-6})/\sqrt m}\right),
\label{supp-eq:pvalue}
\end{align}
the one-sided normal $p$-value of the released Multilabel implementation [24] (default $\delta=0.50$), used unchanged by GoDag and GoCo on the hierarchy-aware loss. In the module-family analyses, GoDag and the other global methods use $\calT\cup\calC$ and $n_{\calT}+m$ in place of $\calC$ and $m$; in the neighbourhood-family analyses, all compared methods certify on $\calC$. Fixed-sequence Learn-Then-Test [23] tests $p_1,p_2,\ldots$ in path order against $H_k:R(\pi_k)\ge\alpha$, stops at the first $k$ with $p_k>\delta$ and returns the preceding candidate $\widehat\pi$ (abstention, risk zero, never occurred). Here $\delta$ is the tolerated false-certification level in the asymptotic guarantee in main-text Proposition~\ref{prop:validity}; smaller $\delta$ requires stronger evidence before a candidate is certified. Given $\calT$ the path is fixed and $\calC$ a random sample of the pool, so every $\Rhat_k$ is an unbiased survey estimate.

The conditional fixed-sequence guarantee is stated in main-text Proposition~\ref{prop:validity}, with its proof in Supplementary Note~\ref{supp-note:proofs-proposition}. The argument also covers global normal-test methods on their respective certification folds. Replacing the normal test with Hoeffding's bound [34] gives finite-sample control on the same prespecified path.

\begin{corollary}[Validity does not depend on the ordering]\label{supp-cor:order}
Under the assumptions of main-text Proposition~\ref{prop:validity}, validity does not depend on how the $\calT$-measurable ordering is constructed. It therefore holds for any ordering of the affected genes, or of individual calls, that is a function of the ranking fold and frozen inputs---orderings based only on frozen inputs, such as a deterministic hash or the score of the added calls, Equation~\eqref{supp-eq:utility}, and the GoCo-N path---with or without the insertion triggers, and for a choice between such paths made on $\calT$. Predictor-native evidence changes which policies are reached before the first failure, not the error bound.
\end{corollary}

\begin{remark}[The $\delta$ identity and held-out exceedance]\label{supp-rem:delta}
Certification is equivalent to $\Rhat_k\le\alpha+z_\delta\max(\shat_k,10^{-6})/\sqrt m$. Thus, at $\delta=0.50$, normal-type procedures with a positive scale estimate reduce to $\Rhat_k\le\alpha$. Multilabel, IID-normal and network-HAC therefore select the same policy on a common path and loss; this cancellation does not apply to concentration-bound tests. Pool failure and held-out exceedance remain different events. Supplementary Notes~\ref{supp-note:proofs-remark} and~\ref{supp-note:split-identity} give the exact split identity, the rounding factors and the stricter $\delta=0.10$ calculation.
\end{remark}

Validity does not depend on the ordering; yield does, through two separable effects.
\begin{lemma}[Granularity and the value of an ordering]\label{supp-lem:power}
Fix $\calT$ and a block, let $B=B_j\cap\calP$ with $M=|B|\ge1$ and block mean $\bar\Delta_B=M^{-1}\sum_{i\in B}\Delta_i$. For a $\calT$-measurable strict order on $B$, let $\operatorname{rk}(i)\in\{1,\ldots,M\}$ be the position of gene $i$, let $\pi_{(k)}$ relax exactly the $k$ first-ranked genes, and let $\pi^{\mathrm{rand}}_{(k)}$ relax a uniformly random $k$-subset of $B$.
\begin{enumerate}
\item[(a)] \emph{Granularity.} $R(\pi_{(k)})=R(S^{H})+N^{-1}\sum_{\operatorname{rk}(i)\le k}\Delta_i$ and $\mathbb E\,R(\pi^{\mathrm{rand}}_{(k)})=R(S^{H})+(k/M)\{R(S^{L})-R(S^{H})\}$. Random partial admission interpolates linearly between the two global endpoints, whereas a global threshold realises at most $r-1$ intermediate policies when the block contains $r$ distinct score values, and none for an exact tie.
\item[(b)] \emph{Ordering value.} $R(\pi_{(k)})-\mathbb E\,R(\pi^{\mathrm{rand}}_{(k)})=(M/N)\operatorname{Cov}_B\{\mathbf 1(\operatorname{rk}(i)\le k),\Delta_i\}$, the empirical covariance over $B$. A useful ordering admits small-$\Delta$ genes first, making this covariance negative. Averaged over $k=1,\ldots,M$, the risk saving of the ordering relative to random admission equals $N^{-1}\operatorname{Cov}_B\{\operatorname{rk}(i),\Delta_i\}$; it is zero in expectation for a uniformly random ordering independent of the increments, and if $\Delta_i$ is nondecreasing in $\operatorname{rk}(i)$ the ordering is optimal among all orderings at every $k$.
\end{enumerate}
\end{lemma}

\noindent\textit{Proof idea.} (a) is linearity of the pool mean; (b) follows from $\sum_{k}\mathbf 1\{\operatorname{rk}(i)\le k\}=M+1-\operatorname{rk}(i)$ and Chebyshev's sum inequality (Supplementary Note~\ref{supp-note:proofs-lemma}).

\suppsubnote{Module reconstruction and frozen reference}{note:data-details}
Wainberg et al.\ built overlapping human co-essential modules from genome-scale CRISPR dependency profiles and assigned functions by leave-one-gene-out GO enrichment of the module [19]. Our reconstruction reproduces the archived Wainberg score file used in this analysis to numerical precision and names, for every row, the winning non-syntenic module and the annotated members carrying the signal: those members are the sources $\mathcal R_{ig}$ (panel 15,794 genes). The Direct reference and its TruePath closure were fixed before analysis (284,700 and 1,851,997 pairs over 19,578 genes; GO release of 1 November 2019, all evidence codes, NOT-qualified annotations removed, closure over \emph{is\_a} and \emph{part\_of}), and every method uses the same 100 split streams, score matrices and held-out truth. Sanger Project Score [35] and Project DRIVE [36] add CRISPR and RNAi evidence, choosing among the same Wainberg modules by correlation with the leave-target-out centroid---no GO information enters that choice---before the same target-excluded enrichment rule (14,885 and 6,283 eligible genes); HAP1 uses the published Global-L2 parent clusters of its genetic-interaction network [37] (15,020 genes). All three share the gene universe and GO reference with Wainberg; only the dependency evidence changes. The neighbourhood-structure family---FunMap random-walk neighbourhoods [38] and the experiments channel of STRING v12.0 [39]---is described in Supplementary Note~\ref{supp-note:neighbourhood}.
\suppnote{Declaration of artificial intelligence assistance}{note:ai}
\textit{Declaration of the use of artificial intelligence.} The authors used Claude Opus 5 (Anthropic) and ChatGPT GPT-5.6 Sol (OpenAI) as aids for language editing of author-written text, consistency checking, and assistance with code writing and documentation. All scientific decisions, interpretations and conclusions were made by the authors. All model-assisted suggestions and code were reviewed and verified by the authors, who take full responsibility for the manuscript.

\clearpage
\begin{table*}[htbp]
\centering
\caption{Operating characteristics of Multilabel, GoDag, GoCo-M, GoCo-N and $\mathcal T$-selected GoCo at $\alpha=0.10$ and $\delta\in\{0.50,0.10\}$, over 100 splits.}
\label{supp-tab:detail}
\scriptsize
\setlength{\tabcolsep}{3.2pt}
\resizebox{\linewidth}{!}{\begin{tabular}{lll rrr rrr r r}
\toprule
$\delta$ & Data set & Method & Gene FDP & $\Pr(\widehat R_{\mathcal E}>\alpha)$ & $\Pr(R>\alpha)$ & Genes w/ correct call & Calls & Correct calls & $\Delta$ correct calls vs GoDag & $\Pr(>\text{GoDag})$ \\
\midrule
0.50 & Wainberg & Multilabel (Direct loss) & 0.0923 & 0.01 & 0.00 & 185.7 & 1614 & 339.6 & +0.0 & -- \\
 &  & GoDag & 0.0923 & 0.01 & 0.00 & 185.7 & 1614 & 339.6 & -- & -- \\
 &  & \textbf{GoCo-M} & 0.0964 & 0.20 & 0.08 & 197.1 & 1955 & 369.1 & +29.4 & 1.00 \\
 &  & \textbf{GoCo-N} & 0.0944 & 0.16 & 0.07 & 129.3 & 1208 & 227.1 & $-$112.5 & 0.01 \\
 &  & \textbf{GoCo}, $\mathcal T$-selected & 0.0964 & 0.20 & 0.08 & 197.1 & 1955 & 369.1 & +29.4 & 1.00 \\
 & Sanger & Multilabel (Direct loss) & 0.0953 & 0.20 & 0.00 & 155.4 & 1919 & 333.0 & $-$20.8 & -- \\
 &  & GoDag & 0.0980 & 0.36 & 0.00 & 160.9 & 2040 & 353.8 & -- & -- \\
 &  & \textbf{GoCo-M} & 0.0980 & 0.42 & 0.20 & 166.9 & 2380 & 380.7 & +26.9 & 0.81 \\
 &  & \textbf{GoCo-N} & 0.0942 & 0.11 & 0.03 & 206.9 & 1905 & 572.2 & +218.4 & 1.00 \\
 &  & \textbf{GoCo}, $\mathcal T$-selected & 0.0943 & 0.11 & 0.04 & 205.7 & 1916 & 567.2 & +213.4 & 0.99 \\
 & DRIVE & Multilabel (Direct loss) & 0.0945 & 0.13 & 0.06 & 96.1 & 1834 & 248.1 & $-$10.5 & -- \\
 &  & GoDag & 0.0965 & 0.28 & 0.28 & 98.7 & 1954 & 258.6 & -- & -- \\
 &  & \textbf{GoCo-M} & 0.0971 & 0.33 & 0.26 & 105.7 & 2296 & 291.7 & +33.1 & 0.89 \\
 &  & \textbf{GoCo-N} & 0.0984 & 0.36 & 0.25 & 78.3 & 1417 & 212.0 & $-$46.6 & 0.12 \\
 &  & \textbf{GoCo}, $\mathcal T$-selected & 0.0974 & 0.34 & 0.25 & 100.6 & 2160 & 277.5 & +18.9 & 0.72 \\
 & HAP1 & Multilabel (Direct loss) & 0.0719 & 0.00 & 0.00 & 77.0 & 847 & 129.4 & +0.0 & -- \\
 &  & GoDag & 0.0719 & 0.00 & 0.00 & 77.0 & 847 & 129.4 & -- & -- \\
 &  & \textbf{GoCo-M} & 0.0943 & 0.19 & 0.05 & 80.8 & 1285 & 137.7 & +8.2 & 1.00 \\
 &  & \textbf{GoCo-N} & 0.0931 & 0.17 & 0.05 & 110.5 & 997 & 228.7 & +99.3 & 1.00 \\
 &  & \textbf{GoCo}, $\mathcal T$-selected & 0.0929 & 0.15 & 0.05 & 107.7 & 1029 & 220.9 & +91.5 & 1.00 \\
\midrule
0.10 & Wainberg & Multilabel (Direct loss) & 0.0923 & 0.01 & 0.00 & 185.3 & 1612 & 338.9 & $-$0.7 & -- \\
 &  & GoDag & 0.0923 & 0.01 & 0.00 & 185.7 & 1614 & 339.6 & -- & -- \\
 &  & \textbf{GoCo-M} & 0.0948 & 0.13 & 0.00 & 195.2 & 1896 & 365.2 & +25.6 & 0.98 \\
 &  & \textbf{GoCo-N} & 0.0899 & 0.05 & 0.00 & 127.0 & 1150 & 222.8 & $-$116.8 & 0.00 \\
 &  & \textbf{GoCo}, $\mathcal T$-selected & 0.0948 & 0.13 & 0.00 & 195.2 & 1896 & 365.2 & +25.6 & 0.98 \\
 & Sanger & Multilabel (Direct loss) & 0.0913 & 0.01 & 0.00 & 150.7 & 1751 & 315.8 & $-$3.3 & -- \\
 &  & GoDag & 0.0923 & 0.12 & 0.00 & 151.7 & 1789 & 319.1 & -- & -- \\
 &  & \textbf{GoCo-M} & 0.0930 & 0.15 & 0.00 & 156.9 & 1983 & 340.0 & +20.9 & 0.95 \\
 &  & \textbf{GoCo-N} & 0.0910 & 0.04 & 0.00 & 202.7 & 1820 & 559.2 & +240.1 & 1.00 \\
 &  & \textbf{GoCo}, $\mathcal T$-selected & 0.0910 & 0.04 & 0.00 & 201.3 & 1824 & 553.4 & +234.3 & 1.00 \\
 & DRIVE & Multilabel (Direct loss) & 0.0913 & 0.13 & 0.00 & 90.0 & 1609 & 227.9 & $-$9.3 & -- \\
 &  & GoDag & 0.0929 & 0.13 & 0.00 & 92.9 & 1718 & 237.2 & -- & -- \\
 &  & \textbf{GoCo-M} & 0.0931 & 0.16 & 0.01 & 99.4 & 1987 & 263.5 & +26.2 & 0.89 \\
 &  & \textbf{GoCo-N} & 0.0923 & 0.18 & 0.01 & 73.3 & 1145 & 194.5 & $-$42.8 & 0.17 \\
 &  & \textbf{GoCo}, $\mathcal T$-selected & 0.0931 & 0.17 & 0.01 & 94.5 & 1850 & 251.5 & +14.3 & 0.75 \\
 & HAP1 & Multilabel (Direct loss) & 0.0719 & 0.00 & 0.00 & 77.0 & 847 & 129.4 & +0.0 & -- \\
 &  & GoDag & 0.0719 & 0.00 & 0.00 & 77.0 & 847 & 129.4 & -- & -- \\
 &  & \textbf{GoCo-M} & 0.0915 & 0.06 & 0.00 & 80.7 & 1252 & 137.6 & +8.1 & 1.00 \\
 &  & \textbf{GoCo-N} & 0.0895 & 0.05 & 0.00 & 108.4 & 965 & 225.1 & +95.7 & 1.00 \\
 &  & \textbf{GoCo}, $\mathcal T$-selected & 0.0893 & 0.04 & 0.00 & 105.7 & 996 & 217.5 & +88.0 & 1.00 \\
\bottomrule
\end{tabular}}
\par\vspace{3pt}\begin{minipage}{\linewidth}\footnotesize Means over the same 100 split streams at $\alpha=0.10$. Gene FDP: held-out gene-level TruePath FDP $\widehat R_{\mathcal E}$. $\Pr(\widehat R_{\mathcal E}>\alpha)$: fraction of splits whose held-out FDP exceeded $\alpha$. $\Pr(R>\alpha)$: fraction of splits in which the TruePath risk of the certified policy over the pool of its own certificate (the whole panel for Multilabel and GoDag, $\mathcal C\cup\mathcal E$ for GoCo) exceeded $\alpha$; for GoDag and GoCo this is the event main-text Proposition~\ref{prop:validity} bounds. Multilabel certifies the Direct loss instead: on that loss its pool failure frequencies were 0.00/0.13/0.06/0.00 (Wainberg/Sanger/DRIVE/HAP1) at $\delta=0.50$ and 0.00/0.00/0.00/0.00 at $\delta=0.10$, and a policy certified on the stricter Direct loss is conservative on TruePath. Correct calls: released gene--GO pairs the frozen reference confirms; genes w/ correct call: held-out genes receiving at least one of them; Calls: released calls of any kind. Multilabel is the released multilabel calibration of Boger et al.\ on the Direct loss, calibrated on 70\% of the panel; GoDag and GoCo-M use the TruePath loss and share the 25-unit grid; GoCo-N walks the top-$\ell$ calls of the ranking-fold model; GoCo, $\mathcal T$-selected is whichever of the two the ranking-fold rule chose in that split (Table~\ref{supp-tab:selection}). $\Delta$ correct calls: paired mean difference in correct calls from GoDag; $\Pr(>\text{GoDag})$: fraction of splits in which the method released more correct calls than GoDag.\end{minipage}
\end{table*}
 \begin{table*}[htbp]
\centering
\caption{Error-target sweep at $\delta=0.10$: mean held-out gene-level TruePath FDP and reference-supported calls at $\alpha\in\{0.05,0.10,0.20\}$ over 100 splits.}
\label{supp-tab:sweep10}
\scriptsize
\begin{tabular}{ll rr rr rr}
\toprule
Data set & Method & \multicolumn{2}{c}{$\alpha=0.05$} & \multicolumn{2}{c}{$\alpha=0.10$} & \multicolumn{2}{c}{$\alpha=0.20$} \\
 & & FDP & Correct & FDP & Correct & FDP & Correct \\
 & & & calls & & calls & & calls \\ \midrule
\multicolumn{8}{l}{\emph{Module structure: GoCo-M (development: Wainberg)}} \\
Wainberg & Multilabel (Direct loss) & 0.0359 & 183.1 & 0.0923 & 338.9 & 0.1699 & 518.9 \\
 & GoDag & 0.0359 & 183.1 & 0.0923 & 339.6 & 0.1699 & 518.9 \\
 & \textbf{GoCo-M} & 0.0438 & 213.2 & 0.0948 & 365.2 & 0.1882 & 539.2 \\
\addlinespace[2pt]
Sanger & Multilabel (Direct loss) & 0.0450 & 198.2 & 0.0913 & 315.8 & 0.1871 & 1136.0 \\
 & GoDag & 0.0461 & 203.9 & 0.0923 & 319.1 & 0.1895 & 1163.2 \\
 & \textbf{GoCo-M} & 0.0463 & 213.7 & 0.0930 & 340.0 & 0.1907 & 1346.8 \\
\addlinespace[2pt]
DRIVE & Multilabel (Direct loss) & 0.0413 & 90.0 & 0.0913 & 227.9 & 0.1823 & 691.4 \\
 & GoDag & 0.0414 & 90.1 & 0.0929 & 237.2 & 0.1844 & 702.4 \\
 & \textbf{GoCo-M} & 0.0436 & 94.2 & 0.0931 & 263.5 & 0.1841 & 702.0 \\
\addlinespace[2pt]
HAP1 & Multilabel (Direct loss) & 0.0436 & 96.0 & 0.0719 & 129.4 & 0.1353 & 160.2 \\
 & GoDag & 0.0436 & 96.1 & 0.0719 & 129.4 & 0.1353 & 160.2 \\
 & \textbf{GoCo-M} & 0.0450 & 102.4 & 0.0915 & 137.6 & 0.1756 & 187.7 \\
\midrule
\multicolumn{8}{l}{\emph{Neighbourhood structure: GoCo-N (development: FunMap)}} \\
FunMap & Multilabel (Direct loss, shared grid) & 0.0312 & 727.9 & 0.0583 & 1004.3 & 0.1543 & 1869.5 \\
 & GoDag, shared grid & 0.0312 & 727.9 & 0.0901 & 1261.2 & 0.1547 & 1873.0 \\
 & GoDag, refined grid & 0.0448 & 869.7 & 0.0941 & 1279.8 & 0.1900 & 2164.4 \\
 & \textbf{GoCo-N} & 0.0431 & 1016.1 & 0.0878 & 1615.8 & 0.1823 & 2615.8 \\
\addlinespace[2pt]
STRING & Multilabel (Direct loss, shared grid) & 0.0385 & 423.2 & 0.0822 & 815.8 & 0.1458 & 1383.8 \\
 & GoDag, shared grid & 0.0386 & 424.9 & 0.0828 & 828.1 & 0.1464 & 1387.9 \\
 & GoDag, refined grid & 0.0454 & 473.1 & 0.0928 & 909.2 & 0.1867 & 1711.7 \\
 & \textbf{GoCo-N} & 0.0429 & 621.4 & 0.0894 & 1151.1 & 0.1830 & 2243.0 \\
\bottomrule
\end{tabular}
\par\vspace{3pt}\begin{minipage}{\linewidth}\footnotesize FDP: mean held-out gene-level TruePath false discovery proportion. Correct calls: released gene--GO pairs the frozen reference confirms. Each structure is shown with the instantiation prespecified for it. Module-family GoCo-M and GoDag share the 25-unit grid. In the neighbourhood family, Multilabel and GoDag are shown on the shared coarse grid, with refined-grid GoDag separately; GoCo-N uses a distinct call-ordered path. Grid refinement does not change the fixed-sequence multiplicity argument.\end{minipage}
\end{table*}
 \begin{table*}[htbp]
\centering
\caption{Certification, pool and held-out risks with yield diagnostics for GoDag and GoCo-M at $\alpha=0.10$ and $\delta\in\{0.50,0.10\}$, averaged over 100 splits. GoCo-M certifies on $\mathcal C$ (60\%); GoDag calibrates on $\mathcal T\cup\mathcal C$ (70\%).}
\label{supp-tab:riskdecomp}
\scriptsize
\resizebox{\linewidth}{!}{\begin{tabular}{l l l rrr rr rrrr}
\toprule
$\delta$ & Data set & Method & Calib.\ mean & $R$ & $\widehat R_{\mathcal E}$ & $\Pr(R>\alpha)$ & $\Pr(\widehat R_{\mathcal E}>\alpha)$ & Genes w/ calls & Calls & Correct calls & Terms per supp.\ gene \\ \midrule
0.50 & Wainberg & GoDag & 0.0924 & 0.0924 & 0.0923 & 0.00 & 0.01 & 561 & 1614 & 339.6 & 1.83 \\
0.50 & Wainberg & GoCo-M & 0.0964 & 0.0964 & 0.0964 & 0.08 & 0.20 & 575 & 1955 & 369.1 & 1.87 \\
0.50 & Sanger & GoDag & 0.0984 & 0.0983 & 0.0980 & 0.00 & 0.36 & 530 & 2040 & 353.8 & 2.20 \\
0.50 & Sanger & GoCo-M & 0.0986 & 0.0984 & 0.0980 & 0.20 & 0.42 & 529 & 2380 & 380.7 & 2.28 \\
0.50 & DRIVE & GoDag & 0.0955 & 0.0958 & 0.0965 & 0.28 & 0.28 & 215 & 1954 & 258.6 & 2.62 \\
0.50 & DRIVE & GoCo-M & 0.0962 & 0.0965 & 0.0971 & 0.26 & 0.33 & 215 & 2296 & 291.7 & 2.76 \\
0.50 & HAP1 & GoDag & 0.0720 & 0.0720 & 0.0719 & 0.00 & 0.00 & 377 & 847 & 129.4 & 1.68 \\
0.50 & HAP1 & GoCo-M & 0.0939 & 0.0941 & 0.0943 & 0.05 & 0.19 & 475 & 1285 & 137.7 & 1.70 \\
\midrule
0.10 & Wainberg & GoDag & 0.0924 & 0.0924 & 0.0923 & 0.00 & 0.01 & 561 & 1614 & 339.6 & 1.83 \\
0.10 & Wainberg & GoCo-M & 0.0946 & 0.0947 & 0.0948 & 0.00 & 0.13 & 568 & 1896 & 365.2 & 1.87 \\
0.10 & Sanger & GoDag & 0.0925 & 0.0924 & 0.0923 & 0.00 & 0.12 & 502 & 1789 & 319.1 & 2.10 \\
0.10 & Sanger & GoCo-M & 0.0933 & 0.0932 & 0.0930 & 0.00 & 0.15 & 505 & 1983 & 340.0 & 2.17 \\
0.10 & DRIVE & GoDag & 0.0918 & 0.0921 & 0.0929 & 0.00 & 0.13 & 208 & 1718 & 237.2 & 2.55 \\
0.10 & DRIVE & GoCo-M & 0.0921 & 0.0925 & 0.0931 & 0.01 & 0.16 & 208 & 1987 & 263.5 & 2.64 \\
0.10 & HAP1 & GoDag & 0.0720 & 0.0720 & 0.0719 & 0.00 & 0.00 & 377 & 847 & 129.4 & 1.68 \\
0.10 & HAP1 & GoCo-M & 0.0912 & 0.0913 & 0.0915 & 0.00 & 0.06 & 462 & 1252 & 137.6 & 1.70 \\
\bottomrule
\end{tabular}}
\end{table*} \begin{table*}[!htbp]
\centering
\caption{Block diagnostics at $\alpha=0.10$ for grid steps in which GoCo-M policies were certified, averaged over 100 splits. Spearman correlations compare the realised loss increment $\Delta_i$ with the GoCo ordering statistic $u_i$ and with the negative mean added-call score.}
\label{supp-tab:blockdiag}
\scriptsize
\resizebox{\linewidth}{!}{\begin{tabular}{l l r r r r r r r r}
\toprule
Data set & Step & Affected & Ranking & mean $\Delta$ & sd $\Delta$ & full mass & $\rho(u,\Delta)$ & $\rho(-\text{score},\Delta)$ & $\Pr(\rho_u>0)$ \\ \midrule
Wainberg & $800\to 775$ & 1267 & 127 & 0.591 & 0.484 & 1.00 & 0.725 & -0.077 & 1.00 \\
Sanger & $450\to 425$ & 454 & 45 & 0.203 & 0.408 & 1.00 & 0.560 & 0.041 & 1.00 \\
Sanger & $425\to 400$ & 211 & 21 & 0.097 & 0.316 & 1.00 & 0.325 & 0.146 & 1.00 \\
Sanger & $400\to 375$ & 1164 & 117 & 0.163 & 0.363 & 1.00 & 0.496 & 0.020 & 1.00 \\
DRIVE & $275\to 250$ & 419 & 41 & 0.066 & 0.237 & 1.00 & 0.222 & 0.120 & 1.00 \\
DRIVE & $250\to 225$ & 503 & 50 & 0.103 & 0.305 & 1.00 & 0.318 & 0.048 & 1.00 \\
HAP1 & $475\to 450$ & 1021 & 102 & 0.762 & 0.421 & 1.00 & 0.703 & 0.433 & 1.00 \\
\bottomrule
\end{tabular}}
\end{table*} \begin{table*}[!htbp]
\centering
\caption{Additional calls released by GoCo-M but not GoDag in Wainberg split~0 at $\alpha=0.10$, $\delta=0.50$ and $q=0.4$. Shown are the eight co-essential modules contributing the most additional calls; $\ast$ marks terms supported by the frozen TruePath reference.}
\label{supp-tab:casestudy}
\scriptsize
\resizebox{\linewidth}{!}{\begin{tabular}{r p{0.26\textwidth} p{0.16\textwidth} r r p{0.34\textwidth}}
\toprule
Module & Members & Genes admitted & Calls & Correct & Terms transferred \\ \midrule
5071 & \texttt{CHEK2, MDM2, PPM1D, TP53, USP7} & \texttt{MDM2, TP53, USP7} & 18 & 2 & GO:0001046 core promoter sequence-specific DNA binding; GO:0001221 transcription cofactor binding; GO:0002309 T cell proliferation involved in immune response; GO:0010216 maintenance of DNA methylation; GO:0070245 positive regulation of thymocyte apoptotic process; GO:0071494 cellular response to UV-C $\ast$; GO:1901525 negative regulation of mitophagy; \dots \\
\addlinespace
792 & \texttt{CDS2, CHMP3, CHMP4B, INPP5A, VRK1} & \texttt{CDS2, CHMP3, INPP5A} & 11 & 1 & GO:0000815 ESCRT III complex; GO:0090148 membrane fission; GO:0090611 ubiquitin-independent protein catabolic process via the multivesicular body sorting pathway; GO:1902188 positive regulation of viral release from host cell; GO:0050792 regulation of viral process $\ast$ \\
\addlinespace
4493 & \texttt{CDC42, CTNNA1, MLLT4, PARD6B, RAPGEF2, TP53BP2} & \texttt{CTNNA1, TP53BP2} & 11 & 0 & GO:0003161 cardiac conduction system development; GO:0017034 Rap guanyl-nucleotide exchange factor activity; GO:0021884 forebrain neuron development; GO:0060501 positive regulation of epithelial cell proliferation involved in lung morphogenesis; GO:2000481 positive regulation of cAMP-dependent protein kinase activity; GO:0071681 cellular response to indole-3-methanol \\
\addlinespace
3178 & \texttt{DDX21, DNAJC2, HSPA14, NCL, PINX1} & \texttt{DDX21, DNAJC2, NCL} & 11 & 0 & GO:0006450 regulation of translational fidelity; GO:0044547 DNA topoisomerase binding; GO:1902570 protein localization to nucleolus; GO:1904751 positive regulation of protein localization to nucleolus \\
\addlinespace
4091 & \texttt{BIRC2, CYLD, N4BP1, TRAF2, TRAF3} & \texttt{BIRC2, N4BP1, TRAF2} & 9 & 4 & GO:0031996 thioesterase binding; GO:0035631 CD40 receptor complex $\ast$; GO:0060544 regulation of necroptotic process $\ast$; GO:0097300 programmed necrotic cell death $\ast$; GO:0034121 regulation of toll-like receptor signaling pathway; GO:1902527 positive regulation of protein monoubiquitination \\
\addlinespace
3710 & \texttt{IL6ST, JAK1, LIFR, SLC39A1, STAT3} & \texttt{IL6ST, JAK1} & 8 & 4 & GO:0034112 positive regulation of homotypic cell-cell adhesion; GO:0060019 radial glial cell differentiation; GO:0070106 interleukin-27-mediated signaling pathway $\ast$; GO:0070757 interleukin-35-mediated signaling pathway $\ast$; GO:0150105 protein localization to cell-cell junction \\
\addlinespace
5036 & \texttt{CDKN1A, TP53, TP53BP1, UBE2K, USP28} & \texttt{CDKN1A, TP53} & 8 & 0 & GO:0001046 core promoter sequence-specific DNA binding; GO:0001221 transcription cofactor binding; GO:0002309 T cell proliferation involved in immune response; GO:0010994 free ubiquitin chain polymerization; GO:0070245 positive regulation of thymocyte apoptotic process; GO:1901525 negative regulation of mitophagy; GO:1902749 regulation of cell cycle G2/M phase transition \\
\addlinespace
4150 & \texttt{DDX21, HSPA14, NCL, NOLC1, PINX1} & \texttt{DDX21, NCL} & 8 & 0 & GO:0001042 RNA polymerase I core binding; GO:0007000 nucleolus organization; GO:0031428 box C/D snoRNP complex; GO:0042306 regulation of protein import into nucleus \\
\bottomrule
\end{tabular}}
\end{table*} \begin{table*}[!htbp]
\centering
\caption{Comparison of pooled call-level FDR and gene-level FDP estimands at $\alpha=0.10$, $\delta=0.50$, over the same 100 splits.}
\label{supp-tab:estimands}
\scriptsize
\begin{tabular}{l l rrrr}
\toprule
Data set & Procedure / target & Calls released & Correct calls & Realised pooled call FDP & Genes annotated \\ \midrule
Wainberg & Conformal selection (pooled-call FDR) & 0.0 & 0.0 & 0.000 & 0.0 \\
 & GoCo-M (gene-level FDP target) & 1955 & 369.1 & 0.811 & 575 \\
\addlinespace
Sanger & Conformal selection (pooled-call FDR) & 5.6 & 4.8 & 0.045 & 5.4 \\
 & GoCo-M (gene-level FDP target) & 2380 & 380.7 & 0.839 & 529 \\
\addlinespace
DRIVE & Conformal selection (pooled-call FDR) & 3.6 & 2.8 & 0.078 & 2.5 \\
 & GoCo-M (gene-level FDP target) & 2296 & 291.7 & 0.872 & 215 \\
\addlinespace
HAP1 & Conformal selection (pooled-call FDR) & 0.2 & 0.1 & 0.006 & 0.1 \\
 & GoCo-M (gene-level FDP target) & 1285 & 137.7 & 0.893 & 475 \\
\bottomrule
\end{tabular}
\end{table*}
 
\begin{table*}[htbp]
\centering
\caption{Source-object structures, observed score granularity and ranking-fold path choice. Rows without a ranking-fold verdict are structural examples.}
\label{supp-tab:structures}
\tiny
\setlength{\tabcolsep}{3pt}
\renewcommand{\arraystretch}{1.15}
\begin{tabular}{@{}p{0.145\textwidth} p{0.135\textwidth} p{0.05\textwidth} p{0.085\textwidth} p{0.105\textwidth} p{0.135\textwidth} p{0.225\textwidth}@{}}
\toprule
Predictor family & Source set $\mathcal R_{ig}$ & Evidence & Score ties & Structure & A cutoff misses & Ranking-fold verdict \\
\midrule
Co-essential modules (Wainberg; Sanger and DRIVE via the same catalogue) & annotated members of the gene's module & shared & exact, up to 3,361 pairs & module-shared & the interior of every tied block & Wainberg: GoCo-M (100\%); Sanger: GoCo-N (96\%); DRIVE: GoCo-M (82\%) \\
Protein complexes (CORUM, BioPlex) & annotated members of the complex & shared & possible exact & module-shared & the interior of exact ties & structural \\
Co-regulation clusters (ProteomeHD) & annotated members of the co-regulated group & shared & possible exact & module-shared & the interior of exact ties & structural \\
Community / partition detection (HAP1 Global-L2 clusters) & annotated members of the community & shared & exact, up to 14,245 pairs & module-shared & the interior of every tied block & HAP1: GoCo-N (89\%) \\
Random-walk neighbourhoods (FunMap) & annotated neighbours, proximity-ranked & gene-specific & none material & neighbourhood & little: a refined grid recovers granularity & FunMap: GoCo-N (92\%) \\
Interaction neighbourhoods (STRING) & annotated interactors, ranked & gene-specific & none material & neighbourhood & little: a refined grid recovers granularity & STRING: GoCo-N (76\%) \\
Network propagation / diffusion & diffusion-weighted neighbours & gene-specific & predictor-dependent & neighbourhood & to be assessed & structural \\
Multi-network integration (GeneMANIA) & (network, neighbour) pairs with weights & gene-specific & predictor-dependent & neighbourhood & to be assessed & structural \\
Sequence-homology hit lists & homologues with alignment scores & gene-specific & predictor-dependent & neighbourhood & to be assessed & structural \\
Embedding $k$-nearest neighbours & nearest neighbours with distances & gene-specific & predictor-dependent & neighbourhood & to be assessed & structural \\
\bottomrule
\end{tabular}
\end{table*}
 \begin{table*}[htbp]
\centering
\caption{Neighbourhood-structure family: gain over the exhaustive matched-risk frontier of the global path, over 100 splits. For each arm and split, the frontier is the best yield among all global score cutoffs, including abstention, whose realised held-out risk does not exceed that arm's realised risk.}
\label{supp-tab:frontier}
\scriptsize
\begin{tabular}{l l l rr rr}
\toprule
Data & $\alpha$ & $\delta$ & \multicolumn{2}{c}{GoCo-M} & \multicolumn{2}{c}{GoCo-N} \\
 & & & Correct calls vs.\ frontier & Splits ahead & Correct calls vs.\ frontier & Splits ahead \\
\midrule
FunMap & 0.05 & 0.10 & +101 & 1.00 & +155 & 0.95 \\
 & 0.05 & 0.50 & +99 & 1.00 & +175 & 0.97 \\
 & 0.10 & 0.10 & +196 & 1.00 & +379 & 1.00 \\
 & 0.10 & 0.50 & +190 & 1.00 & +393 & 1.00 \\
 & 0.20 & 0.10 & +164 & 1.00 & +519 & 1.00 \\
 & 0.20 & 0.50 & +186 & 1.00 & +535 & 1.00 \\
\midrule
STRING & 0.05 & 0.10 & +46 & 1.00 & +172 & 1.00 \\
 & 0.05 & 0.50 & +40 & 1.00 & +186 & 1.00 \\
 & 0.10 & 0.10 & +138 & 1.00 & +268 & 0.97 \\
 & 0.10 & 0.50 & +135 & 1.00 & +282 & 0.98 \\
 & 0.20 & 0.10 & +197 & 1.00 & +561 & 0.99 \\
 & 0.20 & 0.50 & +205 & 1.00 & +599 & 0.99 \\
\bottomrule
\end{tabular}
\end{table*}
 \begin{table*}[htbp]
\centering
\caption{Ranking-fold selection between GoCo-M and GoCo-N over 100 splits at $\delta=0.10$. In each split, the rule selects the instantiation with larger ranking-fold supported yield at the last candidate whose ranking-fold mean loss does not exceed $\alpha$.}
\label{supp-tab:selection}
\scriptsize
\begin{tabular}{l l l rr r rrr r r}
\toprule
Predictor & Evidence & $\alpha$ & \multicolumn{2}{c}{$\mathcal T$ yield} & Rule chooses & \multicolumn{3}{c}{Held-out supported} & $\Pr(\text{N}>\text{M})$ & Rule agrees \\
 & & & GoCo-M & GoCo-N & GoCo-N & GoCo-M & GoCo-N & selected & held-out & with winner \\
\midrule
Wainberg & module-shared & 0.05 & 68.5 & 55.7 & 0.16 & 213.2 & 168.5 & 206.1 & 0.07 & 0.81 \\
 &  & 0.10 & 117.7 & 73.3 & 0.00 & 365.2 & 222.8 & 365.2 & 0.00 & 1.00 \\
 &  & 0.20 & 176.7 & 131.2 & 0.00 & 539.2 & 392.8 & 539.2 & 0.00 & 1.00 \\
\addlinespace[2pt]
Sanger & module-shared & 0.05 & 77.5 & 114.3 & 0.85 & 213.7 & 351.5 & 331.3 & 0.99 & 0.84 \\
 &  & 0.10 & 127.7 & 185.3 & 0.96 & 340.0 & 559.2 & 553.4 & 1.00 & 0.96 \\
 &  & 0.20 & 478.9 & 495.1 & 0.61 & 1346.8 & 1430.8 & 1411.7 & 0.78 & 0.65 \\
\addlinespace[2pt]
DRIVE & module-shared & 0.05 & 36.3 & 37.8 & 0.52 & 94.2 & 97.4 & 96.8 & 0.52 & 0.54 \\
 &  & 0.10 & 98.0 & 72.0 & 0.18 & 263.5 & 194.5 & 251.5 & 0.02 & 0.80 \\
 &  & 0.20 & 231.1 & 225.9 & 0.00 & 702.0 & 701.5 & 702.0 & 0.00 & 1.00 \\
\addlinespace[2pt]
HAP1 & module-shared & 0.05 & 32.9 & 49.3 & 0.86 & 102.4 & 163.3 & 156.9 & 0.97 & 0.83 \\
 &  & 0.10 & 44.6 & 67.6 & 0.89 & 137.6 & 225.1 & 217.5 & 1.00 & 0.89 \\
 &  & 0.20 & 61.8 & 90.6 & 0.90 & 187.7 & 299.3 & 290.2 & 1.00 & 0.90 \\
\midrule
FunMap & neighbourhood & 0.05 & 319.4 & 349.8 & 0.73 & 954.5 & 1016.1 & 996.7 & 0.72 & 0.61 \\
 &  & 0.10 & 488.3 & 552.9 & 0.92 & 1474.4 & 1615.8 & 1608.1 & 0.89 & 0.83 \\
 &  & 0.20 & 759.2 & 885.5 & 0.96 & 2157.9 & 2615.8 & 2598.2 & 1.00 & 0.96 \\
\addlinespace[2pt]
STRING & neighbourhood & 0.05 & 172.1 & 221.5 & 0.89 & 503.6 & 621.4 & 610.0 & 0.96 & 0.86 \\
 &  & 0.10 & 350.2 & 395.6 & 0.76 & 1021.1 & 1151.1 & 1124.5 & 0.89 & 0.69 \\
 &  & 0.20 & 642.1 & 766.6 & 0.93 & 1732.5 & 2243.0 & 2201.7 & 0.99 & 0.92 \\
\bottomrule
\end{tabular}
\end{table*}
 
\begin{table*}[htbp]
\centering
\caption{Comparison with validation $F_{\max}$ and top-$k$ release rules on the same 100 split streams and held-out folds. The $F_{\max}$ threshold is fitted on the 70\% calibration sample using the CAFA protein-centric $F_1$ measure; GoCo is shown at $\alpha=0.10$, $\delta=0.10$.}
\label{supp-tab:current}
\scriptsize
\setlength{\tabcolsep}{4pt}
\resizebox{\linewidth}{!}{\begin{tabular}{l l r r r r r r r}
\toprule
Predictor & Release rule & Threshold & Validation $F$ & Gene FDP & $\Pr(\widehat R_{\mathcal E}>0.10)$ & Genes with a call & Calls & Correct calls \\
\midrule
Wainberg & Validation $F_{\max}$ threshold & 602 & 0.002 & 0.231 & 1.00 & 1232 & 4941 & 559 \\
 & Top-1 GO term per gene & -- & -- & 0.223 & 1.00 & 1232 & 1232 & 175 \\
 & Top-3 GO terms per gene & -- & -- & 0.228 & 1.00 & 1232 & 2759 & 370 \\
& \textbf{GoCo, certified at $\alpha=0.10$} & nested policy & -- & \textbf{0.095} & \textbf{0.13} & 568 & 1896 & \textbf{365} \\
\midrule
Sanger & Validation $F_{\max}$ threshold & 100 & 0.006 & 0.224 & 1.00 & 1094 & 21840 & 1685 \\
 & Top-1 GO term per gene & -- & -- & 0.211 & 1.00 & 1094 & 1094 & 150 \\
 & Top-3 GO terms per gene & -- & -- & 0.214 & 1.00 & 1094 & 3119 & 401 \\
& \textbf{GoCo, certified at $\alpha=0.10$} & nested policy & -- & \textbf{0.093} & \textbf{0.15} & 505 & 1983 & \textbf{340} \\
\midrule
DRIVE & Validation $F_{\max}$ threshold & 100 & 0.004 & 0.184 & 1.00 & 378 & 8370 & 702 \\
 & Top-1 GO term per gene & -- & -- & 0.173 & 1.00 & 378 & 378 & 51 \\
 & Top-3 GO terms per gene & -- & -- & 0.175 & 1.00 & 378 & 1044 & 141 \\
& \textbf{GoCo, certified at $\alpha=0.10$} & nested policy & -- & \textbf{0.093} & \textbf{0.16} & 208 & 1987 & \textbf{263} \\
\midrule
HAP1 & Validation $F_{\max}$ threshold & 100 & 0.004 & 0.986 & 1.00 & 4506 & 102090 & 1341 \\
 & Top-1 GO term per gene & -- & -- & 0.979 & 1.00 & 4506 & 4506 & 94 \\
 & Top-3 GO terms per gene & -- & -- & 0.982 & 1.00 & 4506 & 13508 & 240 \\
& \textbf{GoCo, certified at $\alpha=0.10$} & nested policy & -- & \textbf{0.091} & \textbf{0.06} & 462 & 1252 & \textbf{138} \\
\midrule
FunMap & Validation $F_{\max}$ threshold & 1 & 0.069 & 0.787 & 1.00 & 2912 & 80278 & 11828 \\
 & Top-1 GO term per gene & -- & -- & 0.599 & 1.00 & 2912 & 2912 & 1057 \\
 & Top-3 GO terms per gene & -- & -- & 0.642 & 1.00 & 2912 & 8526 & 2746 \\
& \textbf{GoCo, certified at $\alpha=0.10$} & nested policy & -- & \textbf{0.088} & \textbf{0.05} & 1084 & 2166 & \textbf{1616} \\
\midrule
STRING & Validation $F_{\max}$ threshold & 1 & 0.098 & 0.828 & 1.00 & 2375 & 156273 & 19717 \\
 & Top-1 GO term per gene & -- & -- & 0.591 & 1.00 & 2375 & 2375 & 946 \\
 & Top-3 GO terms per gene & -- & -- & 0.635 & 1.00 & 2375 & 7118 & 2517 \\
& \textbf{GoCo, certified at $\alpha=0.10$} & nested policy & -- & \textbf{0.089} & \textbf{0.06} & 806 & 1622 & \textbf{1151} \\
\bottomrule
\end{tabular}}
\end{table*}
 
\begin{figure*}[!t]
\centering
\includegraphics[width=0.99\textwidth]{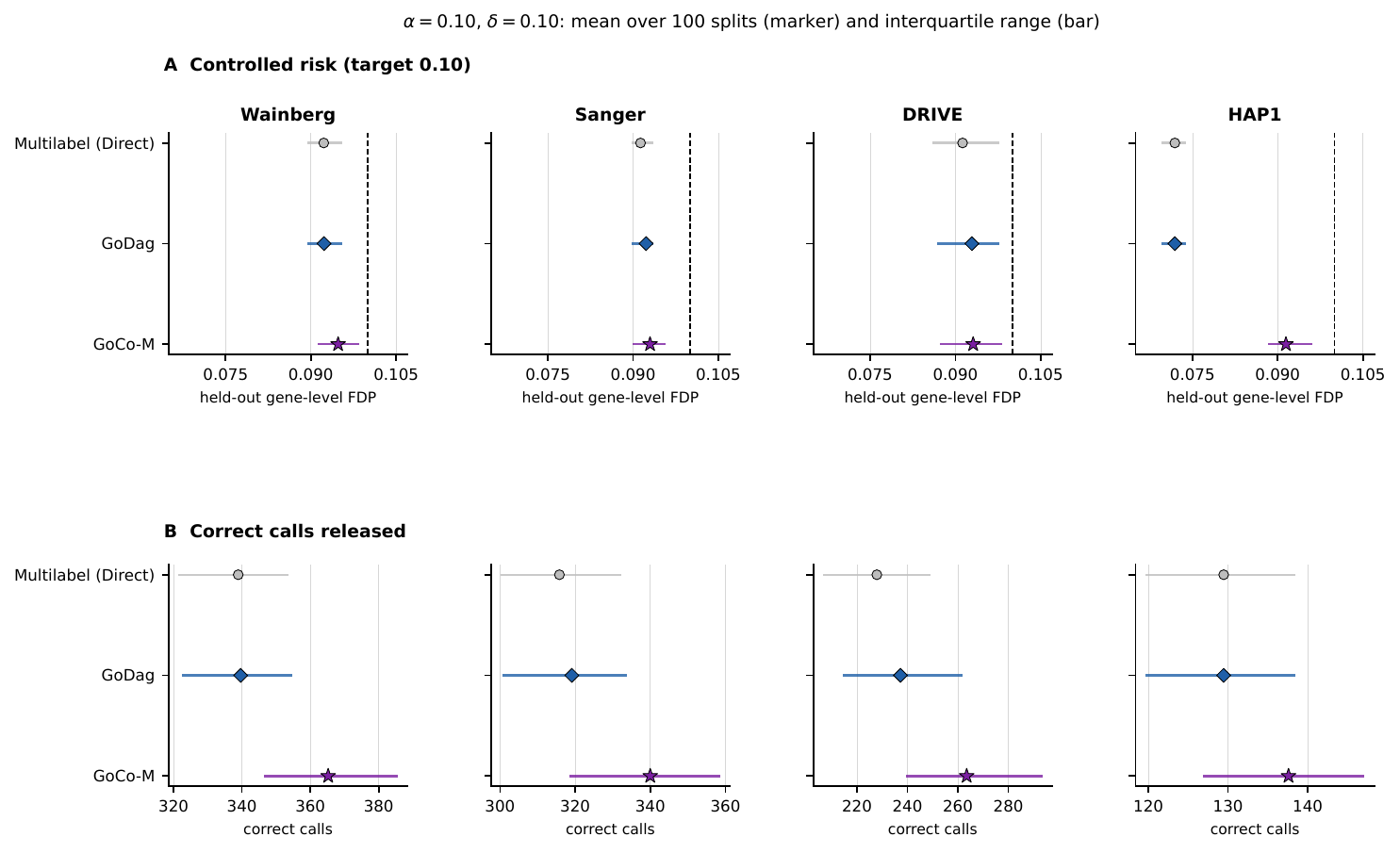}
\caption{Primary comparison of Multilabel, GoDag and GoCo-M at $\alpha=0.10$, $\delta=0.10$: held-out gene-level TruePath FDP and reference-supported calls over 100 splits.}\label{supp-fig:supp-primary}
\end{figure*}

\clearpage
\begin{table*}[!t]
\centering
\caption{Prespecified constants of GoCo. Module-family choices were fixed on Wainberg and the neighbourhood-family instantiation on FunMap; none uses certification- or evaluation-fold labels, and none affects validity (main-text Proposition~\ref{prop:validity}).}
\label{supp-tab:constants}
\scriptsize
\begin{tabular}{p{0.24\textwidth} p{0.70\textwidth}}
\toprule
Constant & Value \\
\midrule
Split & 10\% $\calT$ / 60\% $\calC$ / 30\% $\calE$ by seeded permutation ($20260905+s$ in the module family and $20260910+s$ in the neighbourhood family, $s=0,\ldots,99$); the instantiations were fixed on streams $0$--$9$ of Wainberg and FunMap, respectively; $n_{\calT}=\operatorname{round}(0.10n)$ and $n_{\calT}+m=\operatorname{round}(0.70n)$ with round-half-to-even; module-family global methods calibrate on $\calT\cup\calC$ (70\%), whereas all neighbourhood-family methods certify on $\calC$ (60\%) \\
GoCo-M global grid $\lambda_1>\cdots>\lambda_J$ & 3000, 2500, 2000, 1600, 1400, 1200, 1100, 1050, then 1000 downward in steps of 25 (to 600 in Wainberg, $J=25$; to 100 in Sanger, DRIVE and HAP1, $J=45$); identical for GoDag and GoCo-M in the module family \\
GoCo-M block insertion trigger & ranking-fold mean loss at $\lambda_{j+1}$ exceeds $0.8\alpha$ (i.e.\ 0.08 at $\alpha=0.10$) and $|\calT_j|\ge5$ \\
GoCo-M admission grid $\mathcal Q$ & 0.01, 0.02, 0.03, 0.04, 0.05, 0.075, 0.10, 0.125, 0.15, 0.175, 0.20, 0.225, 0.25, 0.275, 0.30, 0.35, 0.40, 0.50, 0.60, 0.75 ($q=1$ is the next global threshold) \\
Shrinkage weight & $\kappa=1$ in Equation~\eqref{supp-eq:sourcerisk}; two pseudo-calls for the per-call rate features \\
Ranking statistic & GoCo-M: Equation~\eqref{supp-eq:utility}; GoCo-N: individual calls ordered by $\widehat p_{ig}$, with smaller predicted unsupported probability entering first; both use the cross-fitted logistic per-call model ($C=0.5$, standardised features, two hash-defined halves of $\calT$) \\
Instantiation & GoCo-M (partial admission within grid steps) for evidence shared between genes, fixed on Wainberg; GoCo-N (nested top-$\ell$ calls by $\widehat p_{ig}$) for per-gene neighbourhood evidence, fixed on FunMap \\
Tie-break & GoCo-M: deterministic lexicographic ordering by $(u_i,h_i,\text{gene identifier})$, where $h_i$ is the SHA-256-derived hash of the gene identifier; GoCo-N: ties in $\widehat p_{ig}$ are broken by the frozen gene hash and then by call index; exact comparisons, no numerical tolerance \\
Certification & $\alpha=0.10$ (sensitivity 0.05, 0.20); $\delta=0.10$ and $\delta=0.50$; $\shat_k$ floored at $10^{-6}$ \\
GO reference & Direct annotations and TruePath closure frozen at one release for all four dependency data sets (GO release 2019-11-01: \texttt{goa\_human.gaf} and \texttt{go-basic.obo} from the GO archive; all evidence codes; NOT-qualified annotations and root terms removed; closure over \emph{is\_a} and \emph{part\_of}); GO archive 2024-09-08 for FunMap and STRING \\
\bottomrule
\end{tabular}
\end{table*}
\begin{table*}[!t]
\centering
\caption{Methods compared. All are post-prediction release rules applied to the same frozen scores and split streams; columns distinguish the certified loss, candidate path, calibration sample and use of GO hierarchy and predictor-native evidence.}
\label{supp-tab:methods}
\scriptsize
\setlength{\tabcolsep}{4pt}
\resizebox{\linewidth}{!}{%
\begin{tabular}{>{\raggedright\arraybackslash}p{0.15\linewidth}>{\raggedright\arraybackslash}p{0.12\linewidth}>{\raggedright\arraybackslash}p{0.31\linewidth}>{\raggedright\arraybackslash}p{0.15\linewidth}>{\raggedright\arraybackslash}p{0.07\linewidth}>{\raggedright\arraybackslash}p{0.13\linewidth}}
\toprule
Method & Loss certified & Candidate path & Calibration sample & GO hierarchy & Predictor-native evidence \\
\midrule
$F_{\max}$ threshold & none & global cutoff maximising $F_1$ & (threshold fitted on 70\%) & no & no \\
Top-$k$ per gene & none & the $k$ highest-scoring terms of each gene & none & no & no \\
Multilabel [24] & Direct FDP & global grid $\lambda_1>\cdots>\lambda_J$ & 70\% (module family); $\calC$ (60\%) in neighbourhood family & no & no \\
Generic calibrators (nine rules) & Direct FDP & global grid & 70\% (module family); $\calC$ (60\%) in neighbourhood family & no & no \\
GoDag & TruePath FDP, Eq.~\eqref{eq:fdp} & global grid & 70\% (module family); $\calC$ (60\%) in neighbourhood family & yes & no \\
GoCo-M & TruePath FDP & grid, plus partial admission of affected genes within each step, ordered by $u_i$ (Eq.~\eqref{supp-eq:utility}) & $\calC$ (60\%); ordering fitted on $\calT$ (10\%) & yes & yes, at gene resolution within a step \\
GoCo-N & TruePath FDP & top-$\ell$ calls ordered by $\widehat p_{ig}$ & $\calC$; model fitted on $\calT$ & yes & yes, at call resolution \\
\bottomrule
\end{tabular}}
\end{table*}
\begin{table*}[!t]
\centering
\caption{Error-target sweep at $\delta=0.50$: mean held-out gene-level FDP and supported GO-term yield at $\alpha\in\{0.05,0.10,0.20\}$ for every predictor, 100 splits.}
\label{supp-tab:sweep50}
\scriptsize
\begin{tabular}{ll rr rr rr}
\toprule
Data set & Method & \multicolumn{2}{c}{$\alpha=0.05$} & \multicolumn{2}{c}{$\alpha=0.10$} & \multicolumn{2}{c}{$\alpha=0.20$} \\
 & & FDP & Correct & FDP & Correct & FDP & Correct \\
 & & & calls & & calls & & calls \\ \midrule
\multicolumn{8}{l}{\emph{Module structure: GoCo-M (development: Wainberg)}} \\
Wainberg & Multilabel (Direct loss) & 0.0359 & 183.1 & 0.0923 & 339.6 & 0.1699 & 518.9 \\
 & GoDag & 0.0359 & 183.1 & 0.0923 & 339.6 & 0.1699 & 518.9 \\
 & \textbf{GoCo-M} & 0.0466 & 214.8 & 0.0964 & 369.1 & 0.1939 & 539.6 \\
\addlinespace[2pt]
Sanger & Multilabel (Direct loss) & 0.0474 & 212.1 & 0.0953 & 333.0 & 0.1895 & 1163.2 \\
 & GoDag & 0.0484 & 218.2 & 0.0980 & 353.8 & 0.1895 & 1163.2 \\
 & \textbf{GoCo-M} & 0.0491 & 234.5 & 0.0980 & 380.7 & 0.1913 & 1360.9 \\
\addlinespace[2pt]
DRIVE & Multilabel (Direct loss) & 0.0444 & 92.3 & 0.0945 & 248.1 & 0.1844 & 702.4 \\
 & GoDag & 0.0466 & 94.4 & 0.0965 & 258.6 & 0.1844 & 702.4 \\
 & \textbf{GoCo-M} & 0.0482 & 99.4 & 0.0971 & 291.7 & 0.1844 & 702.4 \\
\addlinespace[2pt]
HAP1 & Multilabel (Direct loss) & 0.0436 & 96.1 & 0.0719 & 129.4 & 0.1353 & 160.2 \\
 & GoDag & 0.0436 & 96.1 & 0.0719 & 129.4 & 0.1353 & 160.2 \\
 & \textbf{GoCo-M} & 0.0459 & 102.7 & 0.0943 & 137.7 & 0.1872 & 188.3 \\
\midrule
\multicolumn{8}{l}{\emph{Neighbourhood structure: GoCo-N (development: FunMap)}} \\
FunMap & Multilabel (Direct loss, shared grid) & 0.0312 & 727.9 & 0.0741 & 1129.7 & 0.1543 & 1869.5 \\
 & GoDag, shared grid & 0.0312 & 727.9 & 0.0901 & 1261.2 & 0.1722 & 1995.4 \\
 & GoDag, refined grid & 0.0486 & 911.0 & 0.0986 & 1324.0 & 0.1968 & 2189.4 \\
 & \textbf{GoCo-N} & 0.0462 & 1065.8 & 0.0921 & 1665.8 & 0.1881 & 2671.8 \\
\addlinespace[2pt]
STRING & Multilabel (Direct loss, shared grid) & 0.0386 & 424.9 & 0.0828 & 828.1 & 0.1458 & 1383.8 \\
 & GoDag, shared grid & 0.0386 & 424.9 & 0.0828 & 828.1 & 0.1693 & 1562.5 \\
 & GoDag, refined grid & 0.0487 & 496.1 & 0.0978 & 962.1 & 0.1967 & 1768.4 \\
 & \textbf{GoCo-N} & 0.0466 & 663.9 & 0.0944 & 1208.0 & 0.1905 & 2332.3 \\
\bottomrule
\end{tabular}
\par\vspace{3pt}\begin{minipage}{\linewidth}\footnotesize FDP: mean held-out gene-level TruePath false discovery proportion. Correct calls: released gene--GO pairs the frozen reference confirms. Each structure is shown with the instantiation prespecified for it. Module-family GoCo-M and GoDag share the 25-unit grid. In the neighbourhood family, Multilabel and GoDag are shown on the shared coarse grid, with refined-grid GoDag separately; GoCo-N uses a distinct call-ordered path. Grid refinement does not change the fixed-sequence multiplicity argument.\end{minipage}
\label{supp-supp:last-page}
\end{table*}

\begin{thebibliography}{99}
\bibitem{radivojac2013} Radivojac P, Clark WT, Oron TR, et al. A large-scale evaluation of computational protein function prediction. \textit{Nat Methods}. 2013;10:221--227. doi:10.1038/nmeth.2340.
\bibitem{zhou2019} Zhou N, Jiang Y, Bergquist TR, et al. The CAFA challenge reports improved protein function prediction and new functional annotations for hundreds of genes through experimental screens. \textit{Genome Biol}. 2019;20:244. doi:10.1186/s13059-019-1835-8.
\bibitem{jiang2016cafa2} Jiang Y, Oron TR, Clark WT, et al. An expanded evaluation of protein function prediction methods shows an improvement in accuracy. \textit{Genome Biol}. 2016;17:184. doi:10.1186/s13059-016-1037-6.
\bibitem{lin2024review} Lin B, Luo X, Liu Y, et al. A comprehensive review and comparison of existing computational methods for protein function prediction. \textit{Brief Bioinform}. 2024;25(4):bbae289. doi:10.1093/bib/bbae289.
\bibitem{go2023} The Gene Ontology Consortium. The Gene Ontology knowledgebase in 2023. \textit{Genetics}. 2023;224(1):iyad031. doi:10.1093/genetics/iyad031.
\bibitem{skunca2012quality} \v{S}kunca N, Altenhoff A, Dessimoz C. Quality of computationally inferred Gene Ontology annotations. \textit{PLoS Comput Biol}. 2012;8(5):e1002533. doi:10.1371/journal.pcbi.1002533.
\bibitem{clark2013information} Clark WT, Radivojac P. Information-theoretic evaluation of predicted ontological annotations. \textit{Bioinformatics}. 2013;29(13):i53--i61. doi:10.1093/bioinformatics/btt228.
\bibitem{piovesan2024evaluator} Piovesan D, Zago D, Joshi P, et al. CAFA-evaluator: a Python tool for benchmarking ontological classification methods. \textit{Bioinform Adv}. 2024;4(1):vbae043. doi:10.1093/bioadv/vbae043.
\bibitem{kulmanov2020deepgoplus} Kulmanov M, Hoehndorf R. DeepGOPlus: improved protein function prediction from sequence. \textit{Bioinformatics}. 2020;36(2):422--429. doi:10.1093/bioinformatics/btz595.
\bibitem{yao2021netgo2} Yao S, You R, Wang S, et al. NetGO 2.0: improving large-scale protein function prediction with massive sequence, text, domain, family and network information. \textit{Nucleic Acids Res}. 2021;49(W1):W469--W475. doi:10.1093/nar/gkab398.
\bibitem{tawfiq2026completeness} Tawfiq R, Kulmanov M, Hoehndorf R. Evaluating completeness, coherence, and consistency of genome-scale function annotations. \textit{Brief Bioinform}. 2026;27(3):bbag336. doi:10.1093/bib/bbag336.
\bibitem{zhang2024search} Zhang C, Freddolino PL. A large-scale assessment of sequence database search tools for homology-based protein function prediction. \textit{Brief Bioinform}. 2024;25(4):bbae349. doi:10.1093/bib/bbae349.
\bibitem{zhu2025mkfgo} Zhu YH, Zhu S, Yu X, et al. MKFGO: integrating multi-source knowledge fusion with pretrained language model for high-accuracy protein function prediction. \textit{Brief Bioinform}. 2025;26(4):bbaf420. doi:10.1093/bib/bbaf420.
\bibitem{hu2025dahgl} Hu S, Zhang W, Zhao B. DA-HGL: a domain-augmented heterogeneous graph learning framework for protein function prediction. \textit{Brief Bioinform}. 2025;26(5):bbaf511. doi:10.1093/bib/bbaf511.
\bibitem{kan2026protdml} Kan Y, Yi G. ProtDML: label-aware representation learning for broad-spectrum protein function prediction. \textit{Brief Bioinform}. 2026;27(3):bbag293. doi:10.1093/bib/bbag293.
\bibitem{wang2026ssas} Wang D, Wang H, Jiang T, et al. SSAS-GO: structure-sequence adaptive synergy network for protein function prediction. \textit{Brief Bioinform}. 2026;27(4):bbag408. doi:10.1093/bib/bbag408.
\bibitem{mostafavi2008} Mostafavi S, Ray D, Warde-Farley D, et al. GeneMANIA: a real-time multiple association network integration algorithm for predicting gene function. \textit{Genome Biol}. 2008;9(Suppl 1):S4. doi:10.1186/gb-2008-9-s1-s4.
\bibitem{kustatscher2019} Kustatscher G, Grabowski P, Schrader TA, et al. Co-regulation map of the human proteome enables identification of protein functions. \textit{Nat Biotechnol}. 2019;37:1361--1371. doi:10.1038/s41587-019-0298-5.
\bibitem{wainberg2021} Wainberg M, Kamber RA, Balsubramani A, et al. A genome-wide atlas of co-essential modules assigns function to uncharacterized genes. \textit{Nat Genet}. 2021;53:638--649. doi:10.1038/s41588-021-00840-z.
\bibitem{angelopoulos2023gentle} Angelopoulos AN, Bates S. Conformal prediction: a gentle introduction. \textit{Found Trends Mach Learn}. 2023;16(4):494--591. doi:10.1561/2200000101.
\bibitem{bates2021} Bates S, Angelopoulos AN, Lei L, et al. Distribution-free, risk-controlling prediction sets. \textit{J ACM}. 2021;68(6):43. doi:10.1145/3478535.
\bibitem{angelopoulos2024crc} Angelopoulos AN, Bates S, Fisch A, et al. Conformal risk control. In: \textit{Proceedings of the Twelfth International Conference on Learning Representations (ICLR 2024)}. 2024. \url{https://openreview.net/forum?id=33XGfHLtZg}.
\bibitem{angelopoulos2025ltt} Angelopoulos AN, Bates S, Cand\`es EJ, et al. Learn Then Test: calibrating predictive algorithms to achieve risk control. \textit{Ann Appl Stat}. 2025;19(2):1641--1662. doi:10.1214/24-AOAS1998.
\bibitem{boger2025} Boger RS, Chithrananda S, Angelopoulos AN, et al. Functional protein mining with conformal guarantees. \textit{Nat Commun}. 2025;16:85. doi:10.1038/s41467-024-55676-y.
\bibitem{ding2024cpec} Ding K, Luo J, Luo Y. Leveraging conformal prediction to annotate enzyme function space with limited false positives. \textit{PLoS Comput Biol}. 2024;20(5):e1012135. doi:10.1371/journal.pcbi.1012135.
\bibitem{lopezdecastro2025conformal} L\'opez-De-Castro M, Garc\'ia-Galindo A, Gonz\'alez-Gomariz J, et al. Conformal inference for reliable single cell RNA-seq annotation. \textit{Bioinformatics}. 2025;41(10):btaf521. doi:10.1093/bioinformatics/btaf521.
\bibitem{corbetta2026cellontology} Corbetta D, Finos L, Geistlinger L, et al. Conformal inference for cell-type annotation with graph-structured constraints. \textit{J R Stat Soc Ser C Appl Stat}. 2026;qlag037. doi:10.1093/jrsssc/qlag037.
\bibitem{mortier2026hierarchical} Mortier T, Javanmardi A, Sale Y, et al. Conformal prediction in hierarchical classification with constrained representation complexity. In: \textit{Proceedings of the 29th International Conference on Artificial Intelligence and Statistics (AISTATS 2026)}. PMLR; 2026. arXiv:2501.19038.
\bibitem{zhang2025structured} Zhang B, Li S, Bastani O. Conformal structured prediction. In: \textit{Proceedings of the Thirteenth International Conference on Learning Representations (ICLR 2025)}. 2025. \url{https://openreview.net/forum?id=mKfmLQXP6J}.
\bibitem{shao2007} Shao Y, Tseng CH. Sample size calculation with dependence adjustment for FDR-control in microarray studies. \textit{Stat Med}. 2007;26(23):4219--4237. doi:10.1002/sim.2862.
\bibitem{shang2012fdp} Shang S, Liu M, Shao Y. A tight prediction interval for false discovery proportion under dependence. \textit{Open J Stat}. 2012;2(2):163--171. doi:10.4236/ojs.2012.22018.
\bibitem{hajek1960} H\'ajek J. Limiting distributions in simple random sampling from a finite population. \textit{Publ Math Inst Hungar Acad Sci}. 1960;5:361--374.
\bibitem{erdos1959} Erd\H{o}s P, R\'enyi A. On the central limit theorem for samples from a finite population. \textit{Publ Math Inst Hungar Acad Sci}. 1959;4:49--61.
\bibitem{hoeffding1963} Hoeffding W. Probability inequalities for sums of bounded random variables. \textit{J Am Stat Assoc}. 1963;58:13--30. doi:10.1080/01621459.1963.10500830.
\bibitem{behan2019} Behan FM, Iorio F, Picco G, et al. Prioritization of cancer therapeutic targets using CRISPR-Cas9 screens. \textit{Nature}. 2019;568:511--516. doi:10.1038/s41586-019-1103-9.
\bibitem{mcdonald2017} McDonald ER III, de Weck A, Schlabach MR, et al. Project DRIVE: a compendium of cancer dependencies and synthetic lethal relationships uncovered by large-scale, deep RNAi screening. \textit{Cell}. 2017;170:577--592.e10. doi:10.1016/j.cell.2017.07.005.
\bibitem{billmann2026} Billmann M, Costanzo M, Zhang X, et al. Global genetic interaction network of a human cell maps conserved principles and informs functional interpretation of gene co-essentiality profiles. \textit{Cell}. 2026;189(12):3736--3759.e58. doi:10.1016/j.cell.2026.03.044.
\bibitem{shi2025funmap} Shi Z, Lei JT, Elizarraras JM, et al. Mapping the functional network of human cancer through machine learning and pan-cancer proteogenomics. \textit{Nat Cancer}. 2025;6(1):205--222. doi:10.1038/s43018-024-00869-z.
\bibitem{szklarczyk2023string} Szklarczyk D, Kirsch R, Koutrouli M, et al. The STRING database in 2023: protein--protein association networks and functional enrichment analyses for any sequenced genome of interest. \textit{Nucleic Acids Res}. 2023;51(D1):D638--D646. doi:10.1093/nar/gkac1000.
\bibitem{bentkus2004} Bentkus V. On Hoeffding's inequalities. \textit{Ann Probab}. 2004;32(2):1650--1673. doi:10.1214/009117904000000360.
\bibitem{maurer2009bernstein} Maurer A, Pontil M. Empirical Bernstein bounds and sample-variance penalization. In: \textit{Proceedings of the 22nd Conference on Learning Theory (COLT 2009)}. 2009. \url{https://arxiv.org/abs/0907.3740}.
\bibitem{mcdiarmid1989} McDiarmid C. On the method of bounded differences. In: \textit{Surveys in Combinatorics, 1989}. Cambridge: Cambridge University Press; 1989:148--188. doi:10.1017/CBO9781107359949.008.
\bibitem{janson2004} Janson S. Large deviations for sums of partly dependent random variables. \textit{Random Struct Algorithms}. 2004;24(3):234--248. doi:10.1002/rsa.20008.
\bibitem{newey1987} Newey WK, West KD. A simple, positive semi-definite, heteroskedasticity and autocorrelation consistent covariance matrix. \textit{Econometrica}. 1987;55(3):703--708. doi:10.2307/1913610.
\bibitem{platt1999} Platt J. Probabilistic outputs for support vector machines and comparisons to regularized likelihood methods. In: Smola AJ, Bartlett P, Sch\"olkopf B, Schuurmans D, eds. \textit{Advances in Large Margin Classifiers}. Cambridge, MA: MIT Press; 1999:61--74.
\bibitem{zadrozny2002} Zadrozny B, Elkan C. Transforming classifier scores into accurate multiclass probability estimates. In: \textit{Proceedings of the Eighth ACM SIGKDD International Conference on Knowledge Discovery and Data Mining (KDD 2002)}. 2002:694--699. doi:10.1145/775047.775151.
\bibitem{shi2021n4bp1} Shi H, Sun L, Wang Y, et al. N4BP1 negatively regulates NF-$\kappa$B by binding and inhibiting NEMO oligomerization. \textit{Nat Commun}. 2021;12:1379. doi:10.1038/s41467-021-21711-5.
\bibitem{jin2023selection} Jin Y, Cand\`es EJ. Selection by prediction with conformal p-values. \textit{J Mach Learn Res}. 2023;24(244):1--41.
\bibitem{bates2023outliers} Bates S, Cand\`es E, Lei L, et al. Testing for outliers with conformal p-values. \textit{Ann Stat}. 2023;51(1):149--178. doi:10.1214/22-AOS2244.
\bibitem{vovk2013conditional} Vovk V. Conditional validity of inductive conformal predictors. \textit{Mach Learn}. 2013;92(2--3):349--376. doi:10.1007/s10994-013-5355-6.
\bibitem{gibbs2025conditional} Gibbs I, Cherian JJ, Cand\`es EJ. Conformal prediction with conditional guarantees. \textit{J R Stat Soc Ser B Stat Methodol}. 2025;87(4):1100--1126. doi:10.1093/jrsssb/qkaf008.
\bibitem{ignatiadis2016ihw} Ignatiadis N, Klaus B, Zaugg JB, et al. Data-driven hypothesis weighting increases detection power in genome-scale multiple testing. \textit{Nat Methods}. 2016;13(7):577--580. doi:10.1038/nmeth.3885.
\bibitem{lei2018adapt} Lei L, Fithian W. AdaPT: an interactive procedure for multiple testing with side information. \textit{J R Stat Soc Ser B Stat Methodol}. 2018;80(4):649--679. doi:10.1111/rssb.12274.
\bibitem{vovk2005book} Vovk V, Gammerman A, Shafer G. \textit{Algorithmic Learning in a Random World}. New York: Springer; 2005.
\bibitem{warwick2020negative} Warwick Vesztrocy A, Dessimoz C. Benchmarking gene ontology function predictions using negative annotations. \textit{Bioinformatics}. 2020;36(Suppl 1):i210--i218. doi:10.1093/bioinformatics/btaa466.
\bibitem{jiang2014incomplete} Jiang Y, Clark WT, Friedberg I, et al. The impact of incomplete knowledge on the evaluation of protein function prediction: a structured-output learning perspective. \textit{Bioinformatics}. 2014;30(17):i609--i616. doi:10.1093/bioinformatics/btu472.
\bibitem{zhapa2024pugo} Zhapa-Camacho F, Tang Z, Kulmanov M, et al. Predicting protein functions using positive-unlabeled ranking with ontology-based priors. \textit{Bioinformatics}. 2024;40(Suppl 1):i401--i409. doi:10.1093/bioinformatics/btae237.
\bibitem{littmann2021embedding} Littmann M, Heinzinger M, Dallago C, et al. Embeddings from deep learning transfer GO annotations beyond homology. \textit{Sci Rep}. 2021;11:1160. doi:10.1038/s41598-020-80786-0.
\bibitem{wang2026gocozenodo} [dataset] Wang C, Shao Y, Liu P. GoCo: structure-guided conformal calibration for error-controlled Gene Ontology annotation release --- reference implementation, analysis inputs and split-level outputs. \textit{Zenodo}. 2026. https://doi.org/10.5281/zenodo.22677009.
\end{thebibliography}
\end{document}